\documentclass[11pt]{article}
\usepackage[utf8]{inputenc}

\usepackage{textcomp}
\usepackage[russian, english]{babel}

\usepackage{amsmath, amsfonts, amssymb,amsthm}
\usepackage{verbatim} 

\usepackage{float}
\usepackage{changepage}
\usepackage{array}
\usepackage{enumerate}
\usepackage{hyperref}
\usepackage{authblk}

\usepackage{xcolor}

\usepackage{tikz}
\usetikzlibrary{calc, matrix, arrows,decorations.pathmorphing, positioning}
\numberwithin{equation}{section}
\newtheorem{theorem}{Theorem}[section]
\newtheorem{lemma}{Lemma}[section]
\newtheorem{prop}{Proposition}[section]
\newtheorem{definition}{Definition}[section]
\newtheorem{remark}{Remark}[section]
\newtheorem{example}{Example}[section]

\newcommand{\ri}{\mathrm{i}}
\newcommand{\rd}{\mathrm{d}}

\newcolumntype{M}[1]{>{\centering\arraybackslash}m{#1}}

\title{Hamiltonian reductions in Matrix Painlev\'e systems}
\author[1,2,3]{Mikhail Bershtein \thanks{mbersht@gmail.com}}
\author[2,3,4]{Andrei Grigorev \thanks{andrey4287252@gmail.com}}
\author[2,3]{Anton Shchechkin \thanks{a.shchechkin@skoltech.ru}}
\affil[1]{Landau Institute for Theoretical Physics, Chernogolovka, Russia}
\affil[2]{Skolkovo Institute of Science and Technology, Moscow, Russia}
\affil[3]{National Research University Higher School of Economics, Moscow, Russia}
\affil[4]{Moscow Institute of Physics and Technology, Dolgoprudny, Russia}

\date{}

\begin{document}
\pagenumbering{arabic}

\maketitle

\begin{abstract}
%
	
For certain finite groups \(G\) of B\"acklund transformations we show that the dynamics of \(G\)-invariant configurations of \(n|G|\) Calogero--Painlev\'e particles is equivalent to certain \(n\)-particle  Calogero--Painlev\'e system. We also show that the reduction of dynamics on \(G\)-invariant subset of \(n|G|\times n|G|\) matrix Painlev\'e system is equivalent to certain \(n\times n \)  matrix Painlev\'e system. The groups \(G\) correspond to folding transformations of Painlev\'e equations. The proofs are based on the Hamiltonian reductions.

\end{abstract}

\tableofcontents

\section{Introduction}\label{sec:intro}

In the paper we construct a relation between matrix Painlev\'e equations of different size and Calogero--Painlev\'e systems of different number of particles. Such relations are in correspondence with folding transformations in the Painlev\'e theory. We recall these notions and illustrate our results by simple instructive examples.

\paragraph{Folding transformations.} This notion was introduced in \cite{TOS05}, but actual examples of such transformations were known for many years. By definition, the folding transformation is an algebraic (of degree greater than 1) map between solutions of the Painlev\'e equations. Moreover, this map should go through a quotient of Okamoto-Sakai space of initial conditions (see e.g. \cite{KNY15} for the review).  Probably the simplest example is the folding transformation of Painlev\'e II to itself.

\begin{example}\label{ex:intro_folding}
	The Painlev\'e $\mathrm{II}$ equation is a second order differential equation with parameter $\theta$  
	\begin{equation}
		\frac{d^2q}{dt^2} = 2q^3 + tq + \theta. 
	\end{equation}
    This equation is equivalent to a Hamiltonian system with Hamiltonian
    \begin{equation}\label{Ham_PII_intro}
    	H_{\mathrm{II}}(p,q;t)=\frac{1}{2}p^2 - \frac{1}{2}\left(q^2 + \frac{t}{2}\right)^2-\theta q.
    \end{equation}
    This equation (system) has a natural symmetry \(r\) which transforms parameter \(\theta \mapsto -\theta\) and maps \((p,q) \mapsto (-p,-q)\). Such symmetries are called B\"acklund transformations. For the special value of the parameter \(\theta=0\) the equation is preserved by \(r\) so one can ask for the equation on the functions invariant under this transformation. If we introduce new invariant coordinates \(P,Q\) and new time \(s\)
    \begin{equation}\label{PII_coord}
    	Q=-2^{-1/3}\frac{p}{q}, \quad P=2^{1/3}\left(q^2-\frac{p^2}{2q^2}\right)-\frac{s}2, \qquad s = -2^{1/3}t
    \end{equation}
    we get Painlev\'e $\mathrm{II}$ equation on $Q$ with parameter \(\theta=-1/2\)
	\begin{equation}
		\frac{\rd^2 Q}{\rd s^2} = 2Q^3 + sQ - \frac12.
	\end{equation}
    This is a transformation of degree 2 between the spaces of initial conditions.

\end{example}
Note that in this example we start from the Painlev\'e equation with the special value of the parameter (namely \(\theta=0\)) and come to the Painlev\'e equation with the special value of the parameter (namely $\theta=-1/2$). But, it appears that we can come to the equation with an arbitrary value of parameter if we start from Calogero--Painlev\'e system.

\paragraph{Calogero--Painlev\'e systems.}
These systems can be viewed as an \(N\)-particle generalization of the Painlev\'e equations. Let \(q_1,\dots,q_{N}\) be coordinates of these particles and \(p_1,\dots,p_{N}\) be the corresponding momenta with the standard Poisson bracket. The dynamics of a Calogero--Painlev\'e system is defined by the Hamitonian of the form \cite{Takasaki:2001Painleve}
\begin{equation}\label{Ham_CPII}
	H_{\mathrm{CP}}(\{(p_i,q_i)\};t)=\frac12\sum_{i=1}^N p_i^2+\sum_{i=1}^N V^{(1)}(q_i) +\sum_{1 \le i<j\le N} V^{(2)}(q_i,q_j).
\end{equation}
Here \(V^{(1)}(q)\) is a Painlev\'e potential (i.e. for \(N=1\) we get a Hamiltonian of the Painlev\'e equation) and \(V^{(2)}(q_1,q_2)\) is a Calogero-type interaction. Note that Hamiltonian \eqref{Ham_CPII} is symmetric. By definition the phase space of Calogero--Painlev\'e system is a quotient by the action of permutation group \(S_{N}\).

Note that Hamiltonian \eqref{Ham_CPII} is non-autonomous, since Painlev\'e potential \(V^{(1)}\) depends on time~\(t\). In the autonomous limit, in which \(t\) is a coupling constant, Hamiltonian \eqref{Ham_CPII} belongs to Inozemtsev extension of Calogero integrable system \cite{Inozemtsev:1985Extension}, \cite{Inozemtsev:1989Lax}.  

Another important feature of Calogero--Painlev\'e systems is that they describe isomonodromic deformations of certain natural \(2N\times 2N\) systems \cite{Kawakami:2015Matrix}, \cite{BCR17}.

In the paper we construct a natural analog of the folding transformations for the Calogero--Painlev\'e systems. Let us give an example of Calogero--Painlev\'e II for \(N=2\).

\begin{example} \label{ex:intro_Calogero}
Hamiltonian \eqref{Ham_CPII} in this case has the form 
\begin{equation}\label{Ham_CPII_2}
	H_{\mathrm{CPII}}(\{(p_i,q_i)\};t) =\sum_{i=1}^2\left(\frac{1}{2}p_i^2 - \frac{1}{2}\left(q_i^2 + \frac{t}{2}\right)^2-\theta q_i\right)  + \frac{g^2}{(q_1 - q_2)^2}.
\end{equation}
This system has a symmetry $r_{\mathrm{CP}}: \{(p_1, q_{1}), (p_{2}, q_{2})\} \mapsto \{(-p_1, -q_{1}), (-p_{2}, -q_{2})\}$ with $\theta\mapsto -\theta$. For $\theta=0$ the subset of $r_{\mathrm{CP}}$--invariant points is defined by the equations (recall that we factor over the permutation group $S_2$)
\begin{equation}
q_{1} + q_{2} = 0,\;\; p_1 + p_2 = 0.
\end{equation}
It is easy to check that these equations are preserved by the dynamics. 

For $g = 0$ Calogero--Painlev\'e system is equivalent to two noninteracting Painlev\'e particles \(q_1,q_2\) up to the permutation. Hence by the Example \ref{ex:intro_folding} dynamics on $r_{\mathrm{CP}}$-invariant subset is equivalent to the Painlev\'e~$\mathrm{II}$ equation with \(\theta=-1/2\).	

For \(g \neq 0\) let us take the following coordinates $P,Q$ on $r_{\mathrm{CP}}$-invariant subset and rescale the time
\begin{equation}
	Q=-2^{-1/3}\left(\frac{p_1}{q_1}-\frac{\ri g}{2q_1^2}\right), \quad P=2^{1/3} q_1^2-Q^2-\frac{s}2, \qquad s=-2^{1/3}t.	
\end{equation}
Note that these formulas are $g$-deformed version of formulas \eqref{PII_coord}. It is straightforward to compute the dynamics in terms of \(p_1,q_1\)
\begin{equation}
	\frac{dq_1}{dt}=p_1, \qquad \frac{dp_1}{dt}=2q_1^3+tq_1+\frac{g^2}{4q_1^3},	
\end{equation}
and then get 
\begin{equation}\label{ex_CPII}
	\frac{dQ}{ds}=P,\qquad \frac{dP}{ds}=2Q^3+sQ-\ri g-1/2,	
\end{equation}
which is the Painlev\'e $\mathrm{II}$ equation with parameter \(\theta=-\frac12 -\ri g\).
\end{example}
It appears that this example can be generalized to the Calogero--Painlev\'e system with more than $2$ particles. Namely, for \(2n\)-particle Calogero Painlev\'e~II with \(\theta=0\) the dynamics on (dense open subset of) \(r_{\mathrm{CP}}\)-invariant subset is equivalent to the dynamics of \(n\)-particle Calogero Painlev\'e~II systems with \(\theta=-\frac12 -\ri g\).  Moreover, similar statements hold for other folding transformations, see Theorem~\ref{thm:CP_red} and remarks after it.  This is one of the main results of the paper.

\begin{remark}
	Part of our motivation comes from the papers by Rumanov \cite{R13},\cite{R14}, where certain solutions of \(N\)-particle Calogero--Painlev\'e II systems are related to a spectrum of \(\beta\)-ensemble with \(\beta=2N\). It was also observed in \cite{R14} that this solution for \(N=2\) is described by the scalar Painlev\'e II equation. It appears that this solution is invariant with respect to \(r_{\mathrm{CP}}\), hence this observation is a corollary of Example \ref{ex:intro_Calogero} above. It would be interesting to study subsets corresponding to Rumanov's solutions.
\end{remark}

\paragraph{Matrix Painlev\'e systems.}
Calogero--Painlev\'e systems can be obtained by Hamiltonian reduction \`a la Kazhdan-Kostant-Sternberg \cite{KKS:1978} from the matrix Painlev\'e systems~\cite{BCR17}~\footnote{In the paper we consider only matrix Painlev\'e systems that are Hamiltonian, there are more interesting matrix Painlev\'e analogs, see e.g \cite{Adler:2021Matrix}}. In the Calogero--Painlev\'e case we consider the B\"acklund invariant subset of the phase space. It appears that in the matrix case it is natural to consider a subset invariant under the B\"acklund transformation twisted by the conjugation by a certain permutation matrix. In the matrix case we also perform an additional reduction.

\begin{example} \label{ex:intro_matrix}	
	The phase space of matrix Painlev\'e $\mathrm{II}$ consists of pairs of \(N\times N\) matrices \(p,q\) with symplectic form $\mathrm{Tr}(\rd p\wedge \rd q)$. The Hamiltonian has the form (cf. Hamiltonian in scalar case \eqref{Ham_PII_intro})
	\begin{equation}\label{Ham_MPII_intro}
		H_{\mathrm{MPII}}(p,q;t)=\mathrm{Tr} \left(\frac{1}{2}p^2 - \frac{1}{2}\left(q^2 + \frac{t}{2}\right)^2-\theta q\right),
	\end{equation}	
	where parameter $\theta$ remains a scalar variable. 

	As before we have B\"acklund transformation $r: (p,q)\mapsto (-p,-q)$ with $\theta\mapsto -\theta$. Let us take $2n\times 2n$ matrix Painlev\'e $\mathrm{II}$ with $\theta=0$ and consider the subset of the phase space invariant under $\mathrm{Ad}_{S_2} \circ r$, where $S_2=\mathrm(\mathbf{1}_{n\times n}, -\mathbf{1}_{n\times n})$. This invariant subset is given by block matrices with $n\times n$ blocks
	\begin{equation}
		p=\begin{pmatrix}
		 0	&  \mathfrak{p}_{12}\\
		 \mathfrak{p}_{21} & 0 
		\end{pmatrix},	\quad 
		q=\begin{pmatrix}
			0	&  \mathfrak{q}_{12}\\
			\mathfrak{q}_{21} & 0
		\end{pmatrix}.
	\end{equation}
	It is easy to see that the dynamics preserves this subset and in coordinates \(\mathfrak{p}_{12},\mathfrak{p}_{21},\mathfrak{q}_{12},\mathfrak{q}_{21}\) has the form 
	\begin{equation}
			\dot{\mathfrak{q}}_{12}=\mathfrak{p}_{12}, \quad 	\dot{\mathfrak{q}}_{21}=\mathfrak{p}_{21}, \quad
			\dot{\mathfrak{p}}_{12}=2\mathfrak{q}_{12}\mathfrak{q}_{21}\mathfrak{q}_{12}+t\mathfrak{q}_{12}, \quad \dot{\mathfrak{p}}_{21}=2\mathfrak{q}_{21}\mathfrak{q}_{12}\mathfrak{q}_{21}+t\mathfrak{q}_{21}.
	\end{equation}
	One can check that $m_2=\mathfrak{p}_{21}\mathfrak{q}_{12} - \mathfrak{q}_{21}\mathfrak{p}_{12}$ is an integral of motion. Let us fix its value as $m_2=\ri g \mathbf{1}_{n\times n}$. Then taking $n\times n$ matrices
	$(P,Q)$ and rescaling the time
	\begin{equation}
		Q=-2^{-1/3}\mathfrak{p}_{12}\mathfrak{q}_{12}^{-1},\quad 	P=2^{1/3}\mathfrak{q}_{12}\mathfrak{q}_{21}+Q^2+s/2, \qquad s=-2^{1/3}t
	\end{equation}
	we obtain matrix Painlev\'e $\mathrm{II}$ with $\theta=-\ri g-1/2$
	\begin{equation}\label{ex_MPII}
		\frac{dQ}{ds}=P,\qquad \frac{dP}{ds}=2Q^3+sQ-\ri g-1/2.
	\end{equation}
Note that the parameter value for the resulting matrix Painlev\'e system coincides with one obtained for the scalar Painlev\'e obtained in Example \ref{ex:intro_Calogero} (cf. \eqref{ex_MPII} with \eqref{ex_CPII}).
\end{example}	
Actually the last step in the example is a Hamiltonian reduction with respect to the conjugation by certain \(\mathrm{GL}_n\).
In Theorem \ref{thm:gen_constr} we construct such Hamiltonian reductions of matrix Painlev\'e systems for all folding transformations of the Painlev\'e equations. This is one of the main results of the paper. Moreover, our proof of Theorem \ref{thm:CP_red} mentioned above is based on the Theorem \ref{thm:gen_constr}.


\paragraph{Plan of the paper.}
In Sec. \ref{sec:Backlund} we recall matrix Painlev\'e equations in their  Hamiltonian forms. We also lift all B\"acklund transformations known in scalar case to the matrix case. This lift is not completely straightforward due to noncommutativity of variables, see e.g. Tables \ref{table:matrix_P2_Backlund} and \ref{table:matrix_P6_Backlund} below.

In Sec. \ref{sec:block_reductions} we construct block reductions for the matrix Painlev\'e systems. 
The construction works for B\"acklund transformation which preserves parameters of certain Painlev\'e system but acts on \(p,q\) nontrivially. 
Let \(w\) be such a transformation, \(d\) denotes  order of \(w\) and  \(\bar{w}\) denotes \(w\) twisted by adjoint action of a certain permutation matrix of order \(d\). In Theorem \ref{thm:gen_constr} we consider subset in the phase space of \(nd \times nd\) matrix Painlev\'e system  that is invariant under the action of \(\bar{w}\), and show that the dynamics on its Hamiltonian reduction is equivalent to \(n\times n\) matrix Painlev\'e system. The proof of this theorem is based on case by case considerations, which occupy the bulk of Sec. \ref{sec:block_reductions}. In Sec.~\ref{ssec:special_C2xC2} we prove similar statement for the non-cyclic subgroups.

In Sec.~\ref{ssec:CP_systems_intro} we recall definition of Calogero--Painlev\'e systems. In the 
following Sec.~\ref{ssec:redCP} we prove Theorem \ref{thm:CP_red}. Roughly speaking, the theorem states that the dynamics on \(w\)-invariant set of \(nd\) Calogero--Painlev\'e particles is equivalent to the dynamics of \(n\) Calogero--Painlev\'e particles. As we mentioned above, the proof is based on Theorem \ref{thm:gen_constr} and construction of Calogero--Painlev\'e systems via Hamiltonian reduction of matrix Painlev\'e systems. In particular, we do not need case by case analysis here.

It appears that there are more relations between matrix Painlev\'e systems and Calogero--Painlev\'e systems similar to ones found in Theorems \ref{thm:gen_constr} and \ref{thm:CP_red}. 
We do not intend to classify them and just give several examples in Sec. \ref{sec:further}. 
In particular, in Sec. \ref{ssec:matr_gen} we study another block reductions of matrix Painlev\'e equations. 
In Sec. \ref{ssec:algebrCP} we study \(w\)-invariant configurations of Calogero--Painlev\'e particles in which several particles evolve by algebraic solutions of the Painlev\'e equation.
Finally, in Sec.~\ref{ssec:spin} we discuss spin generalization of the Calogero--Painlev\'e systems.

\paragraph{Acknowledgments.}
We are grateful to V.~Adler, I.~Sechin, V.~Sokolov and V.~Roubtsov for useful discussions and paying attention to the literature.

Large part of calculations related to Sec. \ref{sec:Backlund} and Sec. \ref{sec:block_reductions} is performed with the help of the package for noncommutative computations \cite{NCAlgebra}.

Authors were partially supported by the HSE University basic research program.

\section{Matrix Painlev\'e systems and their B\"acklund transformations}\label{sec:Backlund}

We follow \cite{BCR17} in conventions on matrix Painlev\'e systems.

\subsection{Painlev\'e II}
\label{ssec:PII}
\paragraph{Scalar case.}

We will consider parametrized family of ordinary differential equations (dot is $\frac{d}{dt}$)
\begin{equation}\label{PII_scalar}
	\ddot{q} = 2q^3 + 2tq + \theta .
\end{equation}
We will denote equation \eqref{PII_scalar} by PII$\left(-\theta + \frac{1}{2}, \theta + \frac{1}{2}\right)$.
This equation is Hamiltonian. Namely, one can take $\mathbb{C}^2$ with standard symplectic structure $\omega = \mathrm{d}p\wedge \mathrm{d}q$, then the Hamiltonian \eqref{Ham_PII_intro}
\begin{equation}
H(p,q;t) = \frac{p^2}{2} - \frac{1}{2}\left(q^2 + \frac{t}{2}\right)^2 - \theta q,
\end{equation}
leads to dynamics \eqref{PII_scalar}.

This family of ordinary differential equations has discrete symmetries. For example it is easy to see that if $q(t)$ is a solution of PII$\left(-\theta + \frac{1}{2}, \theta+ \frac{1}{2}\right)$, then $-q(t)$ is a solution of PII$\left(\theta + \frac{1}{2}, -\theta + \frac{1}{2}\right)$.

Let us define B\"acklund transformations in general case. Let $\mathbb{D}\left(\alpha \right)$ be the system of ordinary differential equations with extended phase space $\mathrm{M}$, which depends on the set of parameters $\alpha \in \mathbb{C}^k$.

\begin{definition}
Pair of maps $\left(\pi, \tilde{\pi}\right): \mathrm{M}\times \mathbb{C}^k  \rightarrow \mathrm{M}\times \mathbb{C}^k$ is called B\"acklund transformation if $\pi$ maps solutions of $\mathbb{D}\left(\alpha\right)$ to solutions of $\mathbb{D}\left(\tilde{\pi}\left(\alpha\right)\right)$.
\end{definition}
Let the dynamics of an ordinary differential equation $\mathbb{D}\left(\alpha \right)$ be defined on a symplectic manifold $\left(M, \omega\right)$ by a Hamiltonian $H_{\alpha}(x, t)$. Then a sufficient condition for $\left(\pi, \tilde{\pi}\right)$ to be a B\"acklund transformation is
\begin{equation}\label{symm_cond_2_form}
\pi^*\left(\omega - \mathrm{d}H_{\tilde{\pi}(\alpha)}\wedge \mathrm{d}t\right) = h(x,t,\alpha)\left(\omega - \mathrm{d}H_{\alpha}\wedge \mathrm{d}t\right),
\end{equation} for a certain function $h(x,t,\alpha)$.
Note that it is important to consider extended phase space since we work with non-authonomous system and non-authonomous symmetries.

Below we follow \cite{TOS05} in the description of B\"acklund transformations.

\paragraph{Matrix generalisation.}
Hamiltonian system corresponding to PII$\left(-\theta + \frac{1}{2}, \theta + \frac{1}{2}\right)$ can be generalized to the matrix case. Let us consider $\mathrm{Mat}_{N}\left(\mathbb{C}\right) \times \mathrm{Mat}_{N}\left(\mathbb{C}\right)$ with coordinates $(p, q)$ and symplectic structure $\omega = \mathrm{Tr}\left(\mathrm{d}p\wedge \mathrm{d}q\right)$. Matrix PII$\left(-\theta + \frac{1}{2}, \theta + \frac{1}{2}\right)$ can be defined by the Hamiltonian \eqref{Ham_MPII_intro}
\begin{equation}\label{Ham_MPII}
H_{\theta}(p,q;t) = \mathrm{Tr}\left(\frac{p^2}{2} - \frac{1}{2}\left(q^2 + \frac{t}{2}\right)^2 - \theta q \right).
\end{equation}
Note that the parameter $\theta$ remains scalar. In all cases below we will consider matrix Painlev\'e equations with scalar parameters only. Hamiltonian \eqref{Ham_MPII} leads to equations of motion

\begin{equation}
\dot{q} = p, \qquad
\dot{p} = 2q^3 + tq + \theta.
\end{equation}
In case $N = 1$ this system is equivalent to the equation PII$\left(-\theta + \frac{1}{2}, \theta + \frac{1}{2}\right)$.

The first goal is to generalize formulas for B\"acklund transformations of scalar PII$\left(-\theta + \frac{1}{2}, \theta + \frac{1}{2}\right)$ to the case of matrix PII$\left(-\theta + \frac{1}{2}, \theta + \frac{1}{2}\right)$. 

\begin{prop}\label{prop:matrix_P2_Backlund}
Transformations in the table below are B\"acklund transformations of matrix $\mathrm{PII}$. These transformations generate group $C_2\ltimes W\left(A_1^{(1)}\right)$.
\end{prop}

\begin{table}[H]
\begin{tabular}{M{10cm}M{4cm}}
	\begin{tabular}{ |c|c|c|c|c|c| } 
\hline
 & $q$ & $p$  & $t$ \\
\hline
 $s_1$ & $q + \alpha_1 f^{-1}$ & $p - \alpha_1 \left(qf^{-1} + f^{-1}q\right) - \alpha_1^2f^{-2}$ & $t$\\ 
 \hline
 $r$ & $-q$ & $-p$ & $t$ \\
 \hline
\end{tabular}
&
\begin{tikzpicture}[elt/.style={circle,draw=black!100,thick, inner sep=0pt,minimum size=2mm},scale=1.75]
			\path 	(-1,0) 	node 	(a0) [elt] {}
			(1,0) 	node 	(a1) [elt] {};
		
		     	\node at ($(a0.west) + (-0.2,0)$) 	{$\alpha_0$};
		     	\node at ($(a1.east) + (0.2,0)$) 	{$\alpha_1$};

		    \draw[<->, dashed] (a0) to [bend left=40] node[fill=white]{$r$}(a1);
		    \draw [black,line width=1pt, double distance=.07cm] (a0) -- (a1);
			\end{tikzpicture} 
			\\

\end{tabular}

\caption{B\"acklund transformations for matrix PII.}
\label{table:matrix_P2_Backlund}
\end{table}
Let us explain the notation. Here $f = p + q^2 + \frac{t}{2},\;\;\alpha_1 = \theta + \frac{1}{2},\;\;\alpha_0 = 1 - \alpha_1$. Parameters $\alpha_0$ and $\alpha_1$ are called root variables. Action of $s_1$ on them is given by $s_{1}(\alpha_1) = -\alpha_1,\;\;s_{1}(\alpha_0) = \alpha_0 + \alpha_1 $. The action of $r$ on root variables is indicated on the diagram.

For any Painlev\'e equation group of B\"acklund transformations is isomorphic to an extended affine Weyl group. Thus it is convenient to encode generators and relations of these groups through diagrams. Solid lines and nodes define affine Dynkin diagram. On the $i$-th node we write corresponding root variable $\alpha_i$. By $\{s_i\}$ we denote the reflection corresponding to $i$'th simple root. Solid line between $i$'th and $j$'th nodes corresponds to the relation $s_is_js_i = s_js_is_j$. Absense of solid line between $i$'th and $j$'th nodes corresponds to the relation $s_is_j = s_js_i$. The action on the parameters of the equation is given by 
\begin{equation}
s_{i}(\alpha_j) = \alpha_{j} - a_{ij}\alpha_i,
\end{equation}
where $\{a_{ij}\}$ is the Cartan matrix corresponding to given Dynkin diagram. Then we extend Coxeter group by finite group, acting by automorphisms of the diagram. Dashed arrows show action of automorphisms on parameters and define adjoint action on $\{s_i\}$
\begin{equation}
 \text{if }g(\alpha_i) = \alpha_j\text{, then }gs_ig^{-1} = s_j\text{.}
\end{equation}  
The last elements coming from automorphisms of the diagram will be very important in Sec. \ref{sec:block_reductions}.
\begin{proof}[Proof of proposition \ref{prop:matrix_P2_Backlund}] Let us check if transformations from Table \ref{table:matrix_P2_Backlund} are symmetries of matrix PII$\left(-\theta + \frac{1}{2}, \theta + \frac{1}{2}\right)$ using \eqref{symm_cond_2_form}.

For transformation $r$ it is easy to see that it maps Hamiltonian with parameter $\theta$ to the Hamiltonian with parameter $-\theta$. Also it preserves $\omega$ and $t$. Thus $r$ satisfies equation \eqref{symm_cond_2_form} with $h(x, t, \theta) = 1$.

For transformation $s_1$ let us consider matrix coordinates $\left(f, q\right)$. Then
\begin{align} 
\omega &= \mathrm{Tr}\left(\mathrm{d}f \wedge \mathrm{d} q\right) + \frac{1}{2}\mathrm{d}\mathrm{Tr}\left(q\right)\wedge \mathrm{d}t,
\\ \label{Ham_MPII:f}
H_{\theta}(f,q;t) 
&= \mathrm{Tr}\left(\frac{1}{2}f^2 - f\left(q^2 + \frac{t}{2}\right) - \theta q \right).
\end{align}
So we have
\begin{equation}
\omega - \rd H_{\theta} \wedge \rd t = \mathrm{Tr}\left(\rd f \wedge \rd q\right) - \rd H_{\alpha_1} \wedge \rd t,\;\;\text{where }\alpha_1 = \theta + \frac{1}{2}.
\end{equation}
It is easy to see that $s_1: f \mapsto f, \alpha_1 \mapsto -\alpha_1$. Hence we get
\begin{equation}
\begin{aligned}
s_1^*\left(\mathrm{Tr}\left(\mathrm{d}f\wedge \mathrm{d}q\right)\right) &= \mathrm{Tr}\left(\mathrm{d}f\wedge \mathrm{d}\left(q + \alpha_1 f^{-1}\right)\right) =  \mathrm{Tr}\left(  \mathrm{d}f\wedge \mathrm{d}q\right),\\
s_1^*\left(H_{-\alpha_1}\right) &= \mathrm{Tr}\left(\frac{1}{2}f^2 - f\left(\left(q + \alpha_1 f^{-1}\right)^2 + \frac{t}{2}\right) + \alpha_1\left(q + \alpha_1 f^{-1}\right) \right) = H_{\alpha_1},\\
s_1^*(t) &= t.
\end{aligned}
\end{equation}
Then we have
\begin{equation}
s_1^*(\mathrm{Tr}\left(\mathrm{d}f\wedge \mathrm{d}q\right) - \mathrm{d}H_{-\alpha_1} \wedge \mathrm{d}t) = \mathrm{Tr}\left(\mathrm{d}f\wedge \mathrm{d}q\right) - \mathrm{d}H_{\alpha_1} \wedge \mathrm{d}t.
\end{equation}
Here $s_1$ is a symmetry of the matrix PII$\left(-\theta + \frac{1}{2}, \theta + \frac{1}{2}\right)$.

Now it remains to check group relations between generators. In the case considered we have to check that $s_1$, $r$ are involutions, which can be done by straightforward computation. The group obtained is not smaller than  $C_2 \ltimes W\left(A_{1}^{(1)}\right)$ since the action on parameters is the same as in the case of scalar PII$\left(-\theta + \frac{1}{2}, \theta + \frac{1}{2}\right)$ and action on parameters determines element of $C_2 \ltimes W\left(A_{1}^{(1)}\right)$ uniquely.
\end{proof}
\subsection{Painlev\'e VI}
Definition, Hamiltonian and B\"acklund transformations of scalar PVI$\left(\theta\right)$ can be found in \cite{TOS05}. Let us start from the definition of matrix PVI. Here and below to define matrix Painlev\'e equation we will consider symplectic structure $\omega = \mathrm{Tr}\left(\rd p \wedge \rd q\right)$ on the space of pairs of matrices $\mathrm{Mat}_{n}\left(\mathbb{C}\right) \times \mathrm{Mat}_{n}\left(\mathbb{C}\right) \ni (p, q)$. Then the dynamics can be defined by the Hamiltonian, which for matrix PVI is
\begin{equation}
\begin{aligned}
&t(t-1)H(p,q;t) = \mathrm{Tr}(pq(q-1)p(q-t)-\\
& - (\alpha_4 (q-1)(q-t) + \alpha_3 q(q-t) + (\alpha_0-1)q(q-1))p +
\alpha_2(\alpha_1+\alpha_2)(q-t)).
\end{aligned}
\end{equation}

\begin{prop}\label{prop:matrix_P6_Backlund}
Transformations in the table below are B\"acklund transformations of matrix $\mathrm{PVI}$. These transformations generate group $S_4 \ltimes W\left(D_4^{(1)}\right)$.
\end{prop}

\begin{table}[H]
\begin{adjustwidth}{-1cm}{}

\begin{tabular}{M{12cm}M{3cm}}
\vspace{-2.5cm}
\begin{tabular}{|c|c|c|c|} 
\hline
 & $q$ & $p$ &  $t$ \\
\hline
$s_0$ & $q$ & $p - \alpha_0 (q-t)^{-1}$ &  \\ 
 \cline{1-3}
 $s_1$ & $q$ & $p$ & \\ 
 \cline{1-3}
 $s_2$ & $q + \alpha_2 p^{-1}$ & $p$  & $t$\\
 \cline{1-3}
 $s_3$ & $q$ & $p - \alpha_3 (q-1)^{-1}$ &\\ 
 \cline{1-3}
 $s_4$ & $q$ & $p - \alpha_4 q^{-1}$ & \\ 
 \hline
  \hline
 $\sigma_{34}$ & $1-q$ & $-p$  & $1-t$\\
  \hline
 $\sigma_{14}$ & $t^{[p,q]}q^{-1}t^{-[p,q]}$ & $ -t^{[p,q]}q(pq + \alpha_2)t^{-[p,q]}$ &  $t^{-1}$\\
  \hline
 $\sigma_{03}$ & $t^{[p,q]}(t^{-1}q)t^{-[p,q]}$ & $ t^{[p,q]}(t p)t^{-[p,q]}$  & $t^{-1}$\\
\hline
 \hline
 $\pi_1$ & $tq^{-1}$ & $ -t^{-1}q(pq + \alpha_2)$  & \parbox{0.25cm}{\vspace{0.5cm} $t$}\\
\cline{1-3}
 $\pi_2$ & $t(q - 1)(q - t)^{-1}$ & $-\left( t(t-1)\right)^{-1}(q-t)\left(p(q-t) + \alpha_2\right)$  & \\
  \hline
 
 \end{tabular}
 &
 \begin{tikzpicture}[elt/.style={circle,draw=black!100,thick, inner sep=0pt,minimum size=2mm},scale=1.75]
			\path 	(-1,-1) 	node 	(a0) [elt] {}
			(-1,1) 	node 	(a1) [elt] {}
			(0,0) node  	(a2) [elt] {}
			(1,-1) node  	(a3) [elt] {}
			(1,1) node  	(a4) [elt] {};
			
		     	\node at ($(a0.west) + (-0.2,0)$) 	{$\alpha_0$};
		     	\node at ($(a1.west) + (-0.2,0)$) 	{$\alpha_1$};
		     	\node at ($(a2.north) + (0,0.2)$) 	{$\alpha_2$};
		     	\node at ($(a3.east) + (0.2,0)$) 	{$\alpha_3$};
		     	\node at ($(a4.east) + (0.2,0)$) 	{$\alpha_4$};
		     	
		    \draw[<->, dashed] (a3) to [bend right=40] node[fill=white]{$\sigma_{34}$} (a4);
		     \draw[<->, dashed] (a1) to [bend left=40] node[fill=white]{$\sigma_{14}$} (a4);
		     \draw[<->, dashed] (a3) to [bend left=40] node[fill=white]{$\sigma_{03}$} (a0);
		     
		     \draw[<->, dashed] (a3) to [bend left=40] node[fill=white]{$\pi_2$} (a4);
		     \draw[<->, dashed] (a1) to [bend right=40] node[fill=white]{$\pi_1$} (a4);
		     \draw[<->, dashed] (a3) to [bend right=40] node[fill=white]{$\pi_1$} (a0);
		      \draw[<->, dashed] (a1) to [bend left=40] node[fill=white]{$\pi_2$} (a0);
		   
		    \draw [black,line width=1pt] (a0) -- (a2) -- (a4) (a1) -- (a2) -- (a3);
		   
			\end{tikzpicture} 
			\\
\end{tabular}
 
\caption{B\"acklund transformations for matrix PVI.}
\label{table:matrix_P6_Backlund}
\end{adjustwidth}
\end{table}

Automorphisms $\mathrm{Aut}\left(D_4^{(1)}\right)=S_4$ are generated by permutations $\sigma_{34}, \sigma_{14}, \sigma_{03}$.
In general this group changes the time variable $t$, its time preserving subgroup is $C_2^2$,
generated by $\pi_1=\sigma_{14}\sigma_{03}, \, \pi_2=(\sigma_{34}\sigma_{03}\sigma_{14})^2$.
This is just a group of automorphisms of affine Weyl group.

Proof of Proposition \ref{prop:matrix_P6_Backlund} can be done by tedious but straightforward calculation, similar to PII case.

Note that $s_1, s_2$ and $\sigma$'s generate the rest of transformations. Therefore it suffices to check B\"acklund symmetry condition only for them. The following remarks simplify the proofs.
\begin{remark}\label{rem:comm_tr}
Let $m_1, m_2$ be monomials in $p, q$ with coefficient $1$ such that 
\begin{equation}
\deg_{p}(m_1) = \deg_{p}(m_2),\;\;\deg_{q}(m_1) = \deg_{q}(m_2).
\end{equation}
If $\deg_{p}(m_1) + \deg_{q}(m_1) \leq 3$, then 
\begin{equation}
\mathrm{Tr}\left(m_1\right) = \mathrm{Tr}\left(m_2\right).
\end{equation}
If $\deg_{p}(m_1) + \deg_{q}(m_1) = 4$, then the same is true except for the case $\deg_{p}(m_1) = \deg_{q}(m_1) = 2$. In this case either $\mathrm{Tr}\left(m_1\right) = \mathrm{Tr}\left(p^2q^2\right)$ or $\mathrm{Tr}\left(m_1\right) = \mathrm{Tr}\left(pqpq\right)$.
\end{remark}

\begin{remark}\label{rem:equiv_Ham}
Let $H_1$ and $H_2$ be two Hamiltonians such that $H_1 - H_2 = f(t)\left(\mathrm{Tr}\left(pqpq\right) - \mathrm{Tr}\left(p^2q^2\right)\right)$. Then
\begin{equation}\label{Change_conjugation}
\mathrm{Tr}\left(\mathrm{d}p\wedge \mathrm{d}q\right) -\rd H_1 \wedge \rd t = \mathrm{Tr}\left(\mathrm{d}\tilde{p}\wedge \mathrm{d}\tilde{q}\right) -\rd H_2 \wedge \rd t,
\end{equation}
where 
\begin{equation}\label{Conjugation_transform}
	\tilde{p} = s^{[p, q]}ps^{-[p,q]}\;\;\tilde{q} = s^{[p,q]}qs^{-[p,q]},\;\; s = \exp \left(\int f(t)\rd t\right).
\end{equation}
In other words, the dynamics generated by $H_1$ can be mapped to the dynamics generated by $H_2$ by a certain non-autonomous change of variables.
\end{remark}

One can see with the help of Remark \ref{rem:comm_tr} that for the transformations $s_0, s_2$ form $\omega - \mathrm{d}H \wedge \mathrm{d}t$ changes in the same way as in commutative case.

There is an additional conjugation in $\sigma_{14}, \sigma_{03}$ which disappears in commutative case. Let us explain it's appearance. Consider $\tilde{\sigma}_{03}$ defined by the formula
\begin{equation}
\tilde{\sigma}_{03} : q \mapsto \frac{q}{t},\;\; p \mapsto tp,\;\;t\mapsto t^{-1}.
\end{equation}
In commutative case $\tilde{\sigma}_{03}$ is a B\"acklund transformation. In the matrix case almost all terms in Hamiltonian transform in the same way as in commutative case. The only difference appears in transformation of the first term in the Hamiltonian
\begin{equation}
\tilde{\sigma}_{03}^*\left(\rd t \wedge \rd \left(\frac{1}{t(t-1)}\mathrm{Tr}\left(pq(q-1)p(q-t)\right)\right)\right) = \rd t \wedge \rd \left(\frac{1}{t(t-1)}\mathrm{Tr}\left(pq(q-t)p(q-1)\right)\right).
\end{equation}
Hence we have
\begin{equation}
\begin{aligned}
\tilde{\sigma}_{03}^*\left(\omega - \rd H_{\sigma_{03}(\alpha)} \wedge \rd t\right) = \omega - \rd H_{\alpha} \wedge \rd t - \rd \left(\frac{\mathrm{Tr}\left(pq((q{-}t)p(q{-}1)-(q{-}1)p(q{-}t))\right)}{t(t-1)}\right) \wedge dt = \\ = \omega - \rd \left(H_{\alpha} - \frac{1}{t}\mathrm{Tr}\left(pqpq - p^2q^2\right)\right)\wedge \rd t.
\end{aligned}
\end{equation}
We see that $\tilde{\sigma}_{03}$ fails to satisfy equation \eqref{symm_cond_2_form}. But it follows from Remark \ref{rem:equiv_Ham} that for $\sigma_{03}$ se have
\begin{equation}
\sigma_{03}^*\left(\omega - \rd H_{\sigma_{03}(\alpha)} \wedge \rd t\right) = \omega - \rd H_{\alpha} \wedge \rd t.
\end{equation}
Transformations $\sigma_{14}, \sigma_{34}$ can be treated similarly.

To finish the proof it remains to check group relations encoded in Table \ref{table:matrix_P6_Backlund}. It can be done directly. 

\subsection{Answers for the other Painlev\'e equations}
In this part we encounter groups of B\"acklund transformations which generalize B\"acklund groups of scalar Painlev\'e equations for the matrix case following the next scheme:

\begin{enumerate}
\item Hamiltonian which defines dynamics.
\item Table and diagram which encode action of B\"acklund transformations of corresponding equation.
\item Answer for the B\"acklund group.
\end{enumerate}

Note that all transformations we enlist are symplectic. Also, they do preserve $[p, q]$.

\paragraph{Painlev\'e V.}
The system is defined by Hamiltonian
\begin{equation}
tH(p,q;t) = \mathrm{Tr}\left(p(p+t)q(q-1) - \left(\alpha_1 + \alpha_3 \right)pq + \alpha_1 p + \alpha_2 tq\right).
\end{equation}

\begin{table}[H]
\begin{tabular}{M{10cm}M{4cm}}
\vspace{-1.5cm}
\begin{tabular}{ |c|c|c|c|c| } 
\hline
 & $q$ & $p$  & $t$ \\
\hline
 $s_0$ & $q + \alpha_0(p + t)^{-1}$ & $p$ & $t$\\ 
 \hline
 $s_1$ & $q$ & $p - \alpha_1 q^{-1}$ & $t$\\ 
 \hline
 $s_2$ & $q + \alpha_2 p^{-1}$ & $p$ & $t$\\
 \hline
  $s_3$ & $q$ & $p - \alpha_3 (q - 1)^{-1}$ & $t$\\
 \hline
 $\pi$ & $-t^{-1}p$ & $ t(q - 1)$ & $t$\\
\hline
$\sigma$ & $1-q$ & $ -p-t$ & $-t$\\
\hline

\end{tabular}
&
\begin{tikzpicture}[elt/.style={circle,draw=black!100,thick, inner sep=0pt,minimum size=2mm},scale=1.25]
			\path 	(-1,-1) 	node 	(a0) [elt] {}
			(-1,1) 	node 	(a1) [elt] {}
			(1,-1) node  	(a3) [elt] {}
			(1,1) node  	(a2) [elt] {};
			
		     	\node at ($(a0.west) + (-0.2,0)$) 	{$\alpha_0$};
		     	\node at ($(a1.west) + (-0.2,0)$) 	{$\alpha_1$};
		     	\node at ($(a2.east) + (0.2,0)$) 	{$\alpha_2$};
		     	\node at ($(a3.east) + (0.2,0)$) 	{$\alpha_3$};

		    \draw[->, dashed] (a0) to [bend left=40] node[fill=white]{$\pi$} (a1);
		     \draw[->, dashed] (a1) to [bend left=40] node[fill=white]{$\pi$} (a2);
		    \draw[->, dashed] (a2) to [bend left=40] node[fill=white]{$\pi$} (a3);
		    \draw[->, dashed] (a3) to [bend left=40] node[fill=white]{$\pi$} (a0);
			\draw[<->, dashed] (a3) to [bend right=0] node[fill=white]{$\sigma$} (a1);	   
			\draw [black,line width=1pt] (a0) -- (a1);
			\draw [black,line width=1pt] (a1) -- (a2);
			\draw [black,line width=1pt] (a2) -- (a3);
			\draw [black,line width=1pt] (a3) -- (a0);

			\end{tikzpicture} 
			\\
\end{tabular}
\caption{B\"acklund transformations for matrix PV.}
\label{table:matrix_P5_Backlund}
\end{table}
These transformations generate group $\left( C_2\ltimes C_4 \right)\ltimes W(A_3^{(1)})$.

\paragraph{Painlev\'e III$\mathrm{\big(D_6^{(1)}\big)}$.}
The system is defined by Hamiltonian
\begin{equation}\label{MPIIID6}
tH(p,q;t) = \mathrm{Tr}\left(p^2q^2 - (q^2 - (\alpha_1 + \beta_1)q -t)p - \alpha_1 q\right).
\end{equation}

\begin{table}[H]
\begin{tabular}{M{10cm}M{4cm}}
\begin{tabular}{|c|c|c|c|c|c|c|} 
\hline
 & $q$ & $p$ & $t$ \\
\hline
 $s_1$ & $q + \alpha_1 p^{-1}$ & $p$ &$t$\\ 
 \hline
 $\pi \circ \pi'$ & $ -q $&$1-p + (\alpha_0{-} \beta_1) q^{-1} - t q^{-2}$ & $t$ \\
 \hline
 $s_1'$ & $q + \beta_1 (p - 1)^{-1}$ & $p$ & $t$\\ 
 \hline
  $\pi' $ & $ tq^{-1}$&$-t^{-1}q \left(pq + \alpha_1\right)$ & $t$ \\
  \hline
 $\sigma$ & $-q$ & $1-p$& $-t$\\
  \hline
\end{tabular}

&

\begin{tikzpicture}[elt/.style={circle,draw=black!100,thick, inner sep=0pt,minimum size=2mm},scale=1.5]
			\path 	(-1,0.5) 	node 	(a0) [elt] {}
			(1,0.5) 	node 	(a1) [elt] {}
			(-1,-0.5) node  	(b0) [elt] {}
			(1,-0.5) node  	(b1) [elt] {};
			
		     	\node at ($(b0.west) + (-0.2,0)$) 	{$\beta_0$};
		     	\node at ($(a0.west) + (-0.2,0)$) 	{$\alpha_0$};
		     	\node at ($(a1.east) + (0.2,0)$) 	{$\alpha_1$};
		     	\node at ($(b1.east) + (0.2,0)$) 	{$\beta_1$};
		     	
		    \draw[<->, dashed] (a0) to [bend right=40] node[fill=white]{$\sigma$} (b0);
		    \draw[<->, dashed] (a1) to [bend left=40] node[fill=white]{$\sigma$} (b1);
		    \draw[<->, dashed] (a0) to [bend left=40] node[fill=white]{$\pi$} (a1);		    
			\draw[<->, dashed] (b0) to [bend right=40] node[fill=white]{$\pi'$} (b1);	
		    \draw [black,line width=1pt, double distance=.07cm] (a0) -- (a1);
		    \draw [black,line width=1pt, double distance=.07cm] (b0) -- (b1);

			\end{tikzpicture} 
			\\
\end{tabular}
\caption{B\"acklund transformations for matrix PIII$\mathrm{\big(D_6^{(1)}\big)}$.}
\label{table:matrix_P3D6_Backlund}

\end{table}
These transformations generate group $\left(C_2 \ltimes \left(C_2 \times C_2\right)\right)\ltimes W\left( A_1^{(1)}\right)^2$. 

\paragraph{Painlev\'e III$\mathrm{\big(D_7^{(1)}\big)}$.}
The system is defined by Hamiltonian
\begin{equation}
tH(p,q;t) = \mathrm{Tr}\left(pqpq + \alpha_1 pq + tp + q\right).
\end{equation}

\begin{table}[H]
\begin{tabular}{M{10cm}M{4cm}}
\begin{tabular}{ |c|c|c|c|} 
\hline
 & $q$ & $p$ & $t$ \\
\hline
 $s_0$ & $q$ & $p - \alpha_0 q^{-1} + tq^{-2}$ & $-t$\\  
 \hline
 $\sigma$ & $tp$ & $-t^{-1}q$ & $-t$ \\
 \hline

\end{tabular}
&
\begin{tikzpicture}[elt/.style={circle,draw=black!100,thick, inner sep=0pt,minimum size=2mm},scale=1.75]
			\path 	(-1,0) 	node 	(a0) [elt] {}
			(1,0) 	node 	(a1) [elt] {};
		
		     	\node at ($(a0.west) + (-0.2,0)$) 	{$\alpha_0$};
		     	\node at ($(a1.east) + (0.2,0)$) 	{$\alpha_1$};

		    \draw[<->, dashed] (a0) to [bend left=40] node[fill=white]{$\sigma$}(a1);
		    \draw [black,line width=1pt, double distance=.07cm] (a0) -- (a1);
			\end{tikzpicture} 
			\\

\end{tabular}
\caption{B\"acklund transformations for matrix PIII$\mathrm{\big(D_7^{(1)}\big)}$.}
\label{table:matrix_P3D7_Backlund}
\end{table}

These transformations generate group $C_2\ltimes W(A_1^{(1)})$.

\paragraph{Painlev\'e III$\mathrm{\big(D_8^{(1)}\big)}$.}
The system is defined by Hamiltonian
\begin{equation}
tH(p,q;t) = \mathrm{Tr}\left(pqpq + pq - q - tq^{-1}\right).
\end{equation}

\begin{table}[H]
\begin{center}
\begin{tabular}{ |c|c|c|c|c| } 
\hline
 & $q$ & $p$ & $t$ \\
\hline
 $\pi$ & $tq^{-1}$ & $-t^{-1}q\left(pq + \frac{1}{2}\right)$ & $t$\\ 
 \hline
\end{tabular}
\caption{B\"acklund transformations for matrix PIII$\mathrm{\big(D_8^{(1)}\big)}$.}
\label{table:matrix_P3D8_Backlund}
\end{center}
\end{table}

This transformation generates group $C_2$.

\paragraph{Painlev\'e IV.}
The system is defined by Hamiltonian
\begin{equation}
H(p,q;t) = \mathrm{Tr}\left((p-q-2t)pq - 2\alpha_1 p - 2\alpha_2 q\right).
\end{equation}
\begin{table}[H]
\begin{tabular}{M{11cm}M{4cm}}
\vspace{-2.5cm}
\begin{tabular}{ |c|c|c|c|c|c| } 
\hline
 & $q$ & $p$& $t$ \\
\hline
 $s_0$ & $q + 2\alpha_0(p - q - 2t)^{-1}$ & $p + 2\alpha_0 (p - q - 2t)^{-1}$ & $t$\\ 
 \hline
 $s_1$ & $q$ & $p - 2\alpha_1 q^{-1}$ & $t$\\ 
 \hline
 $s_2$ & $q +2\alpha_2 p^{-1}$ & $p$ & $t$\\
 \hline
 $\pi$ & $-p$ & $-p + q + 2t$ & $t$\\ 
\hline
 $\sigma_1$ & $-\ri p$ & $-\ri q$ & $\ri t$\\
 \hline
\end{tabular}
&
\begin{tikzpicture}[elt/.style={circle,draw=black!100,thick, inner sep=0pt,minimum size=2mm},scale=1.25]
			\path 	(-1,0) 	node 	(a0) [elt] {}
			(1,0) 	node 	(a1) [elt] {}
			(0, 1.73205) 	node 	(a2) [elt] {};
		
		     	\node at ($(a0.west) + (-0.2,0)$) 	{$\alpha_0$};
		     	\node at ($(a1.east) + (0.2,0)$) 	{$\alpha_1$};
		     	\node at ($(a2.north) + (0,0.2)$) 	{$\alpha_2$};

		    \draw[<->, dashed] (a0) to [bend left=40] node[fill=white]{$\sigma_1$}(a1);
		    \draw[->, dashed] (a0) to [bend right=40] node[fill=white]{$\pi$}(a1);
		    \draw[->, dashed] (a1) to [bend right=40] node[fill=white]{$\pi$}(a2);
		    \draw[->, dashed] (a2) to [bend right=40] node[fill=white]{$\pi$}(a0);
		    \draw [black,line width=1pt] (a0) -- (a1);
		    \draw [black,line width=1pt] (a1) -- (a2);
		    \draw [black,line width=1pt] (a2) -- (a0);
			\end{tikzpicture} 
			\\

\end{tabular}

\caption{B\"acklund transformations for matrix PIV.}
\label{table:matrix_P4_Backlund}
\end{table}

These transformations generate group $\left(C_4 \ltimes C_3\right)\ltimes W(A_2^{(1)})$. Structure of the semidirect product in $\mathbb{Z}_4 \ltimes \mathbb{Z}_3$ is defined by the relations
\begin{equation}
\pi^3 = e\;\;, 
\sigma_1^4 = e,\quad \sigma_1\pi\sigma_1^{-1} = \pi^{-1}
\end{equation}
 
\begin{remark}\label{remark:P4_automorphisms}
There is central subgroup $\{\sigma_1^2, e\} \subset C_4 \ltimes C_3$ which acts trivially on the parameters. Quotient $C_4 \ltimes C_3 / \{\sigma_1^2, e\}$ is isomorphic to the group of automorphisms of the diagram. This is only the case when the finite group we extend affine Weyl group by is not isomorphic to the group of automorphisms of the diagram.
\end{remark}

\paragraph{Painlev\'e I.}
The system is defined by the Hamiltonian
\begin{equation}
H(p,q;t) = \mathrm{Tr}\left(\frac{p^2}{2} - \frac{q^3}{2} - \frac{tq}{4}\right).
\end{equation}

\begin{table}[H]
\begin{center}
\begin{tabular}{ |c|c|c|c|c|c| } 
\hline
 & $q$ & $p$& $t$ \\
\hline
 $\pi$ & $\mu^3 q$ & $\mu^3 p$ & $\mu t$\\ 
 \hline
\end{tabular}
\caption{B\"acklund transformations for matrix PI}
\end{center}
\end{table}
Here $\mu$ is a scalar such that $\mu^5 = 1,\;\;\mu\neq 1$. This transformation generates group $C_5$.
\section{Block reduction of Matrix Painlev\'e systems}\label{sec:block_reductions}

\subsection{General construction}\label{ssec:gen_constr}
In this section we provide a construction which connects two matrix Painlev\'e systems of different sizes. In this construction input is a matrix Painlev\'e system and a B\"acklund transformation $w$ of this system. Output is the matrix Painlev\'e system which we call by image. We specify input and output in the Table~\ref{table:automorphisms} below. Let us denote extended phase space by $M = \mathrm{Mat}_{N}\left(\mathbb{C}\right)\times \mathrm{Mat}_{N}\left(\mathbb{C}\right) \times \mathbb{C} \times \mathbb{C}^k = \{(p, q, t, \alpha)\}$, where $k$ is the number of parameters of the equation. For the phase space we will use notations $M_{\alpha} = \mathrm{Mat}_{N}\left(\mathbb{C}\right)\times \mathrm{Mat}_{N}\left(\mathbb{C}\right) \times \mathbb{C} = \{(p, q, t)\},\;\; M_{\alpha, t} = \mathrm{Mat}_{N}\left(\mathbb{C}\right)\times \mathrm{Mat}_{N}\left(\mathbb{C}\right) = \{(p,q)\}$. By $d$ we mean the order of the B\"acklund transformation $w$.
\begin{table}[h]
\begin{center}
\begin{tabular}{|c|c|c|c|c|c|c|c|}
\hline
\textnumero & Equation & Image & $w$ & $d$ & $ q $ & $p$ & Section \\
\hline
1 & PII & PII &$r$& $2$ & $-q$ & $-p$ & \ref{sssec:II_II} \\
\hline
2 & PIII$\mathrm{\big(D_{6}^{(1)}\big)}$ & PIII$\mathrm{\big(D_{8}^{(1)}\big)}$ & $\pi\circ \pi'$ &$2$ &$-q$ & $1 - p - tq^{-2}$ & \ref{sssec:IIID6_IIID8} \\
\hline
3 & PV & PIII$\mathrm{\big(D_{6}^{(1)}\big)}$ &$\pi^2$  & $2$ &$1 - q$ & $-t-p$ & \ref{sssec:V_IIID6} \\
\hline
4 & PIV  & PIV  & $\pi$ & $3$ & $-p$ & $-p + q + 2t$ & \ref{sssec:IV_IV}\\
\hline
5 & PV  & PV  & $\pi$ & $4$ & $-t^{-1}p$ & $t(q-1)$ & \ref{sssec:V_V} \\
\hline
6 & PIII$\mathrm{\big(D_{8}^{(1)}\big)}$ & PIII$\mathrm{\big(D_{6}^{(1)}\big)}$ & $\pi$ & $2$ & $t q^{-1}$ &$-t^{-1}q\left(pq + \frac{1}{2}\right)$ & \ref{sssec:IIID8_IIID6} \\
\hline
7 & PIII$\mathrm{\big(D_{6}^{(1)}\big)}$ & PV& $\pi'$ & $2$ &$t q^{-1}$ &$-t^{-1}q\left(pq + \alpha_1\right)$ & \ref{sssec:IIID6_V} \\
\hline
8 & PVI & PVI & $\pi_1$ & $2$ & $t q^{-1}$ &$-t^{-1}q\left(pq + \alpha_2\right)$ & \ref{sssec:VI_VI} \\
\hline
\end{tabular}
\caption{Input and output.}
\label{table:automorphisms}
\end{center}
\end{table}

\begin{remark}\label{remark_folding_transformations}
Transformations from Table \ref{table:automorphisms} are specified by following conditions
\begin{itemize}
\item There exists $\alpha\in \mathbb{C}^k$ such that $M_{\alpha, t}$ is invariant under the action of $w$.
\item $w$ acts on $M_{\alpha, t}$ nontrivially.
\end{itemize}
Note that all these transformations come from automorphisms of corresponding Dynkin diagrams and preserve time. They are known to be related to folding transformations and are classified in \cite{TOS05}.
\end{remark}

For each case from Table \ref{table:automorphisms} we consider matrix Painlev\'e of size $nd\times nd$. There is the adjoint action of $\mathrm{GL}_{nd}\left(\mathbb{C}\right)$ on $M_{\alpha, t}$, namely $S:(p,q)\mapsto (SpS^{-1}, SqS^{-1})$. Consider twisted B\"acklund transformation $\bar{w} = \mathrm{Ad}_{S_{d}} \circ w$, where $S_d = \mathrm{diag}\left(\mathbf{1}_{n\times n}, e^{\frac{2\pi \ri }{d}}\mathbf{1}_{n\times n}, ..., e^{\frac{2\pi (d-1) }{d}}\mathbf{1}_{n\times n}\right)$. Transformation $\bar{w}$ is a symmetry of the equation of order $d$. Let $M$ be the phase space of the equation considered. Then for $M^{\bar{w}} = \{x \in M | \bar{w}(x) = x\}$ standard arguments imply
\begin{itemize}
\item $M_{\alpha, t}^{\bar{w}}$ is a symplectic submanifold of $M_{\alpha, t}$.
\item $M_{\alpha}^{\bar{w}}$ is preserved by the dynamics of the equation.
\end{itemize}

\begin{lemma}\label{lm:gen_constr_Ham_dyn}
Let $w$ be a transformation from Table \ref{table:automorphisms}. Let $\alpha\in \mathbb{C}^k$ be preserved by $w$. Let $M$ be the extended phase space of the corresponding matrix $ n d\times nd$ Painlev\'e system and $H$ be the corresponding Hamiltonian. Then restriction of the dynamics, generated by $H$ to $M^{\bar{w}}_{\alpha}$ is Hamiltonian.
\end{lemma}
Let us denote the corresponding Hamiltonian by $\tilde{H}$. Note that adjoint action of $\mathrm{GL}_{nd}\left(\mathbb{C}\right)$ restricted to $\mathrm{GL}_{n}^d\left(\mathbb{C}\right) = \{\mathrm{diag} \left(h_1, h_2, ..., h_d\right) | h_{i}\in \mathrm{GL}_{n}\left(\mathbb{C}\right)\}$ do commute with $\bar{w}$. Thus this action preserves $M^{\bar{w}}_{\alpha}$. Consider Hamiltonian reduction on $M^{\bar{w}}_{\alpha}$ by subgroup $\mathrm{GL}^{d-1}_{n}\left(\mathbb{C}\right) = \{\mathrm{diag}\left(1, h_2, ..., h_d\right) | h_{i}\in \mathrm{GL}_{n}\left(\mathbb{C}\right)\} \subset \mathrm{GL}^{d}_{n}\left(\mathbb{C}\right)$ with the moment map value fixed as $\mathbf{g} = (\ri g_2\mathbf{1}_{n\times n}, ..., \ri g_{d}\mathbf{1}_{n\times n})$. Let us denote the reduced space by $\mathbb{M}_{\alpha} := M^{\bar{w}}_{\alpha}{//}_{\mathbf{g}}\mathrm{GL}_{n}^{d-1}\left(\mathbb{C}\right)$.

\begin{theorem}\label{thm:gen_constr}
 Under the conditions from Lemma \ref{lm:gen_constr_Ham_dyn} the dynamics on the manifold $\mathbb{M}_{\alpha}$ corresponding to $\tilde{H}$ is equivalent to the dynamics of matrix $n\times n$ Painlev\'e system written as Image in Table \ref{table:automorphisms}.
\end{theorem}
Proof of Theorem \ref{thm:gen_constr} and Lemma \ref{lm:gen_constr_Ham_dyn} is given below case by case. We perform the following steps for all cases.
\begin{enumerate}
\item[\textbf{Step 1.}]  Solve equations on fixed point and obtain Darboux coordinates on $M^{\bar{w}}_{\alpha}$.
\item[\textbf{Step 2.}] Compute the action of $\mathrm{GL}^{d}_{n}\left(\mathbb{C}\right)$ on $M^{\bar{w}}$ and moment map of this action  in these coordinates.
\item[\textbf{Step 3.}] Obtain Darboux coordinates on $\mathbb{M}_{\alpha,t}$.
\item[\textbf{Step 4.}] Calculate the Hamiltonian for the dynamics on $\mathbb{M}_{\alpha}$ and find coordinates on $\mathbb{M}_{\alpha}$, in which the Hamiltonian is Painlev\'e's standard one.
\end{enumerate}

Step $1$, Step $2$, Step $3$ can be performed simultaneously for cases $1, 2, 3$ and for cases $6, 7, 8$ in Table \ref{table:automorphisms}. We call cases $1$ -- $5$ linear, since for them to obtain parametrization of $M^{\bar{w}}$ one has to solve only linear equations. Remaining cases $6, 7, 8$ are called non-linear.
\begin{remark}
We use the following notation.
\begin{itemize}
\item Standard small letters (for example, $(p,q)$) to denote canonical coordinates on $M_{\alpha, t}$.
\item Gothic letters (for example, $(\mathfrak{p}, \mathfrak{q})$) to denote coordinates on $M_{\alpha, t}^{\bar{w}}$.
\item Standard capital letters (for example, \((P,Q)\)) to denote coordinates on $\mathbb{M}_{\alpha, t}$.
\end{itemize}
\end{remark}

\subsection{Linear cases}\label{ssec:linear_cases}
\paragraph{Step 1.} Equations that determine the symplectic submanifold $M_{\alpha, t}^{\bar{w}}$ for cases $1, 2, 3$
\begin{equation}
    S_2 q S_2^{-1}= \xi - q,\;\; S_2 p S_2^{-1}= -p + \eta(t) - \zeta t q^{-2}.
\end{equation}
Here $\zeta, \xi \in \mathbb{C}$, $\eta\in \mathbb{C}[t]$ are specified for each case. For all cases which we consider $\zeta \xi = 0$ and for them we get
\begin{equation}\label{linear_case_solutions}
    q= \begin{pmatrix} \frac{\xi}{2}\mathbf{1}_{n\times n} & \mathfrak{q}_{12} \\ \mathfrak{q}_{21} & \frac{\xi}{2}\mathbf{1}_{n\times n} \end{pmatrix},\;\;
    p= \begin{pmatrix} \frac{\eta(t)}{2}\mathbf{1}_{n\times n} - t\frac{\zeta}{2}\mathfrak{q}_{21}^{-1}\mathfrak{q}_{12}^{-1} & \mathfrak{p}_{12} \\ \mathfrak{p}_{21} & \frac{\eta(t)}{2}\mathbf{1}_{n\times n} - t\frac{\zeta}{2}\mathfrak{q}_{12}^{-1}\mathfrak{q}_{21}^{-1} \end{pmatrix}.
\end{equation}
Hence we get a symplectic submanifold of dimension \(4n^2\) with symplectic form \(\operatorname{Tr}(\rd \mathfrak{p}_{12}\wedge \rd \mathfrak{q}_{21}) + \operatorname{Tr}(\rd \mathfrak{p}_{21}\wedge \rd \mathfrak{q}_{12})\).

Recall that $H$ is the Hamiltonian which defines matrix Painlev\'e's dynamics on $M_{\alpha}$. Then the equations of motion can be identified with the one dimensional distribution on $M_{\alpha}$ defined as
\begin{equation}
\mathrm{Ker}\left(\omega - \mathrm{d}H\wedge \mathrm{d}t\right).
\end{equation} 
Let us denote the embedding of the set of $\bar{w}$--invariant points as $\iota: M_{\alpha}^{\bar{w}} \rightarrow M_{\alpha}$. Since $\iota^*(\omega) = \operatorname{Tr}(\rd \mathfrak{p}_{12}\wedge \rd \mathfrak{q}_{21}) + \operatorname{Tr}(\rd \mathfrak{p}_{21}\wedge \rd \mathfrak{q}_{12})$ and the dynamics on $M_{\alpha}^{\bar{w}}$ is defined as $\mathrm{Ker}\left(\iota^*\left(\omega - \mathrm{d}H\wedge \mathrm{d}t\right)\right)$ it follows that the dynamics on $M_{\alpha}^{\bar{w}}$ is also Hamiltonian and defined by Hamiltonian $\iota^*\left(H\right)$ and $\left(\mathfrak{p}_{12},\mathfrak{p}_{21}, \mathfrak{q}_{21}, \mathfrak{q}_{12} \right)$ are Darboux coordinates on $M_{\alpha}^{\bar{w}}$.
\paragraph{Step 2.} The remaining gauge freedom consists of block diagonal matrices $h=\mathrm{diag}(h_1, h_2)$ and the moment map is also block diagonal 
\begin{equation}
   [p,q]=m=\begin{pmatrix} m_1 & 0 \\ 0 & m_2\end{pmatrix} = \begin{pmatrix} \mathfrak{p}_{12}\mathfrak{q}_{21} - \mathfrak{q}_{12}\mathfrak{p}_{21} & 0 \\ 0 & \mathfrak{p}_{21}\mathfrak{q}_{12} - \mathfrak{q}_{21}\mathfrak{p}_{12} \end{pmatrix}. 
\end{equation}
\begin{remark}\label{rem:moment_bl_diag}
Block structure of $[p,q]$ appears not accidentally. Consider $(p, q)\in M_{\alpha, t}^{\bar{w}}$ with $\bar{w} = \mathrm{Ad}(S_2) \circ w$. As we mentioned above $[p,q]$ is preserved by B\"acklund transformation $w$. Then we have
\begin{equation*}
\bar{w}^{*}([p,q]) = [p, q] = w^{*}([p,q]),
\end{equation*} 
which implies
\begin{equation*}
\mathrm{Ad}(S_2)^*([p, q]) = [p, q] \Rightarrow [S_2, [p, q]] = 0.
\end{equation*}
From the last equation it follows that $[p,q]$ is block diagonal.
\end{remark}
\paragraph{Step 3.} Let us perform Hamiltonian reduction with respect to $\mathrm{GL}_{n}\left(\mathbb{C}\right) = \{\mathrm{diag}(\mathbf{1}_{n\times n}, h_2) | h_2 \in \mathrm{GL}_{n}\left(\mathbb{C}\right)\}$. We fix the moment map value as follows $m_2 = \mathfrak{p}_{21}\mathfrak{q}_{12} - \mathfrak{q}_{21}\mathfrak{p}_{12} = \ri g_2\mathbf{1}_{n\times n}$. Hence we get \(\mathfrak{p}_{21}=(\mathfrak{q}_{21}\mathfrak{p}_{12}+\ri g_2)\mathfrak{q}_{12}^{-1}\). The functions $\tilde{P} = \mathfrak{p}_{12} \mathfrak{q}_{12}^{-1}$ and $\tilde{Q} = \mathfrak{q}_{12}\mathfrak{q}_{21}$ are invariant under the action of $\mathrm{diag}\left(1, h_2\right)$ and thus define functions on $\mathbb{M}_{\alpha}$. Moreover $\big(\tilde{P}, \tilde{Q}\big)$ are Darboux coordinates on $\mathbb{M}_{\alpha, t}$.

\subsubsection{\texorpdfstring{PII}{2} to \texorpdfstring{PII}{2}}
\label{sssec:II_II}
\paragraph{Steps 1,2,3.} We consider matrix $2n \times 2n$ PII with $\alpha_1 = \frac{1}{2}$. In this case $\xi = 0,\; \eta(t) = 0,\; \zeta = 0$.
\paragraph{Step 4.} Substituting Darboux coordinates $\Tilde{P}, \Tilde{Q}$ into the Hamiltonian of matrix PII we get
\begin{equation}
    H(\Tilde{P}, \Tilde{Q}; t) = \mathrm{Tr}\left(\Tilde{P}\Tilde{Q}\Tilde{P} +\ri g_2\Tilde{P} - \left(\Tilde{Q} + \frac{t}{2}\right)^2\right) = -\mathrm{Tr}\left(\tilde{Q}(\tilde{Q}- \tilde{P}^2 + t) -\ri g_2\tilde{P}\right) - \frac{nt^2}{4}.
\end{equation}
 After the change of coordinates
\begin{equation}\label{P2_proper_coordinates}
    \breve{Q} = -\frac{1}{\sqrt[3]{2}}\tilde{P},\quad \breve{P} = \sqrt[3]{2}\tilde{Q},\qquad s = -\sqrt[3]{2}t,
\end{equation}
we obtain a system with matrix Darboux coordinates  $(\breve{P}, \breve{Q})$ and Hamiltonian
\begin{equation}
    H(\breve{P},\breve{Q}; s) = \mathrm{Tr}\left(\frac{1}{2}\breve{P}(\breve{P}-2\breve{Q}^2 - s) - (-\ri g_2)\breve{Q}\right).
\end{equation}
Then after the last change of coordinates
\begin{equation}\label{P2_proper_std_coordinates}
P = \breve{P} - \breve{Q}^2 -\frac{s}{2},\quad Q = \breve{Q}
\end{equation}
we obtain system with matrix Darboux coordinates  $(P, Q)$ and Hamiltonian
\begin{equation}
    H(P, Q;s) = \mathrm{Tr}\left(\frac{P^2}{2} - \frac{1}{2}\left(Q^2 + \frac{s}{2}\right)^2 - \left(-\ri g_2-\frac{1}{2} \right) Q \right),
\end{equation}
which is the Hamiltonian of matrix $n\times n$ PII$\big(1 + \ri g_2, - \ri g_2\big)$.

\subsubsection{\texorpdfstring{PIII$\mathrm{\big(D_6^{(1)}\big)}$}{3D6} to \texorpdfstring{PIII$\mathrm{\big(D_8^{(1)}\big)}$}{3D8}}\label{sssec:IIID6_IIID8}
\paragraph{Steps 1,2,3.} We consider matrix $2n\times 2n$ PIII$\mathrm{\big(D_{6}^{(1)}\big)}$ with $\alpha_0 = \alpha_1 = \frac{1}{2},\; \beta_0 = \beta_1 = \frac{1}{2}$.
In this case we have $\xi = 0,\; \eta(t) = 1,\; \zeta = 1$.
\paragraph{Step 4.} Substituting Darboux coordinates $\tilde{P}, \Tilde{Q}$ into the Hamiltonian of matrix PIII$\mathrm{\big(D_6^{(1)}\big)}$ we get
\begin{equation}
	tH(\tilde{P}, \tilde{Q}; t)=\mathrm{Tr}\left(\tilde{P}^2\tilde{Q}^2 + \tilde{P}\tilde{Q}\tilde{P}\tilde{Q} + 2(1 +\ri g_2)\tilde{P}\tilde{Q} -\frac{1}{2}\tilde{Q} - \frac{1}{2}t^2\tilde{Q}^{-1}\right).
\end{equation}
After the change of variables
\begin{equation}\label{Proper_coordinates_case_3D6_to_3D8}
	P = 4t^{[\tilde{P}, \tilde{Q}]}\left(\tilde{P} + \left(\frac{\ri g_2}{2}\right)\tilde{Q}^{-1}\right)t^{-[\tilde{P}, \tilde{Q}]},\quad Q = \frac{1}{4}t^{[\tilde{P}, \tilde{Q}]}\tilde{Q}t^{-[\tilde{P}, \tilde{Q}]},\qquad s = \frac{t^2}{16},
\end{equation}
we obtain a system with matrix Darboux coordinates  $(P, Q)$ and Hamiltonian
\begin{equation}
	sH(P,Q;s) = \mathrm{Tr}\left(PQPQ + PQ - Q - sQ^{-1}\right),
\end{equation}
which is the Hamiltonian of matrix $n\times n$ PIII$\mathrm{\big(D_8^{(1)}\big)}$.

\subsubsection{\texorpdfstring{PV}{5} to \texorpdfstring{PIII$\mathrm{\big(D_6^{(1)}\big)}$}{3D6}}
\label{sssec:V_IIID6}
\paragraph{Steps 1,2,3.} We consider matrix $2n \times 2n$ PV with $\alpha_0 = \alpha_2 = \epsilon + \frac{1}{2},\; \alpha_1 = \alpha_3 = -\epsilon$.
In this case $\xi = 1,\; \eta(t) = -t,\; \zeta = 0$.
\paragraph{Step 4.} Substituting Darboux coordinates $\Tilde{P}, \tilde{Q}$ into the Hamiltonian of matrix PV we get
\begin{equation}
    tH(\tilde{P}, \tilde{Q}; t) = \mathrm{Tr}\left(\tilde{P}^2\tilde{Q}^2 + \tilde{P}\tilde{Q}\tilde{P}\tilde{Q} - \frac{\tilde{P}^2\tilde{Q}}{4} + (4\epsilon + 2\ri g_2)\tilde{P}\tilde{Q} - \frac{\ri g_2}{2}\tilde{P} + \frac{t^2\tilde{Q}}{2}\right).
\end{equation}
 After the change of variables
\begin{equation}
    P = 4t^{-[\tilde{P}, \tilde{Q}]}\tilde{Q}t^{[\tilde{P}, \tilde{Q}]},\quad Q = -\frac{1}{4}t^{-[\tilde{P}, \tilde{Q}]}\tilde{P}t^{[\tilde{P}, \tilde{Q}]},\qquad s = -\frac{t^2}{16},
\end{equation}
we obtain a system with matrix Darboux coordinates $(P, Q)$ and Hamiltonian
\begin{equation}
    sH(P, Q; s) = \mathrm{Tr}\left(P^2Q^2 - (Q^2 + (2\epsilon +\ri g_2)Q - s)P +\ri g_2Q\right),
\end{equation}
which is the Hamiltonian of matrix $n\times n$ PIII$\mathrm{\big(D_6^{(1)}\big)}\big(1+\ri g_2, -\ri g_2, 1+2\epsilon, -2\epsilon\big)$.

\subsubsection{\texorpdfstring{PIV}{4} to \texorpdfstring{PIV}{4}}
\label{sssec:IV_IV}
\paragraph{Step 1.} This is case $4$ in Table \ref{table:automorphisms}. We consider matrix $3n\times 3n$ PIV with $\alpha_0 = \alpha_1 = \alpha_2 = \frac{1}{3}$.

Equations that determine symplectic submanifold $M^{\bar{w}}_{\alpha,t}$ are 
\begin{equation}
    S_{3} q S_{3}^{-1}= - p,\;\; S_{3} p S_{3}^{-1}= -p + q + 2t.
\end{equation}
Then
\begin{equation}
    q= \begin{pmatrix} 
 -\frac{2 t}{3}\mathbf{1}_{n\times n} & \mathfrak{q}_{12} & \mathfrak{q}_{13} \\
 \mathfrak{q}_{21} & -\frac{2 t}{3}\mathbf{1}_{n\times n} & \mathfrak{q}_{23} \\
 \mathfrak{q}_{31} & \mathfrak{q}_{32} & -\frac{2 t}{3}\mathbf{1}_{n\times n} \\
 \end{pmatrix},\;\;
    p= \begin{pmatrix}
 \frac{2 t}{3}\mathbf{1}_{n\times n} & -\varpi^2 \mathfrak{q}_{12} & -\varpi  \mathfrak{q}_{13} \\
 -\varpi  \mathfrak{q}_{21} & \frac{2 t}{3}\mathbf{1}_{n\times n} & -\varpi ^2 \mathfrak{q}_{23} \\
 -\varpi^2 \mathfrak{q}_{31} & -\varpi \mathfrak{q}_{32} & \frac{2 t}{3}\mathbf{1}_{n\times n} \\    
    \end{pmatrix},\quad \varpi = e^{\frac{2\pi \ri}{3}}.
\end{equation}
So we get a symplectic manifold of dimension \(6n^2\), with symplectic form \(\sqrt{3}\mathrm{i}\operatorname{Tr}(\rd\mathfrak{q}_{12}\wedge \rd\mathfrak{q}_{21} + \rd\mathfrak{q}_{31}\wedge \rd\mathfrak{q}_{13} + \rd\mathfrak{q}_{23}\wedge \rd\mathfrak{q}_{32})\). 

Similarly to cases above it follows that restriction of the dynamics defined by the Hamiltonian $H$ on $M^{\bar{w}}_{\alpha}$ is Hamiltonian and defined by restriction of $H$ on $M^{\bar{w}}_{\alpha}$.
\paragraph{Step 2.} The remained gauge freedom consists of block diagonal matrices $h=\mathrm{diag}(h_1, h_2, h_3)$ and moment map is also block diagonal 
\begin{equation}
    [p,q]=\begin{pmatrix} m_1 & 0 & 0 \\ 0 & m_2 & 0\\ 0 & 0 & m_3\end{pmatrix} = 
     \sqrt{3}\ri\begin{pmatrix} 
    \mathfrak{q}_{12}\mathfrak{q}_{21} - \mathfrak{q}_{13}\mathfrak{q}_{31} & 0 & 0\\
	0 & \mathfrak{q}_{23}\mathfrak{q}_{32}-\mathfrak{q}_{21}\mathfrak{q}_{12} & 0\\
	0 & 0 & \mathfrak{q}_{31}\mathfrak{q}_{13}-\mathfrak{q}_{32}\mathfrak{q}_{23}
    \end{pmatrix}.
\end{equation}
\paragraph{Step 3.} Now we can perform Hamiltonian reduction with respect to $\mathrm{GL}_n\left(\mathbb{C}\right)^2 = \{ \mathrm{diag}\left(1, h_2, h_3\right)| h_2, h_3 \in \mathrm{GL}_n\left(\mathbb{C}\right)\} $. We fix the moment map value as follows $m_2 = \ri g_2\mathbf{1}_{n\times n},\;\; m_3 = \ri g_3\mathbf{1}_{n\times n}$. 

The functions \(\tilde{P} = (\varpi - \varpi^{-1}) \mathfrak{q}_{12}\mathfrak{q}_{23}\mathfrak{q}_{13}^{-1}\) and \(\tilde{Q} = \mathfrak{q}_{13}\mathfrak{q}_{23}^{-1}\mathfrak{q}_{21}\) are invariant under the action of $\mathrm{GL}_n^2\left(\mathbb{C}\right)$, so define functions on the manifold $\mathbb{M}_{\alpha}$. Moreover these functions are Darboux coordinates on $\mathbb{M}_{\alpha, t}$.

\paragraph{Step 4.} Substitution $\tilde{P}, \tilde{Q}$ into the Hamiltonian gives
\begin{equation}
H(\tilde{P}, \tilde{Q};t) = \mathrm{Tr}\left(\tilde{P} \tilde{Q} (\tilde{P}+ \sqrt{3}\ri(\tilde{Q}+2 t))+\ri (g_2 + g_3)\tilde{P} -  \sqrt{3} g_2 \tilde{Q}\right).
\end{equation}

Then after the change of coordinates
\begin{equation}
	P = \frac{1}{\kappa}\tilde{P},\quad Q = \kappa\tilde{Q},\qquad s = \kappa t,\;\;\;\text{where }\kappa^2 =-\sqrt3 \ri,
\end{equation}
we get
\begin{equation}
H(P,Q;s) = \mathrm{Tr}\left( P Q (P-Q-2 s) + (\ri g_2 + \ri g_3)P - \ri g_2 Q \right),
\end{equation}
which is the Hamiltonian of matrix $n\times n$ PIV$\big(1+\frac{\ri g_3}{2}, -\frac{\ri g_2 + \ri g_3}{2}, \frac{\ri g_2}{2}\big)$.

\subsubsection{\texorpdfstring{PV}{5} to \texorpdfstring{PV}{5}}
\label{sssec:V_V}
\paragraph{Step 1.} This is case 5 from Table \ref{table:automorphisms}. We consider matrix $4n\times 4n$ PV with $\alpha_0 = \alpha_1 = \alpha_2 = \alpha_3 = \frac{1}{4}$.

Equations that determine the symplectic submanifold $M_{\alpha, t}^{\bar{w}}$ in case $5$ are
\begin{equation}
    S_4 q S_4^{-1}= - \frac{p}{t},\;\; S_4 p S_4^{-1}= t(q-1).
\end{equation}
The solution is
\begin{equation}
    q= \begin{pmatrix} 
 \frac{1}{2}\mathbf{1}_{n\times n} & 0 & \mathfrak{q}_{13} & \mathfrak{q}_{14} \\
 0 & \frac{1}{2}\mathbf{1}_{n\times n}  & \mathfrak{q}_{23} & \mathfrak{q}_{24}\\
   \mathfrak{q}_{31} & \mathfrak{q}_{32} & \frac{1}{2}\mathbf{1}_{n\times n} & 0\\
   \mathfrak{q}_{41} & \mathfrak{q}_{42} & 0 & \frac{1}{2}\mathbf{1}_{n\times n}\\
 \end{pmatrix},\;\;
    p= t \begin{pmatrix} 
 -\frac{1}{2}\mathbf{1}_{n\times n} & 0 & \mathfrak{q}_{13} &  -\ri \mathfrak{q}_{14} \\
    0  &  -\frac{1}{2}\mathbf{1}_{n\times n}   & \ri \mathfrak{q}_{23} &  \mathfrak{q}_{24}\\
   \mathfrak{q}_{31} & -\ri\mathfrak{q}_{32} &  -\frac{1}{2}\mathbf{1}_{n\times n}  & 0\\
    \ri\mathfrak{q}_{41} & \mathfrak{q}_{42} & 0 &  -\frac{1}{2}\mathbf{1}_{n\times n} \\
  \end{pmatrix}.
\end{equation}
So we get symplectic manifold of dimension \(8n^2\), with symplectic form
\begin{equation}
2\ri \mathrm{Tr}\left(\rd (\sqrt{t}\mathfrak{q}_{13}) \wedge \rd (\sqrt{t}\mathfrak{q}_{31}) + \rd (\sqrt{t}\mathfrak{q}_{24}) \wedge \rd (\sqrt{t}\mathfrak{q}_{42}) + \rd (\sqrt{t}\mathfrak{q}_{32}) \wedge \rd (\sqrt{t}\mathfrak{q}_{23}) + \rd (\sqrt{t}\mathfrak{q}_{41}) \wedge \rd (\sqrt{t}\mathfrak{q}_{14})\right).
\end{equation} 

Similarly to cases $1$--$3$ it follows that restriction of the dynamics defined by the Hamiltonian $H$ on $M^{\bar{w}}_{\alpha}$ is Hamiltonian and defined by restriction of $H$ on $M^{\bar{w}}_{\alpha}$.
\paragraph{Step 2.}The remained gauge freedom consists of block diagonal matrices $h=\mathrm{diag}(h_1, h_2, h_3, h_4)$ and moment map is also block diagonal
\begin{multline}
    [p,q]=\begin{pmatrix} m_1 & 0 & 0 & 0\\ 0 & m_2 & 0 & 0\\ 0 & 0 & m_3 & 0\\ 0 & 0 & 0 & m_4 \end{pmatrix} = 2\ri t \begin{pmatrix}
    \mathfrak{q}_{13}\mathfrak{q}_{31}{-} \mathfrak{q}_{14}\mathfrak{q}_{41} & 0 & 0 & 0\\ 0 &  \mathfrak{q}_{24}\mathfrak{q}_{42}{-} \mathfrak{q}_{23}\mathfrak{q}_{32} & 0 & 0\\ 0 & 0 &  \mathfrak{q}_{32}\mathfrak{q}_{23}{-} \mathfrak{q}_{31}\mathfrak{q}_{13} & 0\\ 0 & 0 & 0 & \mathfrak{q}_{41}\mathfrak{q}_{14}{-} \mathfrak{q}_{42}\mathfrak{q}_{24}
    \end{pmatrix}.
\end{multline}
\paragraph{Step 3.} Now we can perform reduction with respect to $\mathrm{GL}_n^{3}\left(\mathbb{C}\right) = \mathrm{diag}\left(
1, h_2, h_3, h_4|h_2, h_3, h_4\in \mathrm{GL}_n\left(\mathbb{C}\right)\right)$. We fix the moment map value as follows $m_2 = \ri g_2\mathbf{1}_{n\times n},\;\; m_3 = \ri g_3\mathbf{1}_{n\times n},\;\; m_4 = \ri g_4\mathbf{1}_{n\times n}$. 

The functions \(\tilde{P} = 2\ri t \mathfrak{q}_{13}\mathfrak{q}_{23}^{-1}\mathfrak{q}_{24}\mathfrak{q}_{41}\) and \(\tilde{Q} = \mathfrak{q}_{14}\mathfrak{q}_{24}^{-1}\mathfrak{q}_{23}\mathfrak{q}_{13}^{-1}\) are invariant under the action of $\mathrm{GL}_n^{3}\left(\mathbb{C}\right)$, so define functions on the manifold $\mathbb{M}_{\alpha}$. Moreover these functions are Darboux coordinates on $\mathbb{M}_{\alpha, t}$.
\paragraph{Step 4.} Substituting $\tilde{P}, \tilde{Q}$ into the Hamiltonian we get
\begin{multline}
tH(\tilde{P}, \tilde{Q}; t) = \mathrm{Tr}\left(\tilde{P}(\tilde{Q}-1)^2\tilde{P}\tilde{Q} + \left(-(\ri g_2 + \ri g_3 + 2\ri g_4)\tilde{Q}^2 + (2\ri t + 2\ri g_2 + \ri g_3 + 3 \ri g_4)\tilde{Q} -\right.\right. \\ \left. -(\ri g_2+ \ri g_4))\tilde{P} + \ri g_4 (\ri g_2 + \ri g_3 + \ri g_4)\tilde{Q}\right).
\end{multline} 
Then after the change of coordinates
\begin{equation}
P = -\left((\tilde{Q} - 1)\tilde{P} - \ri g_4\right)(\tilde{Q} - 1) ,\quad Q = \tilde{Q}(\tilde{Q} - 1)^{-1},\qquad s = -2\ri t,
\end{equation}
we get the Hamiltonian
\begin{equation}
sH(P,Q;s) = \mathrm{Tr}\left( P(P + s)Q(Q - 1) - (\ri g_4 - \ri g_3)PQ + (\ri g_2 + \ri g_4) P - \ri g_4 s Q \right),
\end{equation}
which is the Hamiltonian of matrix $n\times n$ PV$\big(1+\ri g_3, \ri g_2 + \ri g_4,-\ri g_4,  - \ri g_2 - \ri g_3\big)$.

\subsection{Non-linear cases}\label{ssec:nonlinear_cases}
\paragraph{Step 1.} Equations that determine  symplectic submanifold $M^{\bar{w}}_{\alpha,t}$ for cases $4$--$6$ are
\begin{equation}\label{nonlinear_leaf_fixed_points}
    S_2 q S_2^{-1} = tq^{-1},\quad S_2pS_2^{-1} = -\frac{q(pq + \nu)}{t}.
\end{equation}
On the open dense subset where the lower left $n\times n$ block of $q$ is invertible the solution is given by
\begin{equation}\label{nonlinear_case_solution}
    q = 
    \begin{pmatrix}
    \tilde{\mathfrak{q}}_{11} & (\tilde{\mathfrak{q}}_{11}^2 - t)\tilde{\mathfrak{q}}_{21}^{-1}\\ 
    \tilde{\mathfrak{q}}_{21} & \tilde{\mathfrak{q}}_{21}\tilde{\mathfrak{q}}_{11}\tilde{\mathfrak{q}}_{21}^{-1}
    \end{pmatrix},\;\;\;
    p = \begin{pmatrix}
    \tilde{\mathfrak{p}}_{11} & -\left((\tilde{\mathfrak{q}}_{11}^2 - t)\tilde{\mathfrak{q}}_{21}^{-1}\tilde{\mathfrak{p}}_{21} + [\tilde{\mathfrak{p}}_{11}, \tilde{\mathfrak{q}}_{11}]_{+} + \nu\right)\tilde{\mathfrak{q}}_{21}^{-1}\\
    \tilde{\mathfrak{p}}_{21} & [\tilde{\mathfrak{q}}_{21}\tilde{\mathfrak{q}}_{11}\tilde{\mathfrak{q}}_{21}^{-1}, \tilde{\mathfrak{p}}_{21}\tilde{\mathfrak{q}}_{21}^{-1}] + \tilde{\mathfrak{q}}_{21}\tilde{\mathfrak{p}}_{11}\tilde{\mathfrak{q}}_{21}^{-1}
    \end{pmatrix}.
\end{equation}
So, $ \tilde{\mathfrak{p}}_{11},  \tilde{\mathfrak{q}}_{11},  \tilde{\mathfrak{p}}_{21},  \tilde{\mathfrak{q}}_{21}$ are local coordinates on $M_{\alpha, t}^{\bar{w}}$. Let us denote the embedding of $M_{\alpha}^{\bar{w}}$ by $\iota : M_{\alpha}^{\bar{w}}  \rightarrow M_{\alpha}$.

We can obtain restriction of the canonical 1-form $\Theta$ on $M_{\alpha}^{\bar{w}}$, namely
\begin{equation}\label{canonical_1form_restriction_non-linear}
\iota^{*}\left(\Theta\right) = \mathrm{Tr}\left(\mathfrak{p}_{11}\mathrm{d}\mathfrak{q}_{11} + \mathfrak{p}_{12}\rd \mathfrak{q}_{21} - \tilde{\mathfrak{p}}_{21}\tilde{\mathfrak{q}}_{21}^{-1}\mathrm{d}t\right),
\end{equation}
where  Darboux coordinates on $M_{\alpha, t}^{\bar{w}}$ are
\begin{multline}\label{nonlinear_darboux_coordinates}
(\mathfrak{p}_{11}, \mathfrak{q}_{11}, \mathfrak{p}_{12}, \mathfrak{q}_{21})=(2(\tilde{\mathfrak{p}}_{11}{+} \tilde{\mathfrak{q}}_{11}\tilde{\mathfrak{q}}_{21}^{-1}\tilde{\mathfrak{p}}_{21}), \tilde{\mathfrak{q}}_{11}, 2\left(t\tilde{\mathfrak{q}}_{21}^{-1}\tilde{\mathfrak{p}}_{21}{-} \tilde{\mathfrak{p}}_{11}\tilde{\mathfrak{q}}_{11}{-} \tilde{\mathfrak{q}}_{11}\tilde{\mathfrak{q}}_{21}^{-1}\tilde{\mathfrak{p}}_{21}\tilde{\mathfrak{q}}_{11}{-} \frac{\nu}{2}\right)\tilde{\mathfrak{q}}_{21}^{-1}, \tilde{\mathfrak{q}}_{21}).
\end{multline}

Recall that $H$ is the Hamiltonian which defines matrix Painlev\'e's dynamics on $M_{\alpha}$. Then the equations of motion can be identified with the one dimensional distribution on $M_{\alpha}$ defined as
\begin{equation}
\mathrm{Ker}\left(\omega - \mathrm{d}H\wedge \mathrm{d}t\right).
\end{equation}
It follows that the dynamics on $M_{\alpha}^{\bar{w}}$ is also Hamiltonian and defined by the Hamiltonian $\iota^*\left(H\right) + \mathrm{Tr}\left(\tilde{\mathfrak{p}}_{21}\tilde{\mathfrak{q}}_{21}^{-1}\right)$.

Even though  $(\mathfrak{p}_{11}, \mathfrak{q}_{11}, \mathfrak{p}_{12}, \mathfrak{q}_{21})$ are Darboux coordinates on $M_{\alpha, t}^{\bar{w}}$, it will be more convenient for us to use $(\tilde{\mathfrak{p}}_{11}, \tilde{\mathfrak{q}}_{11}, \tilde{\mathfrak{p}}_{12}, \tilde{\mathfrak{q}}_{21})$.
\paragraph{Step 2.} As in the linear cases we have the action of block diagonal matrices
\begin{equation}
    h = \mathrm{diag}\left(h_1, h_2\right): (\{\tilde{\mathfrak{q}}_{ij}\}, \{\tilde{\mathfrak{p}}_{ij} \}) \mapsto (\{h_i \tilde{\mathfrak{q}}_{ij}h_{j}^{-1}\}, \{h_i \tilde{\mathfrak{p}}_{ij}h_{j}^{-1}\}).
\end{equation}
The moment map corresponding to this action is block diagonal, namely $[p,q]=\mathrm{diag}(m_1,m_2)$, where 
\begin{equation}
\begin{aligned}
&m_1=-2((\tilde{\mathfrak{q}}_{11}^2 - t)\tilde{\mathfrak{q}}_{21}^{-1}\tilde{\mathfrak{p}}_{21} + \tilde{\mathfrak{q}}_{11}\tilde{\mathfrak{p}}_{11}) - \nu,\\
&m_2=2\left(-t\tilde{\mathfrak{p}}_{21} + \tilde{\mathfrak{q}}_{21}\tilde{\mathfrak{p}}_{11}\tilde{\mathfrak{q}}_{11} + \tilde{\mathfrak{q}}_{21}\tilde{\mathfrak{q}}_{11}\tilde{\mathfrak{q}}_{21}^{-1}\tilde{\mathfrak{p}}_{21}\tilde{\mathfrak{q}}_{11}\right)\tilde{\mathfrak{q}}_{21}^{-1} + \nu.
\end{aligned}
\end{equation}

\paragraph{Step 3.}Let us perform Hamiltonian reduction with respect to $\mathrm{GL}_{n}\left(\mathbb{C}\right) = \{\mathrm{diag}(1, h_2) | h_2 \in \mathrm{GL}_{n}\left(\mathbb{C}\right)\}$. Then fixing $m_2 =\ri g_2\mathbf{1}_{n\times n}$ and resolving it with respect to $\tilde{\mathfrak{p}}_{11}$ on the open dense subset where $\tilde{\mathfrak{q}}_{11}$ is invertible we get
\begin{equation}
    \tilde{\mathfrak{p}}_{11} = t\tilde{\mathfrak{q}}_{21}^{-1}\tilde{\mathfrak{p}}_{21}\tilde{\mathfrak{q}}_{11}^{-1} - \tilde{\mathfrak{q}}_{11}\tilde{\mathfrak{q}}_{21}^{-1}\tilde{\mathfrak{p}}_{21} + \frac{1}{2}(\ri g_2 - \nu)\tilde{\mathfrak{q}}_{11}^{-1}.
\end{equation}
We can take the following coordinates on the reduction
\begin{equation}
    \tilde{Q} = \tilde{\mathfrak{q}}_{11},\quad \tilde{P} = 2t\tilde{\mathfrak{q}}_{21}^{-1}\tilde{\mathfrak{p}}_{21}\tilde{\mathfrak{q}}_{11}^{-1}.
\end{equation}
Then we have a section from $\mathbb{M}_{\alpha}$ to the intersection $M_{\alpha}^{\bar{w}}\cap \{m_2 =\ri g_2\mathbf{1}_{n\times n}\}$ which maps
\begin{equation}\label{Non-linear_cases_reduction_section}
    s: \left(\tilde{P}, \tilde{Q}, t\right) \mapsto \left(\tilde{\mathfrak{p}}_{11} = \frac{\tilde{P}}{2} - \frac{\tilde{Q}\tilde{P}\tilde{Q}}{2t} + \frac{1}{2}(\ri g_2 -\nu)\tilde{Q}^{-1} , \tilde{\mathfrak{q}}_{11} = \tilde{Q}, \tilde{\mathfrak{p}}_{21} = \frac{\tilde{P}\tilde{Q}}{2t},\tilde{\mathfrak{q}}_{21} = 1 , t\right).
\end{equation}
\paragraph{Step 4.} Hamiltonian on the reduction is
\begin{equation}
    s^*\left(H +  \mathrm{Tr}\left(\tilde{\mathfrak{p}}_{21}\tilde{\mathfrak{q}}_{21}^{-1}\right)\right) = s^* \left(H\right) + \mathrm{Tr}\left(\frac{\tilde{P}\tilde{Q}}{2t}\right).
\end{equation}

\subsubsection{\texorpdfstring{PIII$\mathrm{\big(D_8^{(1)}\big)}$}{3D8} to \texorpdfstring{PIII$\mathrm{\big(D_6^{(1)}\big)}$}{3D6}}
\label{sssec:IIID8_IIID6}
\paragraph{Steps 1,2,3.} We consider matrix $2n\times 2n$ PIII$\mathrm{\big(D_{8}^{(1)}\big)}$. In this case $\nu = \frac{1}{2}$.
\paragraph{Step 4.} Here we start from coordinates 
\begin{equation}
	\breve{P} = -\frac{1}{2}\left(\frac{\tilde{Q}}{\sqrt{t}} - 1\right),\; \breve{Q} = 2\sqrt{t}\left(\tilde{P}+\left(\ri g_2 -\frac{1}{2}\right)\tilde{Q}^{-1}\right),\; s = 16\sqrt{t},
\end{equation}
in which we get Hamiltonian
\begin{equation}
	sH(\breve{P}, \breve{Q}; s) = \mathrm{Tr}\left(\breve{P}\breve{Q}\breve{P}\breve{Q} - \left(\breve{Q}^2 - 2\ri g_2\breve{Q} - s\right)\breve{P} - \ri g_2\breve{Q} \right).
\end{equation}
Then after change of variables
\begin{equation}
	P = s^{[\breve{P}, \breve{Q}]}\breve{P}s^{-[\breve{P}, \breve{Q}]},\quad Q = s^{[\breve{P}, \breve{Q}]}\breve{Q}s^{-[\breve{P}, \breve{Q}]},
\end{equation}
we get
\begin{equation}
	sH(P,Q; s) = \mathrm{Tr}\left(P^2Q^2 - \left(Q^2 - 2\ri g_2 Q - s\right)P - \ri g_2 Q \right),
\end{equation}
which is the Hamiltonian of matrix $n\times n$ PIII$\mathrm{\big(D_{6}^{(1)}\big)}\big(1-\ri g_2,  \ri g_2 , 1-\ri g_2, \ri g_2\big)$.

\subsubsection{\texorpdfstring{PIII$\mathrm{\big(D_6^{(1)}\big)}$}{3D6} to \texorpdfstring{PV}{5}}
\label{sssec:IIID6_V}
\paragraph{Steps 1,2,3.} We consider matrix $2n\times 2n$ PIII$\mathrm{\big(D_{6}^{(1)}\big)}$ with $\beta_0=\beta_1= \frac{1}{2},\alpha_1=\epsilon$. In this case $\nu=\epsilon$.
\paragraph{Step 4.} Here we start from coordinates
\begin{equation}\label{D6_D5_int}
    \breve{P} = 2\sqrt{t}\left(\tilde{P}+(\ri g_2 -\epsilon)\tilde{Q}^{-1}\right),\quad \breve{Q} = \frac{1}{2}\left(\frac{\tilde{Q}}{\sqrt{t}} + 1\right),\qquad s = -8\sqrt{t},
\end{equation}
we get a system with Hamiltonian
\begin{equation}
    sH(\breve{P}, \breve{Q}; s) = \mathrm{Tr}\left(\breve{P}(\breve{P}{+}s)\breve{Q}(\breve{Q}{-}1) + (1{-}2\ri g_2)\breve{P}\breve{Q} + \left(\frac{2\ri g_2{-}1}{2}\right)\breve{P} + (\epsilon{-}\ri g_2)s\breve{Q} - \frac{[\breve{P}, \breve{Q}]^2}4\right).
\end{equation}
After change of coordinates
\begin{equation}
    P = s^{-\frac{1}{2}[ \breve{P},\breve{Q}]}\breve{P} s^{\frac{1}{2}[\breve{P},\breve{Q}]},\quad Q = s^{-\frac{1}{2}[\breve{P},\breve{Q}]}\breve{Q} s^{\frac{1}{2}[\breve{P},\breve{Q}]},
\end{equation}
we get
\begin{equation}
    sH(P,Q; s) = \mathrm{Tr}\left(P(P+s)Q(Q - 1) + (- 2\ri g_2)PQ + \ri g_2 P + (\epsilon -\ri g_2)sQ\right),
\end{equation}
which is the Hamiltonian of matrix $n\times n$ PV$\big(1-\ri g_2 - \epsilon, \ri g_2, \epsilon - \ri g_2, \ri g_2\big)$.

\subsubsection{\texorpdfstring{PVI}{6} to \texorpdfstring{PVI}{6}}
\label{sssec:VI_VI}
\paragraph{Steps 1,2,3.}  We consider matrix $2n\times 2n$ PVI with $\alpha_0=\alpha_3 =\epsilon_0,\; \alpha_1=\alpha_4 = \epsilon_1$. In this case $\nu = \alpha_2 = \frac{1}{2} - \epsilon_0 - \epsilon_1$.
\paragraph{Step 4.} We start from coordinates 
\begin{equation}
\breve{P} = 2\sqrt{t}\left(\tilde{P} + \left(\epsilon_0 + \epsilon_1 +\ri g_2 -\frac{1}{2}\right)\tilde{Q}^{-1}\right),\quad \breve{Q} = \frac{1}{2}\left(\frac{\tilde{Q}}{\sqrt{t}} + 1\right),\quad s = \frac{1}{2} + \frac{1}{4}\left(\sqrt{t} + \frac{1}{\sqrt{t}}\right),
\end{equation}
in which we get Hamiltonian
\begin{multline}
    s(s{-}1)H(\breve{P},\breve{Q}; s) = \mathrm{Tr}\Bigg(\breve{P} \breve{Q} (\breve{Q} - 1) \breve{P} (\breve{Q} - s) + \left(\ri g_2(\breve{Q}{-}1)(\breve{Q}{-}s) + \ri g_2\breve{Q}(\breve{Q}{-} s)  + (2\epsilon_0{-}1)\breve{Q}(\breve{Q}{-}1)\right)\breve{P} +\\ \left.+ \left(\left(
    \epsilon_0 - \ri g_2 + \frac{1}{2}\right)^2-\epsilon_1^2\right)\breve{Q} - \frac{\sqrt{s(s-1)}}{2}[\breve{P}, \breve{Q}]^2\right).
\end{multline}
After the last substitution
\begin{equation}
    P = e^{-2\cosh^{-1}\left(\sqrt{s}\right)[\breve{P}, \breve{Q}]}\breve{P}e^{2\cosh^{-1}\left(\sqrt{s}\right)[\breve{P}, \breve{Q}]},\quad Q = e^{-2\cosh^{-1}\left(\sqrt{s}\right)[\breve{P}, \breve{Q}]}\breve{P}e^{2\cosh^{-1}\left(\sqrt{s}\right)[\breve{P}, \breve{Q}]},
\end{equation}
we get
\begin{multline}
    s(s-1)H(P,Q; s) = \mathrm{Tr}\Bigg(P Q (Q - 1) P (Q - s) + \left(\ri g_2(Q - 1)(Q - s) +\ri g_2Q(Q - s)  +\right.\\ \left.\left.+ (2\epsilon_0 - 1)Q(Q - 1)\right)P  + \left(\left(
    \epsilon_0 - \ri g_2 + \frac{1}{2}\right)^2-\epsilon_1^2\right)Q\right),
\end{multline}
which is the Hamiltonian of matrix $n\times n$ PVI$\big(2\epsilon_0, 2\epsilon_1,\frac{1}{2}-\ri g_2 - \epsilon_0 - \epsilon_1, \ri g_2, \ri g_2 \big)$.

\subsection{Special $C_2\times C_2$ cases}\label{ssec:special_C2xC2}
Let $G$ be a group consisting of B\"acklund transformations which preserve a certain matrix Painlev\'e system and act trivially on time variable. It was shown in \cite{TOS05} that such group $G$ consists of transformations coming from automorphisms of diagram (Up to overall conjugation of $G$). From the description of transformations corresponding to automorphisms in Sec. \ref{sec:Backlund} it follows that possible $G$'s are either cyclic, or isomorphic to $C_2\times C_2$.

Twisting generators by $\mathrm{GL}_{n|G|}\left(\mathbb{C}\right)$ action we get the group of symmetries $\bar{G}$, and the submanifold $M_{\alpha}^{\bar{G}}$ which is preserved by the dynamics. Condition of $\bar{G}$-invariance in case of cyclic $G$ simplifies to invariance under the action of generator $\bar{w}$. These cases are described by Theorem \ref{thm:gen_constr}.

There are only two cases of non-cyclic $G$. For them there is a construction of reduction similar to one described by Theorem \ref{thm:gen_constr}. The input and output are given in the following table
\begin{table}[H]\label{table:automorphisms_C2xC2}
\begin{center}
\begin{tabular}{|c|c|c|c|c|c|c|c|}
\hline
\textnumero & Equation & Image  & Generators  & $ q $ & $p$ & Section \\
\hline
1 & PIII$\mathrm{\big(D_6^{(1)}\big)}$ & PIII$\mathrm{\big(D_6^{(1)}\big)}$  &  \begin{tabular}{c}
$\pi \circ \pi'$ \\
$\pi'$ \\
\end{tabular}&\begin{tabular}{c}
$-q$ \\
$tq^{-1}$ \\
\end{tabular} & \begin{tabular}{c}
$1-p - t q^{-2}$ \\
$-t^{-1}q \left(pq + \frac{1}{2}\right)$ \\
\end{tabular} & \ref{sssec:IIID6_IIID6} \\
\hline
2 & PVI & PVI   & \begin{tabular}{c}
$\pi_1$ \\
$\pi_2$ \\
\end{tabular}&\begin{tabular}{c}
$tq^{-1}$ \\
$t(q - 1)(q - t)^{-1}$ \\
\end{tabular} & \begin{tabular}{c}
$-t^{-1}q(pq + \alpha_2)$ \\
$\frac{(q-t)\left(p(q-t) + \alpha_2\right)}{t(1-t)}$ \\
\end{tabular} & \ref{sssec:VI_VI_4} \\
\hline
\end{tabular}
\caption{Input and output for $C_2\times C_2$ cases.}
\end{center}
\end{table}

Let us use notations $w_{+}, w_{-}$ for generators of $G$. We take twists $S_{+} = \mathrm{diag}\left(\mathbf{1}_{n\times n}, \mathbf{1}_{n\times n}, -\mathbf{1}_{n\times n}, -\mathbf{1}_{n\times n}\right)$, $S_{-} = \mathrm{diag}\left(\mathbf{1}_{n\times n}, -\mathbf{1}_{n\times n}, \mathbf{1}_{n\times n}, -\mathbf{1}_{n\times n}\right)$. Then twisted generators are $\bar{w}_{i} = w_{i}\circ Ad_{S_{i}}$.

\subsubsection{\texorpdfstring{PIII$\mathrm{\big(D_6^{(1)}\big)}$}{3D6} to \texorpdfstring{PIII$\mathrm{\big(D_6^{(1)}\big)}$}{3D6}}
\label{sssec:IIID6_IIID6}
\paragraph{Step 1.} Consider matrix $4n\times 4n$ PIII$\mathrm{\big(D_6^{(1)}\big)}$ with parameters $\alpha_1 = \beta_1 = \frac{1}{2}$. In this case we take $w_{+} = \pi'$ and $w_{-} = \pi \circ \pi'$.

Since $M_{\alpha}^{\bar{G}} = M_{\alpha}^{\bar{w}_{+}}\cap M_{\alpha}^{\bar{w}_{-}}$ let us start from $M_{\alpha}^{\bar{w}_{+}}$. Manifold $M_{\alpha}^{\bar{w}_{+}}$ is defined by equations \eqref{nonlinear_leaf_fixed_points} with $\nu = \frac{1}{2}$. So we have the solution on the dense open subset
\begin{equation}
q = \begin{pmatrix}
    \breve{\mathfrak{q}}_{11} & \breve{\mathfrak{q}}_{12} & * & *\\ 
    \breve{\mathfrak{q}}_{21} & \breve{\mathfrak{q}}_{22} & * & *\\ 
    \breve{\mathfrak{q}}_{31} & \breve{\mathfrak{q}}_{32} & * & *\\
    \breve{\mathfrak{q}}_{41} & \breve{\mathfrak{q}}_{42} & * & *\\  
    \end{pmatrix},\;\;
p = \begin{pmatrix}
    \breve{\mathfrak{p}}_{11} & \breve{\mathfrak{p}}_{12} & * & *\\ 
    \breve{\mathfrak{p}}_{21} & \breve{\mathfrak{p}}_{22} & * & *\\ 
    \breve{\mathfrak{p}}_{31} & \breve{\mathfrak{p}}_{32} & * & *\\
    \breve{\mathfrak{p}}_{41} & \breve{\mathfrak{p}}_{42} & * & *\\  
    \end{pmatrix}.
\end{equation}
Here $\breve{\mathfrak{q}}_{ij}, \breve{\mathfrak{p}}_{ij}$ are $n\times n$ blocks and $*$'s are defined in terms of $\breve{\mathfrak{q}}_{ij}$ and $\breve{\mathfrak{p}}_{ij}$ by \eqref{nonlinear_case_solution}. Note that for this case in equations \eqref{nonlinear_case_solution} blocks have size $2n\times 2n$.

Darboux coordinates on $M_{\alpha}^{\bar{w}_{+}}$ are given by \eqref{nonlinear_darboux_coordinates}. Note that one half of these coordinates are simply $\{\breve{\mathfrak{q}}_{ij}\}_{1\leq i\leq 4,\; 1\leq j\leq 2}$, the other half are coordinates conjugate to them, defined by rather complicated formulas following from \eqref{nonlinear_darboux_coordinates}.

Now let us impose invariance under $\bar{w}_{-}$ to obtain $M_{\alpha}^{\bar{G}}$ as the submanifold of $M_{\alpha}^{\bar{w}_{+}}$.
\begin{equation}
q = - S_2 q S_2^{-1},\;\; p = S_2 \left(-p + 1 - t q^{-2}\right) S_2^{-1}.
\end{equation}
These equations are solved by
\begin{equation}\label{3D6_C2xC2_invariant_solution}
\begin{aligned}
\breve{\mathfrak{q}}_{11}& = \breve{\mathfrak{q}}_{22} = \breve{\mathfrak{q}}_{31} = \breve{\mathfrak{q}}_{42} = 0,\\
\breve{\mathfrak{p}}_{11} = 1 - \frac{\breve{\mathfrak{q}}_{12}\breve{\mathfrak{q}}_{21}}{t},\;\; \breve{\mathfrak{p}}_{22} &= 1 -\frac{\breve{\mathfrak{q}}_{21}\breve{\mathfrak{q}}_{12}}{t},\;\; \breve{\mathfrak{p}}_{31} = \frac{\breve{\mathfrak{q}}_{32}\breve{\mathfrak{q}}_{21}}{t},\;\; \breve{\mathfrak{p}}_{42} = \frac{\breve{\mathfrak{q}}_{41}\breve{\mathfrak{q}}_{12}}{t}.
\end{aligned}
\end{equation}
So matrix coordinates on $M_{\alpha, t}^{\bar{G}}$ are $(\breve{\mathfrak{q}}_{12}, \breve{\mathfrak{q}}_{21}, \breve{\mathfrak{q}}_{32}, \breve{\mathfrak{q}}_{41}, \breve{\mathfrak{p}}_{12}, \breve{\mathfrak{p}}_{21}, \breve{\mathfrak{p}}_{32}, \breve{\mathfrak{p}}_{41})$.

From the consideration above it follows that Darboux coordinates on $M_{\alpha, t}^{\bar{G}}$ are $\breve{\mathfrak{q}}_{12}, \breve{\mathfrak{q}}_{21}, \breve{\mathfrak{q}}_{32}, \breve{\mathfrak{q}}_{41}$ and conjugate to them which are restrictions of ones conjugate to them on $M_{\alpha, t}^{\bar{w}_{+}}$ to $M_{\alpha, t}^{\bar{G}}$. For example, to obtain matrix coordinate conjugate to $\breve{\mathfrak{q}}_{21}$ one should take upper right $n\times n$ block of $\mathfrak{p}_{11}$ from formula \eqref{nonlinear_darboux_coordinates} and then restrict it to \eqref{3D6_C2xC2_invariant_solution}.  Darboux coordinates on $M_{\alpha, t}^{\bar{G}}$ are 
\begin{multline}\label{Darboux_special_3D6}
	\left(\hat{\mathfrak{p}}_{21}, \hat{\mathfrak{q}}_{12}, \hat{\mathfrak{p}}_{12}, \hat{\mathfrak{q}}_{21}, \hat{\mathfrak{p}}_{23}, \hat{\mathfrak{q}}_{32}, \hat{\mathfrak{p}}_{14}, \hat{\mathfrak{q}}_{41}\right) = \left( 2(\breve{\mathfrak{p}}_{21} + \breve{\mathfrak{q}}_{21}\breve{\mathfrak{q}}_{41}^{-1}\breve{\mathfrak{p}}_{41}), \breve{\mathfrak{q}}_{12}, 2(\breve{\mathfrak{p}}_{12} + \breve{\mathfrak{q}}_{12}\breve{\mathfrak{q}}_{32}^{-1}\breve{\mathfrak{p}}_{32}), \breve{\mathfrak{q}}_{21},\right.\\\left. 2\left(t \breve{\mathfrak{q}}_{32}^{-1}\breve{\mathfrak{p}}_{32} - \breve{\mathfrak{p}}_{21}\breve{\mathfrak{q}}_{12} - \breve{\mathfrak{q}}_{21}\breve{\mathfrak{q}}_{41}^{-1}\breve{\mathfrak{p}}_{41}\breve{\mathfrak{q}}_{12} -\frac{1}{4}\right)\breve{\mathfrak{q}}_{32}^{-1}, \breve{\mathfrak{q}}_{32}, 2\left( t \breve{\mathfrak{q}}_{41}^{-1}\breve{\mathfrak{p}}_{41} - \breve{\mathfrak{p}}_{12}\breve{\mathfrak{q}}_{21}- \breve{\mathfrak{q}}_{12}\breve{\mathfrak{q}}_{32}^{-1}\breve{\mathfrak{p}}_{32}\breve{\mathfrak{q}}_{21} - \frac{1}{4}\right)\breve{\mathfrak{q}}_{41}^{-1}, \breve{\mathfrak{q}}_{41}	\right) .
\end{multline}

\paragraph{Step 2.} We have the action of block diagonal matrices on $M_{\alpha}^{\bar{G}}$   which maps
\begin{equation}
h=\mathrm{diag}\left(h_1, h_2, h_3, h_4\right): ( \{\breve{\mathfrak{q}}_{ij}\},  \{\breve{\mathfrak{p}}_{ij}\}) \mapsto (\{h_i \breve{\mathfrak{q}}_{ij} h_{j}^{-1}\}, \{ h_i \breve{\mathfrak{p}}_{ij} h_{j}^{-1} \})
\end{equation}
Moment map of this action is defined by blocks of block--diagonal matrix $[p, q]$
\begin{equation}
\begin{aligned}
m_1 &= 2 (t \breve{\mathfrak{q}}_{41}^{-1}\breve{\mathfrak{p}}_{41}- \breve{\mathfrak{q}}_{12}\breve{\mathfrak{p}}_{21}- \breve{\mathfrak{q}}_{12}\breve{\mathfrak{q}}_{21}\breve{\mathfrak{q}}_{41}^{-1}\breve{\mathfrak{p}}_{41})-\frac{1}{2},\\
m_2 &= 2 (t \breve{\mathfrak{q}}_{32}^{-1}\breve{\mathfrak{p}}_{32}- \breve{\mathfrak{q}}_{21}\breve{\mathfrak{p}}_{12}- \breve{\mathfrak{q}}_{21}\breve{\mathfrak{q}}_{12}\breve{\mathfrak{q}}_{32}^{-1}\breve{\mathfrak{p}}_{32})-\frac{1}{2},\\
m_3 &= 2 (-t \breve{\mathfrak{p}}_{32}+ \breve{\mathfrak{q}}_{32}\breve{\mathfrak{p}}_{21}\breve{\mathfrak{q}}_{12} + \breve{\mathfrak{q}}_{32}\breve{\mathfrak{q}}_{21}\breve{\mathfrak{q}}_{41}^{-1}\breve{\mathfrak{p}}_{41}\breve{\mathfrak{q}}_{12})\breve{\mathfrak{q}}_{32}^{-1}+\frac{1}{2},\\
m_{4} &= 2(-t\breve{\mathfrak{p}}_{41} + \breve{\mathfrak{q}}_{41}\breve{\mathfrak{p}}_{12}\breve{\mathfrak{q}}_{21} + \breve{\mathfrak{q}}_{41}\breve{\mathfrak{q}}_{12}\breve{\mathfrak{q}}_{32}^{-1}\breve{\mathfrak{p}}_{32}\breve{\mathfrak{q}}_{21})\breve{\mathfrak{q}}_{41}^{-1} + \frac{1}{2}.
\end{aligned}
\end{equation}
\paragraph{Step 3.} Let us perform Hamiltonian reduction with respect to $\mathrm{GL}_{n}^3\left(\mathbb{C}\right) = \{\mathrm{diag}(1, h_2, h_3, h_4) | h_2, h_3, h_4 \in \mathrm{GL}_{n}\left(\mathbb{C}\right)\}$. We fix the moment map value as follows $m_2 = \ri g_2\mathbf{1}_{n\times n},\;\; m_3 = \ri g_3\mathbf{1}_{n\times n},\;\;  m_4 = \ri g_4\mathbf{1}_{n\times n}.$

Darboux coordinates on the reduction are $(\tilde{P}, \tilde{Q}) = \left(2(\breve{\mathfrak{q}}_{21}^{-1}\breve{\mathfrak{p}}_{21} + \breve{\mathfrak{q}}_{41}^{-1}\breve{\mathfrak{p}}_{41}), \breve{\mathfrak{q}}_{12}\breve{\mathfrak{q}}_{21} \right).$
\paragraph{Step 4.} As in non-linear cases restriction on the set of points invariant under $\pi'$ shifts the Hamiltonian by $\mathrm{Tr}\left(\tilde{\mathfrak{p}}_{21}\tilde{\mathfrak{q}}_{21}^{-1}\right)$, which is in our case $\mathrm{Tr}\left(\breve{\mathfrak{p}}_{32}\breve{\mathfrak{q}}_{32}^{-1} + \breve{\mathfrak{p}}_{41}\breve{\mathfrak{q}}_{41}^{-1}\right)$. Then it remains to restrict the Hamiltonian obtained on the set of points invariant under $\pi\circ \pi'$ and satisfying moment equations and substitute section from the reduction. Then we get
\begin{equation}
tH(\tilde{P}, \tilde{Q}; t) = \mathrm{Tr}\left(\tilde{P}^2\tilde{Q}^2 -\left(\tilde{P}^2 + (\ri g_4 + 2\ri g_3 + \ri g_{2} - 1)\tilde{P} - 4\right)\tilde{Q}  + (\ri g_2 + \ri g_3) \tilde{P}\right).
\end{equation}
Then after change of variables
\begin{equation}
Q = - t\tilde{P},\quad P = \frac{1}{t}\tilde{Q},\qquad s = -\frac{t}{4},
\end{equation}
we get
\begin{equation}
sH(P, Q; s) = \mathrm{Tr}\left( P^2 Q^2 - \left(Q^2 -\left( \left( \ri g_2+ \ri g_3\right)+ \left(\ri g_3+ \ri g_4\right)\right)
   Q-s\right)P- \left(\ri g_2+ \ri g_3\right) Q\right),
\end{equation}
which is the Hamiltonian of matrix $n\times n$ PIII$\mathrm{\big(D_{6}^{(1)}\big)}\big(1- \ri g_2+ \ri g_3, \ri g_2 + \ri g_3, 1 -  \ri g_3 - \ri g_4, \ri g_3 + \ri g_4\big)$.
\subsubsection{\texorpdfstring{PVI}{6} to \texorpdfstring{PVI}{6}}
\label{sssec:VI_VI_4}
\paragraph{Step 1.} Consider matrix $4n\times 4n$ PVI  with parameters $\alpha_0 = \alpha_1 = \alpha_3 = \alpha_4 = \epsilon$. In this case we take $w_{+} = \pi_1$ and $w_{-} = \pi_2$.

We can obtain $M_{\alpha}^{\bar{w}_{+}}$. This manifold is defined by equations \eqref{nonlinear_leaf_fixed_points} with $\nu = \alpha_2$. So we have the solution on an dense open subset
\begin{equation}
q = \begin{pmatrix}
    \breve{\mathfrak{q}}_{11} & \breve{\mathfrak{q}}_{12} & * & *\\ 
    \breve{\mathfrak{q}}_{21} & \breve{\mathfrak{q}}_{22} & * & *\\ 
    \breve{\mathfrak{q}}_{31} & \breve{\mathfrak{q}}_{32} & * & *\\
    \breve{\mathfrak{q}}_{41} & \breve{\mathfrak{q}}_{42} & * & *\\  
    \end{pmatrix},\;\;
p = \begin{pmatrix}
    \breve{\mathfrak{p}}_{11} & \breve{\mathfrak{p}}_{12} & * & *\\ 
    \breve{\mathfrak{p}}_{21} & \breve{\mathfrak{p}}_{22} & * & *\\ 
    \breve{\mathfrak{p}}_{31} & \breve{\mathfrak{p}}_{32} & * & *\\
    \breve{\mathfrak{p}}_{41} & \breve{\mathfrak{p}}_{42} & * & *\\  
    \end{pmatrix}.
\end{equation}
Here $\breve{\mathfrak{q}}_{ij}, \breve{\mathfrak{p}}_{ij}$ are $n\times n$ blocks and $*$'s are defined in terms of $\breve{\mathfrak{q}}_{ij}$ and $\breve{\mathfrak{p}}_{ij}$ by \eqref{nonlinear_case_solution}. Note that for this case in equations \eqref{nonlinear_case_solution} blocks $\tilde{\mathfrak{p}}_{ij}, \tilde{\mathfrak{q}}_{ij}$ have size $2n\times 2n$.

Let us obtain $M_{\alpha}^{\bar{G}} = M_{\alpha}^{\bar{w}_{+}}\cap M_{\alpha}^{\bar{w}_{-}} \subset M_{\alpha}^{\bar{w}_{+}}$. We have to solve on $M_{\alpha}^{\bar{w}_{+}}$
\begin{equation}
q-t =t(t-1) S_2 (q-t)^{-1} S_2^{-1},\;\;\;\; p = -\frac{1}{t(t-1)} S_2 \left((q-t)(p(q-t)+ \alpha_2)\right) S_2^{-1}.
\end{equation}
On the open dense subset where $\breve{\mathfrak{q}}_{21}, \breve{\mathfrak{q}}_{31}, \breve{\mathfrak{q}}_{41}$ are invertible these equations can be solved by
\begin{equation}\label{C2xC2_P6_fixed_points}
\begin{aligned}
\breve{\mathfrak{q}}_{12}& = (\breve{\mathfrak{q}}_{11} - t)(\breve{\mathfrak{q}}_{11} - 1)\breve{\mathfrak{q}}_{21}^{-1},\;\;\;\; \breve{\mathfrak{q}}_{22} = \breve{\mathfrak{q}}_{21}\breve{\mathfrak{q}}_{11}\breve{\mathfrak{q}}_{21}^{-1},\\
\breve{\mathfrak{q}}_{32}& = \breve{\mathfrak{q}}_{31}(\breve{\mathfrak{q}}_{11} - t)\breve{\mathfrak{q}}_{21}^{-1},\;\;\;\; \breve{\mathfrak{q}}_{42} = \breve{\mathfrak{q}}_{41}(\breve{\mathfrak{q}}_{11} - 1)\breve{\mathfrak{q}}_{21}^{-1},\\
\breve{\mathfrak{p}}_{11}& = \breve{\mathfrak{p}}_{22} = 0,\;\;\;\; \breve{\mathfrak{p}}_{31} = -\breve{\mathfrak{q}}_{31}(\breve{\mathfrak{q}}_{21}^{-1}\breve{\mathfrak{p}}_{21} + \breve{\mathfrak{q}}_{41}^{-1}\breve{\mathfrak{p}}_{41}),\\
\breve{\mathfrak{p}}_{12}& = -\left((\breve{\mathfrak{q}}_{11} - 1)\breve{\mathfrak{q}}_{31}^{-1}\breve{\mathfrak{p}}_{32} + (\breve{\mathfrak{q}}_{11} - t)\breve{\mathfrak{q}}_{41}^{-1}\breve{\mathfrak{p}}_{42}\right).
\end{aligned}
\end{equation}
So matrix coordinates on $M_{\alpha, t}^{\bar{G}}$ are $\breve{\mathfrak{q}}_{11}, \breve{\mathfrak{q}}_{21}, \breve{\mathfrak{q}}_{31}, \breve{\mathfrak{q}}_{41}, \breve{\mathfrak{p}}_{21}, \breve{\mathfrak{p}}_{41}, \breve{\mathfrak{p}}_{32}, \breve{\mathfrak{p}}_{42}$. Let us denote embedding by $\iota:M_{\alpha}^{\bar{G}}\rightarrow M_{\alpha}$. Note that on $M_{\alpha}^{\bar{G}}$ blocks of $q$ depend only of $\breve{\mathfrak{q}}_{11}, \breve{\mathfrak{q}}_{21}, \breve{\mathfrak{q}}_{31}, \breve{\mathfrak{q}}_{41}, t$, hence
\begin{equation}\label{C2xC2_P6_1form_restriction}
\iota^{*}\Theta = \mathrm{Tr}\left(\hat{\mathfrak{p}}_{11}\rd \breve{\mathfrak{q}}_{11} + \hat{\mathfrak{p}}_{12}\rd \breve{\mathfrak{q}}_{21} + \hat{\mathfrak{p}}_{13}\rd \breve{\mathfrak{q}}_{31} + \hat{\mathfrak{p}}_{14}\rd \breve{\mathfrak{q}}_{41}\right) - F\rd t,
\end{equation}
for certain $\hat{\mathfrak{p}}_{11}, \hat{\mathfrak{p}}_{12}, \hat{\mathfrak{p}}_{13}, \hat{\mathfrak{p}}_{14}, F$. One can calculate all of them, but we will need only $\hat{\mathfrak{p}}_{11}, F$
\begin{align}
\hat{\mathfrak{p}}_{11}& = -2\left(t\breve{\mathfrak{q}}_{21}^{-1}\breve{\mathfrak{p}}_{21} + (t-1)\breve{\mathfrak{q}}_{41}^{-1}\breve{\mathfrak{p}}_{41}\right),\\
F& = \frac{1}{t-1}\mathrm{Tr}\left((1-\breve{\mathfrak{q}}_{11})\breve{\mathfrak{q}}_{21}^{-1}\breve{\mathfrak{p}}_{21} + (t-1)\breve{\mathfrak{q}}_{41}^{-1}\breve{\mathfrak{p}}_{41} + \breve{\mathfrak{q}}_{41}^{-1}\breve{\mathfrak{p}}_{42}\breve{\mathfrak{q}}_{21} + \breve{\mathfrak{q}}_{31}^{-1}\breve{\mathfrak{p}}_{32}\breve{\mathfrak{q}}_{21} + \alpha_2\right).
\end{align}
\paragraph{Step 2.} We have a Hamiltonian action of $\mathrm{GL}_{n}^{4}\left(\mathbb{C}\right)$ on $M_{\alpha, t}^{\bar{G}}$
\begin{equation}
h=\mathrm{diag}(h_1, h_2, h_3, h_4): (\{\breve{\mathfrak{q}}_{ij}\}, \{\breve{\mathfrak{p}}_{ij}\}) \mapsto (\{h_i \breve{\mathfrak{q}}_{ij} h_{j}^{-1}\}, \{ h_i \breve{\mathfrak{p}}_{ij} h_{j}^{-1} \}).
\end{equation}
The moment map of this action is given by blocks of block--diagonal matrix $[p, q]$
\begin{equation}
\begin{aligned}
m_1 &= 2t(\breve{\mathfrak{q}}_{11} - 1)\breve{\mathfrak{q}}_{21}^{-1}\breve{\mathfrak{p}}_{21} + 2(t-1)\breve{\mathfrak{q}}_{11}\breve{\mathfrak{q}}_{41}^{-1}\breve{\mathfrak{p}}_{41} - \alpha_2, \\
m_2 &= 2t\breve{\mathfrak{q}}_{21}\breve{\mathfrak{q}}_{41}^{-1}\breve{\mathfrak{p}}_{42} - 2\breve{\mathfrak{q}}_{21}\breve{\mathfrak{q}}_{31}^{-1}\breve{\mathfrak{p}}_{32} - \alpha_2, \\
m_3 &= 2 \breve{\mathfrak{p}}_{32}\breve{\mathfrak{q}}_{21}\breve{\mathfrak{q}}_{31}^{-1} - 2t\breve{\mathfrak{q}}_{31}(\breve{\mathfrak{q}}_{21}^{-1}\breve{\mathfrak{p}}_{21} + \breve{\mathfrak{q}}_{41}^{-1}\breve{\mathfrak{p}}_{41})(\breve{\mathfrak{q}}_{11}-1)\breve{\mathfrak{q}}_{31}^{-1} + \alpha_2,\\
m_{4} &= 2\breve{\mathfrak{p}}_{41}(\breve{\mathfrak{q}}_{11} - t)\breve{\mathfrak{q}}_{41}^{-1} + 2t\breve{\mathfrak{p}}_{42}\breve{\mathfrak{q}}_{21}\breve{\mathfrak{q}}_{41}^{-1} + \alpha_2.
\end{aligned}
\end{equation}

\paragraph{Step 3.}  Let us perform Hamiltonian reduction with respect to $\mathrm{GL}_{n}^3\left(\mathbb{C}\right) = \{\mathrm{diag}(1, h_2, h_3, h_4) | h_2, h_3, h_4 \in \mathrm{GL}_{n}\left(\mathbb{C}\right)\}$. We fix the moment map value as follows $m_2 = \ri g_2\mathbf{1}_{n\times n},\;\; m_3 = \ri g_3\mathbf{1}_{n\times n},\;\;  m_4 = \ri g_4\mathbf{1}_{n\times n}.$

Darboux coordinates on the reduction are 
\begin{equation}
(\tilde{P}, \tilde{Q}) = \left(\hat{\mathfrak{p}}_{11}, \breve{\mathfrak{q}}_{11}\right) = \left(-2\left(t\breve{\mathfrak{q}}_{21}^{-1}\breve{\mathfrak{p}}_{21} + (t-1)\breve{\mathfrak{q}}_{41}^{-1}\breve{\mathfrak{p}}_{41}\right), \breve{\mathfrak{q}}_{11}\right).
\end{equation}

\paragraph{Step 4.} From \eqref{C2xC2_P6_1form_restriction} it follows that the dynamics on $M_{\alpha}^{\bar{G}}$ is given by the Hamiltonian $\iota^{*}(H) + F$.

The Hamiltonian on the reduction in the coordinates
\begin{equation}
\breve{P} = \frac{\tilde{P}}{t-1},\;\; \breve{Q} = (t-1)\tilde{Q} + 1,\;\; s = \frac{1}{1-t},
\end{equation}
is given by 
\begin{multline}
s(s-1)H(\breve{P}, \breve{Q}; s) = \mathrm{Tr}\left(\breve{P}\breve{Q}(\breve{Q}-s)\breve{P}(\breve{Q}-1) - \left((\ri g_3 + \ri g_4 - 1)\breve{Q}(\breve{Q}-1) + (\ri g_2 +\ri g_4)\breve{Q}(\breve{Q}-s)+\right.\right. \\ \left.\left.+ (\ri g_2 + \ri g_3) (\breve{Q}-1)(\breve{Q}-s) \right) \breve{P} + \frac{1}{4} \left(1 {-}2 \ri g_2{-}2 \ri g_3{-}2 \ri g_4{-}4 \epsilon\right) \left(1{-}2 \ri
   g_2{-}2 \ri g_3{-}2 \ri g_4+4 \epsilon\right) \breve{Q} \right).
\end{multline}
After change of coordinates
\begin{equation}
P = s^{-[\breve{P}, \breve{Q}]}\breve{P}s^{[\breve{P}, \breve{Q}]}, \;\; Q = s^{-[\breve{P}, \breve{Q}]}\breve{Q}s^{[\breve{P}, \breve{Q}]},
\end{equation}
we get
\begin{multline}
s(s-1)H(P, Q; s)=\mathrm{Tr}\Big(PQ(Q-1)P(Q-s) - \left( (\ri g_3 + \ri g_4 - 1)Q(Q-1) + (\ri g_2 + \ri g_4)Q(Q-s)\right. + \\ + \left. (\ri g_2 + \ri g_3) (Q-1)(Q-s) \right) P +  \frac{1}{4} \left(1 {-}2 \ri g_2{-}2 \ri g_3{-}2 \ri g_4{-}4 \epsilon \right) \left(1 {-}2 \ri g_2{-}2 \ri g_3{-}2 \ri g_4+4 \epsilon\right) Q \Big),
\end{multline}
which is the Hamiltonian of matrix $n\times n$ PVI$\big(\ri g_3{+}\ri g_4, 4\epsilon, \frac{1}{2}\left(1 - 4\epsilon - 2(\ri g_2{+}\ri g_3{+}\ri g_4) \right),\ri g_2{+}\ri g_4, \ri g_2{+}\ri g_3\big)$.

\section{Application to Calogero--Painlev\'e systems}
\label{sec:CP}
\subsection{From Matrix Painlev\'e to Calogero--Painlev\'e}\label{ssec:CP_systems_intro}
Every matrix $N\times N$ Painlev\'e system defined above corresponds to Calogero--Painlev\'e system. We briefly recall its construction following \cite{BCR17}.

Consider phase space of matrix Painlev\'e system, which is $M_{\alpha, t} = \{(p, q)\in \mathrm{Mat}_{N\times N}^{2}\left(\mathbb{C}\right)\}$. There is an action of $\mathrm{GL}_{N}\left(\mathbb{C}\right)$ by overall conjugation of $p$ and $q$. This action is Hamiltonian and the moment map equals $\mu_{N}(p, q) = [p,q]$. Let us define phase space of corresponding Calogero--Painlev\'e system as a Hamiltonian reduction
\begin{equation}\label{CP_phase_space_def}
\mathsf{M}_{\alpha, t} = M_{\alpha, t}//_{\mathbf{O}_{N,g}}\mathrm{GL}_{N}\left(\mathbb{C}\right).
\end{equation}
Here $\mathbf{O}_{N,g}$ is a coadjoint orbit $\mathbf{O}_{N,g} = \{\ri g(\mathbf{1}_{N\times N} - \xi \otimes \eta)| \xi \in \mathrm{Mat}_{N\times 1}\left(\mathbb{C}\right), \eta \in \mathrm{Mat}_{1\times N}\left(\mathbb{C}\right), \eta\xi = N\}$, $g\in \mathbb{C}$.

Let $H$ be the Hamiltonian of matrix $N\times N$ Painlev\'e system. The Hamiltonian $H$ is invariant with respect to the action of $\mathrm{GL}_{N}\left(\mathbb{C}\right)$, thus $H$ defines a Hamiltonian dynamics on $\mathsf{M}_{\alpha}$, which is called dynamics of the corresponding Calogero--Painlev\'e system.

The open dense subset of $\mathsf{M}_{\alpha, t}$ can be described as the phase space of system of particles considered up to permutations. This space is $\mathrm{T}^*\left(\left(\mathbb{C}^N\backslash\mathrm{diags}\right)/ \mathrm{S}_{N}\right)$, its points are sets $\{(p_j, q_j)\}_{j = 1,..., N}$ where $q_{j}$'s are distinct. Let us consider a map $\zeta_{N}: \mathrm{T}^*\left(\left(\mathbb{C}^N\backslash\mathrm{diags}\right)/ \mathrm{S}_{N}\right) \rightarrow M_{\alpha, t}$
\begin{align}\label{CP_section}
\zeta_N: \{(p_j, q_j)\}_{j = 1,..., N} \mapsto \Bigg(
\begin{pmatrix}
p_1& \frac{\mathrm{i}g}{q_1 {-} q_2}& \dots & \dots & \frac{\mathrm{i}g}{q_1 {-} q_N}\\
\frac{\mathrm{i}g}{q_2 {-} q_1} & p_2&  \frac{\mathrm{i}g}{q_2 {-} q_3}& \ddots & \frac{\mathrm{i}g}{q_2 {-} q_N}\\
\vdots & \frac{\mathrm{i}g}{q_3 {-} q_2} &  \ddots & \ddots & \vdots\\
\vdots & \ddots & \ddots & \ddots & \frac{\mathrm{i}g}{q_{N{-}1} {-} q_N}\\
\frac{\mathrm{i}g}{q_{N} {-} q_1}& \dots & \dots & \frac{\mathrm{i}g}{q_N {-} q_{N{-}1}} & p_N
\end{pmatrix},\;\;
\begin{pmatrix}
q_1& 0& \dots & \dots & 0\\
0& q_2& 0& \ddots & \vdots\\
\vdots & 0 &  \ddots & \ddots & \vdots\\
\vdots & \ddots & \ddots & \ddots & 0\\
0& \dots & \dots & 0 & q_N
\end{pmatrix}\Bigg).
\end{align}
Note that $\mathrm{Im}(\zeta_{N})\subset \mu_{N}^{-1}\left(\{\ri g\left(\mathbf{1}_{N\times N} - v_{N}\otimes v_{N}^t\right)\}\right)$, where $v_{N} = \begin{pmatrix} 1\\ \vdots \\ 1\end{pmatrix}$.

Let us denote by $[(p,q)]$ the orbit of $(p,q)$ with respect to the action of $\mathrm{GL}_{N}\left(\mathbb{C}\right)$. Then consider a map 
\begin{equation}
\begin{aligned}
\tilde{\zeta}_{N}:\mathrm{T}^*\left(\left(\mathbb{C}^N\backslash\mathrm{diags}\right)/ \mathrm{S}_{N}\right) &\rightarrow \mathsf{M}_{\alpha, t}\\
 \{(p_j, q_j)\}_{j = 1,..., N} &\mapsto [\zeta_{N}(\{(p_j, q_j)\})].
\end{aligned}
\end{equation}
 The map $\tilde{\zeta}_{N}$ is injective and the image of $\tilde{\zeta}_{N}$ consists of classes $[(p, q)]$ such that $q$ is diagonalisable with different eigenvalues. Let us use the notations $\mathrm{Im}(\tilde{\zeta}_{N}) = \mathsf{M}^{\mathrm{reg}}_{\alpha, t},\;\; M_{\alpha,t}^{\mathrm{reg}} = \{(p, q)\in \mathrm{Mat}_{N\times N}^{2}\left(\mathbb{C}\right): [(p,q)] \in\mathsf{M}^{\mathrm{reg}}_{\alpha, t}\}$. We have the inverse map
\begin{equation}
\begin{aligned}
\pi_{N}:	 M^{\mathrm{reg}}_{\alpha, t}\cap \mu_{N}^{-1}\left(\mathbf{O}_{N, g}\right) &\rightarrow \mathrm{T}^*\left(\left(\mathbb{C}^N\backslash\mathrm{diags}\right)/ \mathrm{S}_{N}\right)\\
	(p, q) &\rightarrow \{(p_j, q_j)\}_{j = 1,..., N}.
\end{aligned}
\end{equation}
Here $q_{j}$'s are eigenvalues of $q$ and $p_{j}$ is the $j$-th element on the diagonal of $p$ in a basis where $q = \mathrm{diag}\left(q_1, ..., q_N\right)$.

B\"acklund transformations of matrix Painlev\'e systems obtained in Sec. \ref{sec:Backlund} are rational in $p,q,t$ and thus do commute with the action of $\mathrm{GL}_{N}\left(\mathbb{C}\right)$. Also B\"acklund transformations preserve $[p,q]$. Thus we have
\begin{prop}
 B\"acklund transformation of matrix Painlev\'e system defines B\"acklund transformation for corresponding Calogero--Painlev\'e system.
\end{prop}
Let $w$ be a B\"acklund transformation of a matrix Painlev\'e system. The corresponding transformation of the Calogero--Painlev\'e system can be written as $\pi_N \circ w\circ \zeta_{N} = \mathsf{w}$.
\begin{remark}
Let us denote $w((p,q)) = (\tilde{p}, \tilde{q})$. Then 
\begin{equation}
[(p,q)]\in \mathsf{M}_{\alpha, t}^{\pi_N \circ w\circ \zeta_{N}} \Leftrightarrow	 \pi_N((p,q)) = \pi_N((\tilde{p}, \tilde{q})) \Leftrightarrow	 \exists S\in \mathrm{GL}_{N}\left(\mathbb{C}\right):\;\; (\tilde{p}, \tilde{q}) = (SpS^{-1}, SqS^{-1}).
\end{equation}
\end{remark}

Let us illustrate constructions above by an example (in addition to Example \ref{ex:intro_Calogero})

\begin{example}
To obtain Calogero--Painlev\'e  $\mathrm{III\big(D_6^{(1)}\big)}$ Hamiltonian one should restrict the corresponding matrix Painlev\'e $\mathrm{III\big(D_6^{(1)}\big)}$ Hamiltonian \eqref{MPIIID6} on the image of the map \eqref{CP_section}. In this way we obtain
\begin{equation}\label{Ham_CPIIID6}
t H_N(\{(p_i,q_i)\};t)= \sum_{i=1}^N (p_i^2 q_i^2+(-q_i^2+(\alpha_1+\beta_1)q_i+t)p_i-\alpha_1q_i)+ g^2\sum_{1\leq j<i\leq N} \frac{q_i^2+q_j^2}{(q_i-q_j)^2}.
\end{equation}

Note that this Hamiltonian is not of the physical form \eqref{Ham_CPII}, to obtain such form, one should make 
logarithmic change of variables
\begin{equation}
	\mathsf{t}=\log t, \qquad \mathsf{q}_i=\log q_i+\frac{\mathsf{t}}2, \quad 
	\mathsf{p}_i=p_iq_i-\frac{q_i}2+\frac{t}{2q_i}+\frac{\alpha_1+\beta_1}2,
\end{equation}
which gives
\begin{equation}\label{Ham_CPIIID6_phys}
	H_{phys}=\sum_{i=1}^N \left(\mathsf{p}_i^2-e^{\mathsf{t}} \sinh^2 \mathsf{q}_i+e^{\mathsf{t}/2} \left(\left(\beta_1{+}\frac12\right) \cosh \mathsf{q}_i-\left(\alpha_1{+}\frac12\right)\sinh \mathsf{q}_i \right) \right)+\!\!\sum_{1\leq j<i \leq N}\frac{g^2}{2\sinh^2 \left(\frac{\mathsf{q}_i-\mathsf{q}_j}2\right)}.
\end{equation}
So this is system of trigonometric Calogero type, in difference with rational Calogero--Painlev\'e $\mathrm{II}$ \eqref{Ham_CPII_2}, \eqref{Ham_CPII}.
However, Hamiltonian \eqref{Ham_CPIIID6} is
more convenient for us than \eqref{Ham_CPIIID6_phys} because it is given in terms of rational functions.
\end{example}

\subsection{Reduction at Calogero--Painlev\'e level}\label{ssec:redCP}
Let $G$ be a group of symmetries of a certain matrix Painlev\'e system and let $\mathsf{G}$ be the group of corresponding symmetries of the Calogero--Painlev\'e system. Then $\mathsf{M}_{\alpha}^{\mathsf{G}}$ is preserved by the Calogero--Painlev\'e dynamics. We aim to study the dynamics on this subset.

Consider a certain matrix $n|G|\times n|G|$ Painlev\'e system with the finite group $G$ of its symmetries from Table \ref{table:automorphisms} or from Table \ref{table:automorphisms_C2xC2}.
\begin{theorem}\label{thm:CP_red}
Let $G\cong C_2$ or $G\cong C_2 \times C_2$. Then there is an open subset $U \subset \left(\mathsf{M}^{\mathrm{reg}}_{\alpha}\right)^{\mathsf{G}}$ such that
\begin{itemize}
\item The dynamics on $U$ is equivalent to the dynamics of the $n$--particle Calogero--Painlev\'e system called by the Image in Tables \ref{table:automorphisms}, \ref{table:automorphisms_C2xC2}
with coupling constant $|G|g$.
\item $U$ is open and dense on the connected component of the largest dimension in $\left(\mathsf{M}^{\mathrm{reg}}_{\alpha}\right)^{\mathsf{G}}$.
\end{itemize}
\end{theorem}
In general $\left(\mathsf{M}^{\mathrm{reg}}_{\alpha}\right)^{\mathsf{G}}$ is not connected. We will see in the proof that the component of largest dimension in $\left(\mathsf{M}^{\mathrm{reg}}_{\alpha}\right)^{\mathsf{G}}$ is unique.

\begin{proof}
In the proof we combine two different Hamiltonian reductions. First, recall that above we defined reduction \eqref{CP_phase_space_def}, with the corresponding moment map $\mu_{N}:M^{\mathrm{reg}}_{\alpha}\rightarrow \mathfrak{gl}_{N}\left(\mathbb{C}\right)$ and the projection $\pi_N: \mu_{N}^{-1}\left(\mathbf{O}_{N,g}\right)\cap M_{\alpha, t}^{\mathrm{reg}}\rightarrow \mathsf{M}_{\alpha, t}^{\mathrm{reg}}$. Second, in the setting of Theorem \ref{thm:gen_constr} we have the corresponding moment map $\mathbf{m}:M_{\alpha}^{\bar{G}}\rightarrow \left(\mathfrak{gl}_{n}\left(\mathbb{C}\right)\right)^{|\bar{G}| - 1}$ and the projection $\mathrm{pr}: \mathbf{m}^{-1}\left(\mathbf{g}\right) \rightarrow \mathbb{M}_{\alpha}$. Recall that $\mathbf{m}$ is given just by the diagonal $n\times n$ blocks of $[p,q]$ from the second to the last one.

The main idea of the proof is to construct a map $\varphi$ such that the following diagram
\begin{figure}[h]
\begin{center}
\begin{tikzpicture}[scale=3]
\node (A) at (0,0.5) {$\mu_{|G|n}^{-1}\left(\mathbf{O}_{|G|n, g}\right)\cap \left(M_{\alpha}^{\mathrm{reg}}\right)^{\bar{G}}\cap \mathbf{m}^{-1}\left(\mathbf{g}\right)$};
\node (B) at (2,0.5) {$U$};
\node (C) at (0,0) {$\mu_{n}^{-1}\left(\mathbf{O}_{n, |G|g}\right)$};
\node (D) at (2,0) {$\mathsf{M}_{\beta}$};
\node [left=0 cm of A] {$M_{\alpha}^{\bar{G}} \supset$};
\node [right=0 cm of B] {$\subset \left(\mathsf{M}^{\mathrm{reg}}_{\alpha}\right)^{\mathsf{G}}$};
\node [left=0 cm of C] {$\mathbb{M}_{\alpha}\supset$};

\path[->,font=\scriptsize,>=angle 90]
(A) edge node[left]{$\mathrm{pr}$} (C)
(A) edge node[above]{$\pi_{|G|n}$} (B)
(C) edge node[below]{$\pi_{n}$} (D)
(B) edge [dashed] node [right] {$\varphi$} (D);
\end{tikzpicture}
\end{center}
\caption{Description of the map $\varphi$.}
\label{fig:CP_red_thm}
\end{figure}
is commutative.

Let us consider the cases $G \cong C_2$ only (for the cases $G \cong C_2\times C_2$ the proof is similar). Group $G\cong C_2$ generated by the transformation $w$. Let $\mathsf{w}$ be the corresponding transformation of the Calogero--Painlev\'e system.
\paragraph{Step 1.} Using explicit formulas for $w$ we see that $\mathsf{w}$ has special form.
\begin{equation}\label{Special_form_CP_Backlund}
\mathsf{w}: \{(p_j, q_j)\}_{j = 1,..., 2n}\mapsto \{(a_{\mathsf{w}}(q_j,t)p_j + b_{\mathsf{w}}(q_j,t), c_{\mathsf{w}}(q_j,t))\}_{j = 1,..., 2n}.
\end{equation}
Here $a_{\mathsf{w}}, b_{\mathsf{w}}, c_{\mathsf{w}}$ are certain rational functions. Invariance condition then can be written as follows.
\begin{equation}\label{CP_C2_fixed_point_eqns}
\exists \sigma\in S_{2n}: a_{\mathsf{w}}(q_i,t)p_i + b_{\mathsf{w}}(q_i,t) = p_{\sigma(i)},\;\;\;c_{\mathsf{w}}(q_i,t) = q_{\sigma(i)}.
\end{equation}
The permutation $\sigma$ is unique since $q_i$'s are pairwise distinct.

Only conjugacy class of $\sigma$ is well--defined since coordinates $((p_1,q_1), (p_2, q_2), ..., (p_{2n}, q_{2n}))$ are defined up to permutation. Hence for any point $x\in \left(\mathsf{M}_{\alpha}^{\mathrm{reg}}\right)^{\mathsf{w}}$ we have conjugacy class which we denote by $[\sigma_{x}]$. In notation $U_{[\sigma]} = \{x\in \left(\mathsf{M}_{\alpha}^{\mathrm{reg}}\right)^{\mathsf{w}} | [\sigma_{x}] = [\sigma]\}$, we get $\left(\mathsf{M}_{\alpha}^{\mathrm{reg}}\right)^{\mathsf{w}} = \sqcup_{[\sigma]}U_{[\sigma]}$.
\paragraph{Step 2.} Let us express $\sigma$ as the product of independent cycles. Recall that $\mathsf{w}$ is an involution, thus, taking square of $\mathsf{w}$, we get $q_{\sigma^2(i)} = q_{i}$, hence $\sigma^2 = \mathrm{Id}$. Then $\sigma$ is the product of independent transpositions. We denote the cyclic type of $\sigma$ by $[2^{k}1^{2n-2k}]$. Let us compute $\mathrm{dim}\left(U_{[2^{k}1^{2n-2k}]}\right)$. It is easy to see that each cycle in $\sigma$ imply system of equations \eqref{CP_C2_fixed_point_eqns} and systems for different cycles are independent.

Let $\sigma(i) = i$, then
\begin{equation}
a_{\mathsf{w}}(q_i,t)p_i + b_{\mathsf{w}}(q_i,t) = p_{i},\;\;\;c_{\mathsf{w}}(q_i,t) = q_{i}.
\end{equation}
Hence $(p_i(t), q_i(t))$ is an algebraic solution of the corresponding Painlev\'e equation.

If we have cycle $(l_1,l_2)$ in $\sigma$ then we get
\begin{subequations}
\begin{align}
a_{\mathsf{w}}(q_{l_1},t)p_{l_1} + b_{\mathsf{w}}(q_{l_1},t) &= p_{l_2},\;\;\;c_{\mathsf{w}}(q_{l_1},t) = q_{l_2},\label{C2_CP_inv_proof_transp_eq1}\\
a_{\mathsf{w}}(q_{l_2},t)p_{l_2} + b_{\mathsf{w}}(q_{l_2},t) &= p_{l_1},\;\;\;c_{\mathsf{w}}(q_{l_2},t) = q_{l_1}\label{C2_CP_inv_proof_transp_eq2}
\end{align}
\end{subequations}
Note that since $\mathsf{w}$ is an involution \eqref{C2_CP_inv_proof_transp_eq2} follows from  \eqref{C2_CP_inv_proof_transp_eq2}. Thus the cycle $(l_1,l_2)$ implies two independent equations.

As the result we get 
\begin{equation}
\mathrm{dim}\left(U_{[2^{k}1^{2n-2k}]}\right) = 4n - (2k + 2(2n - 2k)) = 2k.
\end{equation}
So $U:= U_{[2^{n}]}$ has the maximal dimension equal to $2n$.
\paragraph{Step 3a.} The map $\pi_{2n}$ should be surjective for the existence of $\varphi$. In other words we have to check $\forall x\in U: \pi_{2n}^{-1}(\{x\})\cap \left(M_{\alpha}^{\mathrm{reg}}\right)^{\bar{w}}\cap m_2^{-1}\left(\ri g \mathbf{1}_{n\times n}\right) \neq \varnothing$.

For $\sigma \in \mathrm{S}_{2n}$ let us denote by $S_{\sigma}$ the matrix corresponding to $\sigma$. Then for $x\in U_{[\sigma]}$ consider $(\breve{p},\breve{q}) = \zeta_{2n}(x)\in \pi_{2n}^{-1}(\{x\})$. Then we have
\begin{equation}
(S_{\sigma}\breve{p}S_{\sigma}^{-1}, S_{\sigma}\breve{q} S_{\sigma}^{-1}) = w(\breve{p}, \breve{q}).
\end{equation}
Since $x\in U_{[2^n]}$ we get  $\exists A\in \mathrm{GL}_{2n}\left(\mathbb{C}\right): AS_{\sigma}A^{-1} = S_{2}$. Recall that $S_2 = \mathrm{diag}\left(\mathbf{1}_{n\times n}, -\mathbf{1}_{n\times n}\right)$ was defined in Sec. \ref{ssec:gen_constr}. Then we get 
\begin{equation}\label{CP_C2_fixed_lift_Calogero_gauge}
(S_{2}A\breve{p}A^{-1}S_{2}^{-1}, S_{2}A\breve{q}A^{-1} S_{2}^{-1}) = w(A\breve{p}A^{-1}, A\breve{q}A^{-1})
\end{equation}
Therefore $\exists (\hat{p},\hat{q}) = (A\breve{p}A^{-1}, A\breve{q}A^{-1})\in \pi_{2n}^{-1}(\{x\}) \cap M_{\alpha, t}^{\bar{w}}$.

From Remark \ref{rem:moment_bl_diag} it follows that $[\hat{p}, \hat{q}]$ is block diagonal. Let us use the notation $[\hat{p}, \hat{q}] = \mathrm{diag}(m_1, m_2)$.

Since $(\hat{p}, \hat{q}) \in \mu_{2n}^{-1}\left(\mathbf{O}_{2n, g}\right)$ we get either $m_1 = \ri g \mathbf{1}_{n\times n}$ or $m_2 = \ri g \mathbf{1}_{n\times n}$. In the second case let us take $(p,q):=(\hat{p}, \hat{q}) \in m_{2}^{-1}\left(\ri g \mathbf{1}_{n\times n}\right)$. In the first case we take $(p, q) := (R\hat{p}R^{-1}, R\hat{q}R^{-1})$, where $ R = \begin{pmatrix}
0 & \mathbf{1}_{n\times n}\\
\mathbf{1}_{n\times n} & 0
\end{pmatrix}$ is a matrix, which preserves $\mu_{2n}^{-1}\left(\mathbf{O}_{2n, g}\right)\cap \left(M_{\alpha}^{\mathrm{reg}}\right)^{\bar{w}}$. In both cases we get 
\begin{equation}
(p,q) \in \pi_{2n}^{-1}(\{x\})\cap \left(M_{\alpha,t}^{\mathrm{reg}}\right)^{\bar{w}} \cap m_2^{-1}\left(\ri g\mathbf{1}_{n\times n}\right).
\end{equation}
\paragraph{Step 3b.} Let $(p,q) \in \pi_{2n}^{-1}(\{x\})\cap \left(M_{\alpha,t}^{\mathrm{reg}}\right)^{\bar{w}}\cap m_2^{-1}\left(\ri g \mathbf{1}_{n\times n}\right)$. Then we have to check that $\pi((p,q)) \in \mu_{n}^{-1}\left(\mathbf{O}_{n, |G|g}\right)$. 

We introduced coordinates $\tilde{P}, \tilde{Q}$ on $\mathbb{M}_{\alpha,t}$ on Steps 3 in Sec. \ref{ssec:linear_cases}, \ref{ssec:nonlinear_cases} and \ref{ssec:special_C2xC2}.
\begin{equation}
m_1 = [\tilde{P}, \tilde{Q}] - \ri g\mathbf{1}_{n\times n}.
\end{equation}
For the final coordinates $P, Q$ in each case one can check that 
\begin{equation}
[\tilde{P}, \tilde{Q}] = [P, Q].
\end{equation}
Since $[p,q] = \mathrm{diag}(m_1, \ri g\mathbf{1}_{n\times n})\in \mu_{2n}^{-1}\left(\mathbf{O}_{2n, g}\right)$ we get
\begin{equation}
[P, Q] - \ri g\mathbf{1}_{n\times n} = \ri g\left(\mathbf{1}_{n\times n} - \xi\otimes \eta\right),\;\; \eta\xi = 2n,
\end{equation}
which means $(P,Q) \in \mu_{n}^{-1}\left(\mathbf{O}_{n, 2g}\right)$. 
\paragraph{Step 3c.} Let us define $\varphi(x) = \pi_n \circ \mathrm{pr}\left(\left(p, q\right)\right)$. We have to check that the right side does not depend on the choice of $(p,q)\in \pi_{2n}^{-1}(\{x\})\cap \left(M_{\alpha,t}^{\mathrm{reg}}\right)^{\bar{w}} \cap m_2^{-1}\left(\ri g\mathbf{1}_{n\times n}\right)$.

Consider $(p', q')\in \pi_{2n}^{-1}(\{x\})\cap \left(M_{\alpha,t}^{\mathrm{reg}}\right)^{\bar{w}} \cap m_2^{-1}\left(\ri g\mathbf{1}_{n\times n}\right)$, then 
\begin{equation}
\exists B\in \mathrm{GL}_{2n}\left(\mathbb{C}\right) : (p', q') = (BpB^{-1}, BqB^{-1}).
\end{equation}
Acting on both sides by $\bar{w}$ we get
\begin{equation}
(p', q') = ((S_2BS_2^{-1})p(S_2BS_2^{-1})^{-1}, (S_2BS_2^{-1})q(S_2BS_2^{-1})^{-1}),
\end{equation}
which implies
\begin{equation}
[B^{-1}S_2BS_2^{-1}, p] = [B^{-1}S_2BS_2^{-1}, q] = 0.
\end{equation}
Since $(p,q)\in \mu_{2n}^{-1}\left(\mathbf{O}_{2n, g}\right)\cap M_{\alpha, t}^{\mathrm{reg}}$ we can consider these equalities in a gauge, where $q$ is diagonal and eigenvalues of $q$ will be different. Then we get
\begin{equation}
B^{-1}S_2BS_2^{-1} = \lambda\mathbf{1}_{2n\times 2n} \text{ for some } \lambda\in \mathbb{C}^{*}.
\end{equation}
Rewriting this as $B^{-1}S_2B = \lambda S_2$ and taking square of both sides we get $\lambda^2 = 1$. Then
\begin{equation}
B = 
\begin{cases}
\mathrm{diag}(b_1, b_2), \lambda = 1\\
R\mathrm{diag}(b_1, b_2), \lambda = -1.
\end{cases}
\end{equation}
Here $b_1, b_2\in \mathrm{GL}_{n}\left(\mathbb{C}\right)$. Note that conjugation $R$ does not preserve $\mu_{n}^{-1}\left(\mathbf{O}_{n, 2g}\right) \cap m_2^{-1}\left(\ri g\mathbf{1}_{n\times n}\right)$, while conjugation by $\mathrm{diag}(b_1, b_2)$ does. Then the only possible case is the $\lambda = 1$.

Conjugation by $\mathrm{diag}(\mathbf{1}_{n\times n}, b_2)$ preserves $\mathrm{pr}$ by the definition. Conjugation by $\mathrm{diag}(b_1, b_2)$ descends to the overall conjugation of $(P, Q)$ by $b_1$. But this conjugation preserves $\pi_n$ by the definition. Hence $\pi_n \circ \mathrm{pr} ((p,q)) = \pi_n \circ \mathrm{pr} ((p',q'))$.
\paragraph{Step 4.} One can inverse $\varphi$ on the $M_{\beta}^{\mathrm{reg}}$ taking $\varphi^{-1} = \pi_{2n}\circ s \circ \zeta_{n}$. Here $s:\mathbb{M}_{\alpha, t} \rightarrow M_{\alpha, t}^{\bar{w}}$ is a section for the reduction from Theorem \ref{thm:gen_constr} (for the example in non-linear cases see \eqref{Non-linear_cases_reduction_section}).
\paragraph{Step 5.}\label{thm:CP_red_step5}
It remains to check that $\varphi$ maps the dynamics on $U$ to the dynamics of the $n$--particle Calogero--Painlev\'e system called by Image in Table \ref{table:automorphisms}. 

Let $\gamma\subset U$ be an integral curve of $2n$--particle Calogero--Painlev\'e system. Then for $(x_0, t_0)\in \gamma$ let $(\tilde{x}_0, t_0)\in \mu_{2n}^{-1}\left(\mathbf{O}_{2n, g}\right)\cap \left(M_{\alpha}^{\mathrm{reg}}\right)^{\bar{w}}\cap \mathbf{m}^{-1}\left(\mathbf{g}\right)$ be a lift of $(x_0, t_0)$ i. e. $\pi_{2n}(\tilde{x}_0) = x_0$. Let us denote by $\Gamma$ the integral curve of the corresponding matrix $2n\times 2n$ Painlev\'e system through $(\tilde{x}_0, t_0)$. Note that $\mu_{2n}^{-1}\left(\mathbf{O}_{2n, g}\right)\cap \left(M_{\alpha}^{\mathrm{reg}}\right)^{\bar{w}}\cap \mathbf{m}^{-1}\left(\mathbf{g}\right)$ is locally preserved by the dynamics of $2n\times 2n$ matrix Painlev\'e system, so $\Gamma \subset \mu_{2n}^{-1}\left(\mathbf{O}_{2n, g}\right)\cap \left(M_{\alpha}^{\mathrm{reg}}\right)^{\bar{w}}\cap \mathbf{m}^{-1}\left(\mathbf{g}\right)$. Then $\Gamma$ is a lift of $\gamma$ to $\mu_{2n}^{-1}\left(\mathbf{O}_{2n, g}\right)\cap \left(M_{\alpha}^{\mathrm{reg}}\right)^{\bar{w}}\cap \mathbf{m}^{-1}\left(\mathbf{g}\right)$. By Theorem \ref{thm:gen_constr}, $\mathrm{pr}$ maps $\Gamma$ to an integral curve of the $n\times n$ matrix Painlev\'e system, called by Image in Table \ref{table:automorphisms}. Then by definition $\pi_{n}$ maps $\mathrm{pr}(\Gamma)$ to an integral curve of the corresponding $n$--particle Calogero--Painlev\'e system.
\end{proof}
\begin{remark}
In cases $4, 5$ from Table \ref{table:automorphisms} we have open subset 
\begin{equation}
U = \pi_{nd}\left(\mu_{|G|n}^{-1}\left(\mathbf{O}_{|G|n, g}\right)\cap \left(M_{\alpha}^{\mathrm{reg}}\right)^{\bar{G}}\cap \mathbf{m}^{-1}\left(\mathbf{g}\right)\right) \subset \mathsf{M}_{\alpha, t}^{\mathsf{w}}.
\end{equation}
For this subset one can also construct $\varphi$ as in Theorem \ref{thm:CP_red}. Namely surjectivity of $\pi_{nd}$ follows from definition of $U$, while steps $3b$--$5$ of the proof can be performed with the slight modification. Note that for cases $4,5$ B\"acklund transformations are not of the form \eqref{Special_form_CP_Backlund}, so steps $1,2$ of the proof have no sense, so for this $U$ we do not have the second statement of Theorem \ref{thm:CP_red}.
\end{remark}

\begin{example}\label{ex:CPIIID6_CPIIID8}
	Let us consider	Calogero--Painlev\'e $\mathrm{III\big(D_6^{(1)}\big)}$
	and corresponding matrix Painlev\'e $\mathrm{III\big(D_6^{(1)}\big)}$.
	Hamiltonian \eqref{Ham_CPIIID6} gives equations of motion 
	\begin{equation}\label{eqmot_CPIIID6}
	\begin{aligned}
		&t\dot{q}_i=2q_i^2 p_i-q_i^2+(\alpha_1+\beta_1)q_i+t,	\\ &t\dot{p}_i=-2p_i^2q_i+2p_iq_i-(\alpha_1+\beta_1)p_i+\alpha_1+
		2g^2\sum_{j=1, j\neq i}^N\frac{q_j(q_i+q_j)}{(q_i-q_j)^3}.
	\end{aligned}	 
    \end{equation}
	
	Let us at first illustrate Theorem \ref{thm:CP_red} for $G$ generated by $\pi\circ \pi'$ (case 2 from Table \ref{table:automorphisms}).
	We take $N=2n$ particles and consider open subset $U_{[2^n]}\subset\left(\mathsf{M}_{\alpha}^{\mathrm{reg}}\right)^{\mathsf{w}}$ with $\mathsf{w}=\pi_N \circ (\pi\circ \pi') \circ \zeta_N$,
	choosing $\sigma=(1,n+1) \ldots (n,2n)$.
	Then $\alpha_0{=}\alpha_1{=}\beta_0{=}\beta_1{=}1/2$ and on $U_{[2^n]}$ we have
	\begin{equation}\label{invsubset_CPIIID6_l} 
	q_{i+n}=-q_{i}, \quad p_{i+n}=1-p_i-t/q_i^2.
	\end{equation}
	Then equations of motion \eqref{eqmot_CPIIID6}
	reduce to equations on $\{(p_i,q_i)\}_{i=1,\ldots n}$
	\begin{equation}\label{eqmotinv_CPIIID6_l}
		t\dot{q}_i=2q_i^2 \tilde{p}_i+q_i, \qquad t\dot{\tilde{p}}_i=-2\tilde{p}_i^2q_i-\tilde{p}_i+\frac{q_i}2-\frac{t^2}{2q_i^3}+
		16g^2\sum_{j=1, j\neq i}^n\frac{q_iq_j^2(q_i^2+q_j^2)}{(q_i^2-q_j^2)^3}.
	\end{equation}
	Here for convenience we introduce $\tilde{p}_i=p_i-1/2+t/(2q_i^2)$,
	such that $\mathsf{w}^*(\tilde{p}_i)=-\tilde{p}_i$.
	Now we find coordinates in which this dynamics
	has Calogero--Painlev\'e type.
	In other words, we have to find map $\varphi$
	from Theorem \ref{thm:CP_red} using Hamiltonian reduction on the matrix level (with $g_2=g$). 
	We take $\{(p_i,q_i)\}_{i=1,\ldots 2n}\in U_{[2^n]}$,  namely under condition \eqref{invsubset_CPIIID6_l}. Choosing the following matrix $A$ from the proof of Theorem \ref{thm:CP_red} we have the following $(p,q)\in M_{\alpha}^{\bar{w}}$
	\begin{equation}
	A=\frac1{\sqrt2}	\begin{pmatrix}
		\mathbf{1}_{n\times n} & \mathbf{1}_{n\times n}\\
		-\mathbf{1}_{n\times n} & \mathbf{1}_{n\times n}
		\end{pmatrix},
	\end{equation}
    \begin{equation}
	\begin{aligned}
   \mathrm{Ad}_A\circ\zeta_{2n}(\{(p_i,q_i)\})
   =\Bigg(\begin{pmatrix}
  	\mathrm{diag}(\frac12-\frac{t}{2q_i^2}) & \mathfrak{p}_{12}\\
  	\mathfrak{p}_{21} & \mathrm{diag}(\frac12-\frac{t}{2q_i^2})
  \end{pmatrix}, \begin{pmatrix}
  0 & -\mathrm{diag} (q_i)\\
  -\mathrm{diag} (q_i) & 0
\end{pmatrix} \Bigg), \\
    (\mathfrak{p}_{12})_{ii}=-\tilde{p}_i+\frac{\ri g}{2q_i}, \,  (\mathfrak{p}_{21})_{ii}=-\tilde{p}_i-\frac{\ri g}{2q_i},\quad i\neq j:\, (\mathfrak{p}_{12})_{ij}=\frac{2\ri g q_j}{q_j^2-q_i^2}, \,
    (\mathfrak{p}_{21})_{ij}=\frac{2\ri g q_i}{q_j^2-q_i^2}.
    \end{aligned}
	\end{equation}
   After the conjugation by matrix $A$ the moment map condition becomes
   \begin{equation}
   [p,q]=\ri g\begin{pmatrix}
   \mathbf{1}_{n\times n} - 2 v_n v_n^t  & 0\\
   0 & \mathbf{1}_{n\times n} \\
   \end{pmatrix}\, \textrm{due to}\quad
A v_{2n}=\sqrt2 (\underbrace{1,\ldots 1}_n,\underbrace{0,\ldots 0}_n)^t.	
   \end{equation}
   So, using coordinates $(\tilde{P},\tilde{Q})=(\mathfrak{p}_{12}\mathfrak{q}_{12}^{-1},\mathfrak{q}_{12}\mathfrak{q}_{21})$ from Sec. \ref{ssec:linear_cases}
   we obtain on $ \mathrm{pr}\left(\mathrm{Ad}_A\circ\zeta_{2n}(U_{[2^n]})\right)$
   \begin{equation}\label{ex42_PQ_int}
   	\tilde{P}_{ii}=\tilde{p}_i/q_i-\frac{\ri g}{2q_i^2}, \quad i\neq j:\,
   	\tilde{P}_{ij}=
   	\frac{2\ri g }{q_i^2-q_j^2}, \qquad  \tilde{Q}_{ij}=\delta_{ij}q_i^2.
   	\end{equation} 
	In coordinates  \eqref{Proper_coordinates_case_3D6_to_3D8}
	we obtain following coordinates $\{(P_i,Q_i)\}_{i=1,\ldots n}$ and time $s$ on 
	$ \pi_n\circ\mathrm{pr}\left(\mathrm{Ad}_A\circ\zeta_{2n}(U_{[2^n]})\right)$
   \begin{equation}
   	P_i=4 \tilde{p}_i/q_i, \quad Q_i=q_i^2/4, \qquad s=t^2/16.
   \end{equation}
	Finally, one can check from \eqref{eqmotinv_CPIIID6_l} that these $\{(P_i,Q_i)\}_{i=1,\ldots n}$ satisfy Calogero--Painlev\'e $\mathrm{III\big(D_8^{(1)}\big)}$
	with coupling constant $2g$
	\begin{equation}
		s\frac{dQ_i}{ds}=2Q_i^2 P_i+Q_i, \qquad s\frac{dP_i}{ds}=-2P_i^2 Q_i-P_i+1-s/Q_i^2+
		8g^2\sum_{j=1, j\neq i}^n\frac{Q_j(Q_i+Q_j)}{(Q_i-Q_j)^3}.	
	\end{equation}
    \end{example}
     
	\begin{example}\label{ex:CPIIID6_CPV}
		Let us then illustrate Theorem \ref{thm:CP_red} also for Calogero--Painlev\'e $\mathrm{III\big(D_6^{(1)}\big)}$system, but for $G$ generated by  $\pi'$ (case 7 from Table \ref{table:automorphisms}).
	Note that this case is non-linear in difference with the previous example. We take $N=2n$ particles
	and consider $U_{[2^n]}\subset\left(\mathsf{M}_{\alpha}^{\mathrm{reg}}\right)^{\mathsf{w}}$ with $\mathsf{w}=\pi_N\circ \pi' \circ \zeta_N$.
	Then $\beta_0{=}\beta_1{=}1/2$ and on $U_{[2^n]}$ we have
	\begin{equation}\label{invsubset_CPIIID6_nl}
	q_{i+n}=\frac{t}{q_i}, \quad p_{i+n}=-\frac{q_i(p_iq_i+\alpha_1)}t.
	\end{equation} 
	Then equations of motion \eqref{eqmot_CPIIID6} reduce to equations on $\{(p_i,q_i)\}_{i=1,\ldots n}$
	\begin{equation}\label{eqmotinv_CPIIID6_nl}
	\begin{aligned}
		&t\dot{q}_i=2q_i^2 p_i-q_i^2+(\alpha_1+1/2)q_i+t	\\ &t\dot{p}_i=-2p_i^2q_i+2p_iq_i-\left(\alpha_1+\frac12\right)p_i+\alpha_1+
		2g^2 \frac{t q_i (q_i^2+t)}{(q_i^2-t)^3}
		+2g^2\sum_{j=1, j\neq i}^n\left(\frac{q_j(q_i+q_j)}{(q_i-q_j)^3}+\frac{tq^{-1}_j(q_i+tq_j^{-1})}{(q_i-tq_j^{-1})^3} \right).
	\end{aligned}
    \end{equation}
We want to find $\varphi$ from Theorem \ref{thm:CP_red}
using Hamiltonian reduction at matrix level  (with $g_2=g$). We choose $A$ as in the previous example, and take $\{(p_i,q_i)\}_{i=1,\ldots 2n}\in U_{[2^n]}$, namely
under condition~ \eqref{invsubset_CPIIID6_nl}.
Then we have the following $(p,q)\in M_{\alpha}^{\bar{w}}$
\begin{equation}
	\begin{aligned}
	 \mathrm{Ad}_A\circ\zeta_{2n}(\{(p_i,q_i)\})=	
	\Bigg(\begin{pmatrix}
		* & *\\
		\tilde{\mathfrak{p}}_{21}  & *
	\end{pmatrix},	
	\frac12\begin{pmatrix}
	\mathrm{diag} (q_i+t/q_i)	& \mathrm{diag} (-q_i+t/q_i)\\
	 \mathrm{diag} (-q_i+t/q_i)	& \mathrm{diag} (q_i+t/q_i)
	\end{pmatrix}\Bigg)
\\
(\tilde{\mathfrak{p}}_{21})_{ii}=-\frac{p_i}2-\frac{q_i(p_iq_i+\alpha_1)}{2t}+\frac{\ri g}2\frac{q_i+q_j}{t-q_iq_j},\quad i\neq j:\,
(\tilde{\mathfrak{p}}_{21})_{ij}=\frac{\ri g}2\left( \frac{q_iq_j+t}{t(q_j-q_i)}+\frac{q_i+q_j}{t-q_iq_j}\right),
\end{aligned}
\end{equation}
where we calculate only necessary matrix block for momentum. 
So, using formulas from Sec. \ref{ssec:nonlinear_cases}
we obtain following $(\tilde{P},\tilde{Q})$ on $ \mathrm{pr}\left(\mathrm{Ad}_A\circ\zeta_{2n}(U_{2^n})\right)$ (only diagonal elements of $\tilde{P}$ are calculated)
\begin{equation}
	\tilde{P}_{ii}=(2t \tilde{\mathfrak{q}}_{21}^{-1}
	\tilde{\mathfrak{p}}_{21} \tilde{\mathfrak{q}}_{11}^{-1})_{ii}=\frac4{q_i-tq_i^{-1}}\left(\left(p_i+\frac{\alpha_1}{q_i+tq_i^{-1}}\right)q_i+\frac{2t\ri g }{q_i^2-t^2q_i^{-2}}\right), \quad
	\tilde{Q}_{ij}=(\tilde{\mathfrak{q}}_{11})_{ij}=
	\frac12\delta_{ij}(q_i+tq_i^{-1}).
\end{equation}
In coordinates  \eqref{D6_D5_int} 
we obtain following coordinates $\{(P_i,Q_i)\}_{i=1,\ldots n}$ and time $s$ on 
$\pi_n\circ\mathrm{pr}\left(\mathrm{Ad}_A\circ\zeta_{2n}(U_{2^n})\right)$ 
\begin{equation}\label{ex_CPIIID6_CPV_PQ}
	P_i=8\frac{p_iq_i+\alpha_1/2}{q_i/\sqrt{t}-\sqrt{t}/q_i}+\frac{4\ri g \left(q_i/\sqrt{t}+\sqrt{t}/q_i\right)}{\left(q_i/\sqrt{t}-\sqrt{t}/q_i\right)^2}, \quad Q_i=\frac12+\frac14\left(\frac{q_i}{\sqrt{t}}+\frac{\sqrt{t}}{q_i}\right), \qquad s=-8\sqrt{t}, 
\end{equation}
Finally, these $\{(P_i,Q_i)\}_{i=1,\ldots n}$ satisfy Calogero--Painlev\'e  $\mathrm{V}(1-\alpha_1-\ri g,\ri g,\alpha_1-\ri g,\ri g)$ with coupling constant $2g$, which one can check from \eqref{eqmotinv_CPIIID6_nl}
\begin{equation}
\begin{aligned}
&s\frac{dQ_i}{ds}=(2P_i+s) Q_i (Q_i-1)+\ri g (1-2Q_i),\\
&s\frac{dP_i}{ds}=(P_i+s)P_i (1-2 Q_i)+\ri g (2 P_i+s)-\alpha_1 s+4g^2 \sum_{j=1, j\neq i}^n \frac{2Q_iQ_j+2Q_j^2-Q_i-3Q_j}{(Q_i-Q_j)^3}.	
\end{aligned}	
\end{equation}

\end{example}

\begin{example}\label{ex:CPII_CPII}
Let us now consider Calogero--Painlev\'e $\mathrm{II}$ system, as
in Example \ref{ex:intro_Calogero} from Introduction, but with an arbitrary number of particles. Then Hamiltonian \eqref{Ham_CPII_2} becomes
\begin{equation}
	H_N(\{(p_i,q_i)\};t) = \sum_{i=1}^N\left(\frac{1}{2}p_i^2 - \frac{1}{2}\left(q_i^2 + \frac{t}{2}\right)^2-(\alpha_1-1/2)q_i\right) + \sum_{j,i=1, j<i}^{N}\frac{g^2}{(q_i - q_j)^2}.
\end{equation}
Corresponding equations of motion are	
\begin{equation}\label{eqmot_CPII}
	\dot{q}_i=p_i, \qquad	\dot{p}_i=2q_i^3+tq_i+(\alpha_1-1/2)+2g^2 \sum_{j=1,j\neq i}^N \frac1{(q_i-q_j)^3}.
\end{equation}	

Let us illustrate Theorem \ref{thm:CP_red} for $G$ generated by $r$ (case 1 from Table \ref{table:automorphisms}). We take $N=2n$ particles and consider $U_{[2^n]}\subset\left(\mathsf{M}_{\alpha}^{\mathrm{reg}}\right)^{\mathsf{w}}$ with $\mathsf{w}=\pi_N \circ r \circ \zeta_N$. Then $\alpha_1=1/2$ and on $U_{[2^n]}$ we have
\begin{equation}\label{invsubset_CPII} 
q_{i+n}=-q_i,\, p_{i+n}=-p_i.
\end{equation}
Then equations of motion \eqref{eqmot_CPII} reduce to 
equations on $\{(p_i,q_i)\}_{i=1,\ldots n}$ 
\begin{equation}\label{eqmotinv_CPII}
	\dot{q}_i=p_i, \qquad	\dot{p}_i=2q_i^3+tq_i+
	\frac{g^2}{4q_i^3}+4g^2 \sum_{j=1, j\neq i}^n \frac{q_i(q_i^2+3 q_j^2)}{(q_i^2-q_j^2)^3}. 
\end{equation}	
We want to find $\varphi$ from Theorem \ref{thm:CP_red}
using matrix Hamiltonian reduction (with $g_2=g$).
The calculations just resemble those from Example \ref{ex:CPIIID6_CPIIID8}. As a result, we have \eqref{ex42_PQ_int} for $(\tilde{P}$, $\tilde{Q})$ with $p_i$ instead of $\tilde{p}_i$. Then, in coordinates \eqref{P2_proper_coordinates} we obtain momentum $\breve{P}$ in diagonal form
\begin{equation}\label{ex_CPII_CPII_PQ}
	\breve{Q}_i=-2^{-1/3}\tilde{P}=-2^{-1/3}\left(\frac{p_i}{q_i}-\frac{\ri g}{2q_i^2}\right), \quad \breve{P}_i=2^{1/3}\tilde{Q}=2^{1/3} q_i^2, \qquad s=-2^{1/3}t,
\end{equation}
where $\breve{Q}_i$ and $\breve{P}_i$
are diagonal elements of $\breve{Q}$ and $\breve{P}$
respectively. So the natural projection $\pi_N$ from $\mathrm{pr}(\mathrm{Ad}_A\circ\zeta_{2n}(U_{[2^n]}))$ gives us Calogero--Painlev\'e system in gauge, where $\breve{P}$ is diagonal. 
 
It follows from \eqref{eqmotinv_CPII} that
\begin{equation}
	\frac{d\breve{P}_i}{ds}=2\breve{P}_i\breve{Q}_i-\ri g, \qquad  \frac{d\breve{Q}_i}{ds}=
	\breve{P}_i-\breve{Q}_i^2-\frac{s}2+4 g^2 \sum_{j\neq i}\frac{\breve{P}_i+3\breve{P}_j}{(\breve{P}_i-\breve{P}_j)^3}.	
\end{equation}
This dynamics is Hamiltonian with 
\begin{equation}
	H_n(\{(\tilde{P}_i,\breve{Q}_i)\};t)=\sum_{i=1}^n (\frac{\breve{P}_i^2}2-(\breve{Q}_i^2+\frac{s}2)\breve{P}_i+\ri g \breve{Q}_i)-(2g)^2 \sum_{j\neq i} \frac{\breve{P}_i+\breve{P}_j}{(\breve{P}_i-\breve{P}_j)^2}.
\end{equation}
It is easy to see that this Hamiltonian can be obtained from the Matrix Painlev\'e Hamiltonian \eqref{Ham_MPII:f} in the gauge, where $\breve{P}$ is diagonal. So we get Calogero--Painlev\'e $\mathrm{II} (-\ri g-1/2)$ with coupling constant~$2g$. 

In order to obtain transformation of Calogero--Painlev\'e system to the standard gauge (as in r.h.s. of \eqref{CP_section}), one should diagonalize matrix $Q$. This cannot give an algebraic formula for general $n$.
\begin{remark}
Note that for $N=2$ and $g=0$ coordinate transformations from above Examples \ref{ex:CPIIID6_CPIIID8}, \ref{ex:CPIIID6_CPV}, \ref{ex:CPII_CPII} reproduce formulas \cite[5.10-5.12, 5.5-5.7, 9.5-9.7]{TOS05} 
for the folding transformations of the corresponding Painlev\'e equations. 
\end{remark}

\end{example}

\section{Further examples}
\label{sec:further}

\subsection{Matrix reduction generalization}
\label{ssec:matr_gen}
For the construction in Sec.
\ref{sec:block_reductions} we choose for a certain B\"acklund transformation the additional twist $S_d = \mathrm{diag}\left(\mathbf{1}_{n\times n}, e^{\frac{2\pi \ri }{d}}\mathbf{1}_{n\times n}, ..., e^{\frac{2\pi (d-1) }{d}}\mathbf{1}_{n\times n}\right)$.
At Calogero level this twist corresponds to a permutation class $[d^n]$.  We also choose the moment map value 
$\mathbf{g} = (\ri g_2\mathbf{1}_{n\times n}, ..., \ri g_{d}\mathbf{1}_{n\times n})$.
It is natural to try to weaken such restrictions.
Below we present two examples for such generalizations.

\begin{example}\label{ex:3n_II}
	Let us consider $3n\times 3n$ matrix Painlev\'e $\mathrm{II}$. We take $M_{\alpha}^{\bar{r}}$ with  $\bar{r}=\mathrm{Ad}_{S} \circ r$,
	where $S=\mathrm{diag}\left(\mathbf{1}_{n\times n},- \mathbf{1}_{2n\times 2n}\right)$. Recall that $r: (p,q)\mapsto (-p,-q)$ and $\theta\mapsto -\theta$, so we set $\theta=0$.
	
	\noindent $\mathbf{Step \, 1.}$
	The matrices $(p,q)\in M_{\alpha}^{\bar{r}}$ are given by
	\begin{equation}
		p=
		\begin{pmatrix}
		 0	& \mathfrak{p}_{12} & \mathfrak{p}_{13} \\
		 \mathfrak{p}_{21} & 0 & 0 \\
		 \mathfrak{p}_{31} & 0 & 0
		\end{pmatrix}, \qquad
	q=
	\begin{pmatrix}
		0	& \mathfrak{q}_{12} & \mathfrak{q}_{13} \\
		\mathfrak{q}_{21} & 0 & 0 \\
		\mathfrak{q}_{31} & 0 & 0
	\end{pmatrix}
    \end{equation}
    with a symplectic form $\mathrm{Tr}(\rd \mathfrak{p}_{12}\wedge \rd \mathfrak{q}_{21}) + \mathrm{Tr}(\rd \mathfrak{p}_{21}\wedge \rd \mathfrak{q}_{12})+\mathrm{Tr}(\rd \mathfrak{p}_{13}\wedge \rd \mathfrak{q}_{31}) + \mathrm{Tr}(\rd \mathfrak{p}_{31}\wedge \rd \mathfrak{q}_{13})$.
    
    \noindent $\mathbf{Step \, 2.}$
    The remaining gauge freedom is $\mathrm{diag} (\mathrm{GL}_{n}(\mathbb{C})),\mathrm{GL}_{2n}(\mathbb{C}))$ and the moment map is
    \begin{equation}\label{ex_3n_moment}
    [p,q]=\begin{pmatrix}
     (m_1)_{n\times n}	& 0\\
     0  & (m_2)_{2n\times 2n}
    \end{pmatrix}=	
    \begin{pmatrix}
    \mathfrak{p}_{12} \mathfrak{q}_{21} {+} \mathfrak{p}_{13} \mathfrak{q}_{31} {-} \mathfrak{q}_{12} \mathfrak{p}_{21} {-} \mathfrak{q}_{13} \mathfrak{p}_{31}	& 0 & 0  \\
    0	& \mathfrak{p}_{21} \mathfrak{q}_{12}{-} \mathfrak{q}_{21} \mathfrak{p}_{12}  & \mathfrak{p}_{21} \mathfrak{q}_{13} {-} \mathfrak{q}_{21} \mathfrak{p}_{13}  \\
    0	& \mathfrak{p}_{31} \mathfrak{q}_{12} {-} \mathfrak{q}_{31} \mathfrak{p}_{12}  & \mathfrak{p}_{31} \mathfrak{q}_{13} {-} \mathfrak{q}_{31} \mathfrak{p}_{13}
    \end{pmatrix}	
	\end{equation}
    
    \noindent $\mathbf{Step \, 3.}$
    We perform the Hamiltonian reduction with respect to
	$\mathrm{GL}_{2n}(\mathbb{C})=\{\mathrm{diag}(\mathbf{1}_{n\times n},h_2)|h_2\in \mathrm{GL}_{2n}(\mathbb{C}) \}$.
	We take the moment map value $m_2=\mathrm{diag}\left(\ri g_2\mathbf{1}_{n\times n},\ri g_3\mathbf{1}_{n\times n}\right)$, where $g_2\neq g_3$. Note that its stabilizer is $\mathrm{GL}_n^2(\mathbb{C})=\{\mathrm{diag}(\mathbf{1}_{n\times n},h_2,h_3)|h_2, h_3\in \mathrm{GL}_n(\mathbb{C}) \}$.
	On the reduction $\mathbb{M}_{\alpha}$ we can choose the following Darboux coordinates 
	\begin{equation}
		(\tilde{P},\tilde{Q})=(\mathfrak{p}_{12}\mathfrak{q}_{12}^{-1},(1-g_3/g_2)\mathfrak{q}_{12}\mathfrak{q}_{21}).
	\end{equation}
	
	\noindent $\mathbf{Step \, 4.}$
	After such Hamiltonian reduction we obtain $n\times n$ matrix system with 
	Hamiltonian
	\begin{equation}
		H(\tilde{P},\tilde{Q};t)+3t^2/8=\mathrm{Tr}((\ri g_2-\ri g_3) \tilde{P} - t\tilde{Q}  - \tilde{Q}^2 + \tilde{P} \tilde{Q} \tilde{P}).
	\end{equation}
	Omitting $3t^2/8$ and making standard substitution
	\begin{equation}
		\tilde{P}=2^{1/3}Q, \quad \tilde{Q}=2^{-1/3}(P+Q^2)-t/2, \qquad s= -2^{1/3} t
	\end{equation}
	we obtain standard Hamiltonian of $\mathrm{PII}(-\ri (g_2-g_3)-1/2)$.
\end{example}

The block sizes and the moment map value 
in above example, as well as in Sec. \ref{sec:block_reductions} are quite special.
At least after the Hamiltonian reduction the phase space dimension should correspond to a matrix system, namely it should be equal to a doubled square of an integer. For the standard situation from Sec. \ref{sec:block_reductions} we have
\begin{equation}
	\underbrace{2(dn)^2}_{\mathrm{dim} M_{\alpha, t}}\rightarrow 
	\underbrace{2dn^2}_{\mathrm{dim} M_{\alpha, t}^{\bar{w}}}=\underbrace{2n^2}_{\mathrm{dim}\mathbb{M}_{\alpha,t}}+\underbrace{(d-1)n^2}_{\mathrm{dim}\mathrm{Mat}^{d-1}_n(\mathbb{C})}+\underbrace{(d-1)n^2}_{\mathrm{dim}\mathrm{GL}^{d-1}_n(\mathbb{C})},
\end{equation}
where the last two terms correspond to the value of moment map and its stabilizer correspondingly.
For the situation from the Example \ref{ex:3n_II} analogous calculation gives
\begin{equation}
	\underbrace{2(3n)^2}_{\mathrm{dim} M_{\alpha,t}}\rightarrow 
	\underbrace{8n^2}_{\mathrm{dim} M_{\alpha,t}^{\bar{w}}}=\underbrace{2n^2}_{\mathrm{dim}\mathbb{M}_{\alpha,t}}+\underbrace{4n^2}_{\mathrm{dim}\mathrm{Mat}_{2n}(\mathbb{C})}+\underbrace{2n^2}_{\mathrm{dim}\mathrm{GL}^2_n(\mathbb{C})}.
\end{equation}
However, even if the dimension is not a doubled square of integer, sometimes we can make certain additional reduction to obtain matrix Painlev\'e. We illustrate this by the following example.
\begin{example}
Let us again consider a Hamiltonian reduction of $M^{\bar{r}}_{\alpha}$ for $3n\times 3n$ matrix Painlev\'e $\mathrm{II}$, but, in difference with Example \ref{ex:3n_II}, with respect to $\mathrm{GL_n}(\mathbb{C})=\mathrm{diag}(h_1,\mathbf{1}_{2n\times 2n}|h_1\in \mathrm{GL_n}(\mathbb{C}))$.
We take for the moment map (see \eqref{ex_3n_moment}) the value
\begin{equation}
	\mathfrak{p}_{12} \mathfrak{q}_{21} + \mathfrak{p}_{13} \mathfrak{q}_{31} - \mathfrak{q}_{12} \mathfrak{p}_{21} - \mathfrak{q}_{13} \mathfrak{p}_{31}=\ri g_1 \mathbf{1}_{n\times n},
\end{equation}
so after the reduction we obtain $6n^2$-dimensional system
\begin{equation}
	\underbrace{2(3n)^2}_{\mathrm{dim} M_{\alpha,t}}\rightarrow 
	\underbrace{8n^2}_{\mathrm{dim} M_{\alpha,t}^{\bar{w}}}=\underbrace{6n^2}_{\mathrm{dim}\mathbb{M}_{\alpha,t}}+\underbrace{n^2}_{\mathrm{dim}\mathrm{Mat}_{n}(\mathbb{C})}+\underbrace{n^2}_{\mathrm{dim}\mathrm{GL}_n(\mathbb{C})}.
\end{equation}
On the reduction we introduce Darboux coordinates
\begin{equation}
	\begin{pmatrix}
		\tilde{P}_1 & \tilde{P}_2 & \tilde{P}_3\\ \tilde{Q}_1 & \tilde{Q}_2 & \tilde{Q}_3
	\end{pmatrix}=
        \begin{pmatrix}
		\mathfrak{q}_{21}\mathfrak{p}_{13}-\mathfrak{p}_{21}\mathfrak{q}_{13} & (\mathfrak{p}_{31}-\mathfrak{q}_{31}\mathfrak{q}_{21}^{-1}\mathfrak{p}_{21})\mathfrak{q}_{21}^{-1} &\mathfrak{q}_{21}\mathfrak{q}_{12}-\frac12(\mathfrak{p}_{21}\mathfrak{q}_{21}^{-1})^2+\mathfrak{q}_{21}\mathfrak{q}_{13}\mathfrak{q}_{31}\mathfrak{q}_{21}^{-1}+\frac{t}2\\ \mathfrak{q}_{31}\mathfrak{q}_{21}^{-1} & \mathfrak{q}_{21}\mathfrak{q}_{13} & -\mathfrak{p}_{21}\mathfrak{q}_{21}^{-1}
		\end{pmatrix}.
\end{equation}	

This matrix system has the Hamiltonian
\begin{equation}\label{Ham_3M}
	H(\{\tilde{P}_i,\tilde{Q}_i\};t)+3t^2/8=\mathrm{Tr}\left(-\tilde{P}_3^2+\frac14 (\tilde{Q}_3^2-t)^2-(\ri g+1/2) \tilde{Q}_3+
	\tilde{P}_1\tilde{P}_2-2\tilde{Q}_2\tilde{P}_2\tilde{Q}_3\right).
\end{equation}

On the coordinates $(\tilde{P}_3,\tilde{Q}_3)$ we have almost matrix Painlev\'e $\mathrm{II}$ Hamiltonian. 
It appears that we can perform two successive Hamiltonian reductions to obtain matrix system only on $(\tilde{P}_3,\tilde{Q}_3)$.
For the first one, with respect to translations of $\tilde{Q}_1$, we fix moment map value $\tilde{P}_1=0$.
Then the Hamiltonian becomes invariant with respect to $(\tilde{P}_2,\tilde{Q}_2)\mapsto(h\tilde{P}_2,\tilde{Q}_2h^{-1}), h\in \mathrm{GL}_n(\mathbb{C})$.
We fix the value of corresponding moment map $\tilde{P}_2\tilde{Q}_2$ to be a scalar matrix, namely $\tilde{P}_2\tilde{Q}_2=\theta\mathbf{1}_{n\times n}=\tilde{Q}_2\tilde{P}_2$. Then after   coordinate and time rescaling
\begin{equation}
	P=2^{1/3}\tilde{P},\quad Q=2^{-1/3}\tilde{Q},\qquad s= -2^{1/3}t,
\end{equation}
we see that Hamiltonian \eqref{Ham_3M} becomes the standard Hamiltonian of matrix $n\times n$ $\mathrm{PII}\big(-\ri g-\frac12-2\theta \big)$. 
\end{example}

\subsection{Adding algebraic solutions to Calogero--Painlev\'e}\label{ssec:algebrCP}
Let us consider the Calogero--Painlev\'e system together with order $2$ B\"acklund transformation $\mathsf{w}$.
Let us suppose that $\mathsf{w}$
leads to the permutations of cyclic type $[2^n 1^m]$, in difference
with Step 2 in the proof of Theorem \ref{thm:CP_red}.
Then we obtain that $\mathsf{w}((p_i,q_i))=(p_i,q_i)$ for $2n<i\leq 2n+m$, so the last $m$ particles dynamics evolve as certain algebraic functions.  
Below we give two examples 
for the dynamics on the Calogero--Painlev\'e invariant subset involving such additional algebraic solutions.  
\begin{example}
	We modify Example \ref{ex:CPII_CPII} for Calogero--Painlev\'e $\mathrm{II}$.
	The additional algebraic solution is
	\begin{equation}
  \mathsf{w}((p_i,q_i))=(-p_i,-q_i)=(p_i,q_i) \Rightarrow (p_i,q_i)= (0,0).
   \end{equation}
  We can add only one such particle,
  due to condition $q_i\neq q_j$ for $i\neq j$ on $\mathsf{M}^{\mathrm{reg}}_{\alpha}$.
    Adding such particle $q_{2n+1}=p_{2n+1}=0$,
	for the rest of the particles we have equations of motion on $(p_i,q_i)$, $1\leq i\leq n$
	\begin{equation}
		\dot{q}_i=p_i, \qquad	\dot{p}_i=2q_i^3+tq_i+
		\frac{9g^2}{4q_i^3}+4g^2 \sum_{j=1, j\neq i}^n \frac{q_i(q_i^2+3 q_j^2)}{(q_i^2-q_j^2)^3}. 
	\end{equation}	
which differ from
\eqref{eqmotinv_CPII}
only in coefficient of term $q_i^{-3}$.
Then it is easy to see that we can modify
Example \ref{ex:CPII_CPII} by hands.
Namely, it is enough to modify $Q_i$ from \eqref{ex_CPII_CPII_PQ} by $g\rightarrow 3g$
\begin{equation}
	\breve{Q}_i=-2^{-1/3}\left(\frac{p_i}{q_i}-\frac{3\ri g}{2q_i^2}\right).
\end{equation}
Finally we obtain Calogero--Painlev\'e~$\mathrm{II}(-3\ri g-1/2)$ instead of $(-\ri g-1/2)$. 
	
\end{example}

\begin{example}
	We modify Example \ref{ex:CPIIID6_CPV} for Calogero--Painlev\'e $\mathrm{III}\big(D_6^{(1)}\big)$.
    The additional algebraic solutions are
    \begin{equation}
    	\mathsf{w}((p_i,q_i))=(-t^{-1}q_i(p_iq_i+\alpha_1),t/q_i)=(p_i,q_i) \Rightarrow (p_i,q_i)=\left(\mp\frac{\alpha_1}{2\sqrt{t}},\pm \sqrt{t}\right) 
    \end{equation}
	We can add one of them or both to our system.
	So let us add to the Calogero--Painlev\'e $\mathrm{III}\big(D_6^{(1)}\big)$ from Example \ref{ex:CPIIID6_CPV} $n_1=0,1$ particles $\sqrt{t}$ and $n_2=0,1$ particles $(-\sqrt{t})$.
	Modifying momenta $P_i$ from \eqref{ex_CPIIID6_CPV_PQ}
    by an additional term
    \begin{equation}
    	P_i=8\frac{p_iq_i+\alpha_1/2}{q_i/\sqrt{t}-\sqrt{t}/q_i}+\frac{4(1+2n_1n_2)\ri g \left(q_i/\sqrt{t}+\sqrt{t}/q_i+4(n_1-n_2)\right)}{\left(q_i/\sqrt{t}-\sqrt{t}/q_i\right)^2}
    \end{equation}
    we obtain Calogero--Painlev\'e $\mathrm{V}(1-\alpha_1-(1+2n_1n_2)\ri g,(1-2(n_1-n_2)+2n_1n_2)\ri g,\alpha_1-(1+2n_1n_2)\ri g,(1-2(n_2-n_1)+2n_1n_2)\ri g)$.
	
\end{example}

\begin{remark}
It would be interesting to obtain such Calogero--Painlev\'e relations from matrix Painlev\'e ones. Note that permutation matrix for $[2^n1^m]$ is conjugated to $S=\mathrm{diag}(\mathbf{1}_{(m+n)\times(m+n)},-\mathbf{1}_{n\times n})$, cf. Sec. \ref{ssec:matr_gen}.
\end{remark}

\subsection{Spin Calogero--Painlev\'e systems}\label{ssec:spin}
One can consider a more general case of Calogero--type systems taking in \eqref{CP_phase_space_def} general orbit $\mathbf{O}$ instead of $\mathbf{O}_{N, g}$. Let us denote the corresponding reduction map by $\pi_{\mathbf{O}}$ instead of $\pi_{N}$. 

The regular part of the phase space of spin Calogero--Painlev\'e system can be defined as
\begin{equation}\label{Spin_CM_phase_space}
\mathsf{M}_{\alpha, t}^{\mathrm{reg}}= M_{\alpha, t}^{\mathrm{reg}}//_{\mathbf{O}}\mathrm{GL}_{N}\left(\mathbb{C}\right).
\end{equation}
It will be more convenient for us to use the identification (for the details see \cite{Reshetikhin02}, Theorem 3).
\begin{equation}\label{Spin_CM_coordinates}
M_{\alpha, t}^{\mathrm{reg}}//_{\mathbf{O}}\mathrm{GL}_{N}\left(\mathbb{C}\right) = \left(\mathrm{T}^*\left(\mathbb{C}^N\backslash\mathrm{diags}\right)\times \mathbf{O}//_{0}\mathrm{GL}_{1}^{N}\left(\mathbb{C}\right)\right)/ \mathrm{S}_{N}.
\end{equation}
Let us recall the construction of the reduction $\mathbf{O}//_{0}\mathrm{GL}_{1}^{N}\left(\mathbb{C}\right)$ in the right hand side of \eqref{Spin_CM_coordinates}.

Let $X \in \mathfrak{gl}_{N}\left(\mathbb{C}\right)$, then the fundamental vector field, corresponding to the coadjoint action of $X$ on $\mathbf{O}$ is $v_{X}(m) = -\mathrm{ad}^*_{X}(m)$. Then Kirillov-Kostant-Souriau form can be written as $\omega_{\mathrm{KKS}}|_{m}(v_{X}, v_{Y}) = \left< m, [X, Y]\right>$. Then we have $\iota_{v_{X}}\omega_{\mathrm{KKS}}|_{m}(v_{Y}) = -\left< \mathrm{ad}^*_{Y}(m), X\right> = \left< \mathrm{d}m(v_{Y}), X\right>$, which means that the coadjoint action on $\mathbf{O}$ is Hamiltonian with the moment map $\mu_{KKS}: m \mapsto m$. Below we will identify points of $\mathbf{O} \subset \mathfrak{gl}_{N}\left(\mathbb{C}\right)^*$ with matrices using Killing form. 

Next, we restrict coadjoint action to the subgroup of diagonal matrices, so the moment map becomes the projection, which maps $(m_{ij})_{i,j=1,...,N}\mapsto (m_{ii})_{i = 1,..., N}$. Then the element of $\mathbf{O}//_{0}\mathrm{GL}_{1}^{N}\left(\mathbb{C}\right)$ is a matrix $m \in \mathbf{O}$ such that  $m_{ii} = 0\; \forall\; i = 1,..., N$, considered up to conjugations by diagonal matrices $(m_{ij})_{i,j=1,...,N} \sim (a_{i}a_{j}^{-1}m_{ij})_{i,j=1,...,N}$. We will denote this class by $[m]$. It will be useful to denote the coadjoint orbit of $m$ by $\mathbf{O}\left(m\right)$.

 Points of $\mathrm{T}^*\left(\mathbb{C}^N\backslash\mathrm{diags}\right)$ are ordered sets of pairs $((p_j, q_j))_{j = 1,..., N}$, where $q_j$'s are distinct. The action of $S_N$ on $\left(\mathrm{T}^*\left(\mathbb{C}^N\backslash\mathrm{diags}\right)\times \mathbf{O}//_{0}\mathrm{GL}_{1}^{N}\left(\mathbb{C}\right)\right)/ \mathrm{S}_{N}$ is defined as follows
\begin{equation}
\sigma: (((p_j, q_j))_{j = 1,..., N}, [m]) \mapsto (((p_{\sigma(j)}, q_{\sigma(j)}))_{j = 1,..., N}, [S_{\sigma}mS_{\sigma}^{-1}]),
\end{equation} 
where $S_{\sigma}$ is the matrix corresponding to $\sigma\in S_{n}$.

Finally the identification \ref{Spin_CM_coordinates} can be done as follows (cf. $\tilde{\zeta}_{N}$ in Sec. \ref{ssec:CP_systems_intro})
 \begin{equation}
 [(((p_j, q_j))_{j = 1,..., N}, [m])] \mapsto \Bigg[\Bigg(
\begin{pmatrix}
p_1& \frac{m_{12}}{q_1 {-} q_2}& \dots & \dots & \frac{m_{1N}}{q_1 {-} q_N}\\
\frac{m_{21}}{q_2 {-} q_1} & p_2&  \frac{m_{23}}{q_2 {-} q_3}& \ddots & \frac{m_{2N}}{q_2 {-} q_N}\\
\vdots & \frac{m_{32}}{q_3 {-} q_2} &  \ddots & \ddots & \vdots\\
\vdots & \ddots & \ddots & \ddots & \frac{m_{N{-}1,N}}{q_{N{-}1} {-} q_N}\\
\frac{m_{N1}}{q_{N} {-} q_1}& \dots & \dots & \frac{m_{N,N{-}1}}{q_N {-} q_{N{-}1}} & p_N
\end{pmatrix},\;\;
\begin{pmatrix}
q_1& 0& \dots & \dots & 0\\
0& q_2& 0& \ddots & \vdots\\
\vdots & 0 &  \ddots & \ddots & \vdots\\
\vdots & \ddots & \ddots & \ddots & 0\\
0& \dots & \dots & 0 & q_N
\end{pmatrix}\Bigg)\Bigg].
 \end{equation}
 The dynamics of a spin Calogero--Painlev\'e system is defined as descent of the matrix Painlev\'e dynamics with respect to the reduction \eqref{Spin_CM_phase_space}. Let us consider an analogue of Theorem \ref{thm:CP_red} for the spin Calogero--Painlev\'e $\mathrm{III\big(D_6^{(1)}\big)}$.
\begin{example}
The spin Calogero--Painlev\'e $\mathrm{III\big(D_6^{(1)}\big)}$ is defined by the Hamiltonian (cf. \eqref{Ham_CPIIID6}).
	\begin{equation}
	t H([(((p_j, q_j))_{j = 1,..., N}, [m])], t) = \sum_{i=1}^N (p_i^2 q_i^2+(-q_i^2+(\alpha_1+\beta_1)q_i+t)p_i-\alpha_1q_i) - \sum_{1\leq i<j \leq N}\frac{m_{ij}m_{ji}(q_i^2 + q_j^2)}{(q_i - q_j)^2}.
	\end{equation}
	We will consider a spin generalization of Example \ref{ex:CPIIID6_CPIIID8}. So, we take $N = 2n$ and $\alpha_1 = \beta_1 = \frac{1}{2}$. 
	
From the diagram \ref{fig:CP_red_thm} it can be seen that the natural candidate for $U$ from Theorem \ref{thm:CP_red} is $\pi_{\mathbf{O}}\left(\tilde{U}\right)$, where $\tilde{U} = \left(\left(M_{\alpha, t}^{\overline{\pi\circ \pi'}}\right)^{\mathrm{reg}}\cap \mu_{2n}^{-1}\left(\mathbf{O}\right) \cap \mathbf{m}^{-1}\left(\ri g_2 \mathbf{1}_{n\times n}\right)\right)$.  We want to obtain the coordinates on $\pi_{\mathbf{O}}\left(\tilde{U}\right)$ in which the dynamics corresponds to spin Calogero--Painlev\'e $\mathrm{III\big(D_8^{(1)}\big)}$ i. e. we aim to find an analogue of the map $\varphi$ from the diagram \ref{fig:CP_red_thm}.

\paragraph{Step 1.} Let us obtain coordinates on $\pi_{\mathbf{O}}\left(\tilde{U}\right)$. Consider $[(((p_j, q_j))_{j = 1,..., N}, [m])] = \pi_{\mathbf{O}}((p,q))$, where $(p,q) \in \tilde{U}$, then we get
$
[p, q] = \begin{pmatrix}
m_1 & 0\\
0 & \ri g_2 \mathbf{1}_{n\times n}
\end{pmatrix}.
$
Let $A$ be a matrix such that $\tilde{q}:= A^{-1}q A = \mathrm{diag}\left(q_1, ..., q_{2n}\right)$. Then we get 
\begin{equation}\label{ex:spin_CP_inv}
\begin{aligned}
\tilde{S}_2\tilde{q}\tilde{S}_{2}^{-1} &= -\tilde{q}\\
\tilde{S}_2\tilde{p}\tilde{S}_{2}^{-1} &= -\tilde{p} + 1 - t\tilde{q}^{-2}
\end{aligned}
\end{equation}
where $\tilde{S}_{2} = A^{-1}S_2 A,\;\; \tilde{p} = A^{-1} p A$.
Then without loss of generality we get 
\begin{equation}\label{ex:spin_CP_inv_q}
q_{i+n} = -q_{i}\textit{ where } 1\leq i \leq n.
\end{equation}
It can be seen that after additional multiplication of $A$ by a certain diagonal matrix from the right (which preserves $\tilde{q}$) we get $\tilde{S}_{2} = \begin{pmatrix}
0 & \mathbf{1}_{n\times n}\\
\mathbf{1}_{n\times n} & 0
\end{pmatrix}$.
So, for diagonal entries of $\tilde{p}$ from \eqref{ex:spin_CP_inv} we get
\begin{equation}\label{ex:spin_CP_inv_p}
p_{i+n} = 1 - p_{i} - \frac{t}{q_{i}^2}\textit{ where } 1\leq i \leq n.
\end{equation}

Now let us obtain conditions on spin variables $m_{ij}$'s. Note that since $A^{-1}S_2A = \tilde{S}_2$, we have
$
A^{-1} = \frac{1}{\sqrt{2}}\begin{pmatrix}
 \mathbf{1}_{n\times n} &  -\mathbf{1}_{n\times n}\\
 \mathbf{1}_{n\times n} &  \mathbf{1}_{n\times n}
\end{pmatrix}
\begin{pmatrix}
\alpha_{1} & 0\\
0 & \alpha_{2}
\end{pmatrix},
$ where $\alpha_1, \alpha_2 \in \mathrm{GL}_{n}\left(\mathbb{C}\right)$. Then we get
\begin{equation}\label{ex:spin_CP_inv_m}
m = [\tilde{p}, \tilde{q}] = A^{-1}[p, q]A = \begin{pmatrix}
\frac{1}{2}M & \frac{1}{2}M-\ri g_2 \mathbf{1}_{n\times n}\\
\frac{1}{2}M-\ri g_2 \mathbf{1}_{n\times n} & \frac{1}{2}M
\end{pmatrix},
\end{equation}
where $M = \alpha_{1}m_1\alpha_{1}^{-1} + \ri g_2 \mathbf{1}_{n\times n}$. Then we have $M\in \tilde{\mathbf{O}}$, where $\tilde{\mathbf{O}} = \mathbf{O}\left(m_1\right) + \ri g_2 \mathbf{1}_{n\times n}$. Note that $M$ has zero diagonal entries and is defined up to conjugations by diagonal matrices, which means that we have $[M] \in \tilde{\mathbf{O}}//_{0}\mathrm{GL}_{n}\left(\mathbb{C}\right)$.

Taking into account remaining symmetry corresponding to permutations of $(p_{i}, q_{i})$'s we get that $\pi_{\mathbf{O}}\left(\tilde{U}\right)$ is parametrized by points $[((P_{i}, Q_{i})_{i = 1,..., n}, [M])]\in \left(\tilde{\mathbf{O}}//_{0}\mathrm{GL}_{1}^{n}\left(\mathbb{C}\right)\times \mathrm{T}^*\left(\left(\mathbb{C}^n\backslash\mathrm{diags}\right)\right)\right)/ \mathrm{S}_{n}$, where
\begin{equation}\label{ex:spin_CP_proper_coord_PQ}
	Q_i = \frac{1}{4}q_i^2,\;\;P_i = \frac{4}{q_i}\left(p_{i} - \frac{1}{2} + \frac{t}{2q_i^2}\right),\;\; 1\leq i \leq n
\end{equation}
We will explain meaning of formulas \eqref{ex:spin_CP_proper_coord_PQ} below.

\paragraph{Step 2.} Let us explain why $[((P_{i}, Q_{i})_{i = 1,..., n}, [M])]$ are desired coordinates corresponding to the dynamics of spin Calogero--Painlev\'e $\mathrm{III\big(D_8^{(1)}\big)}$. Following the arguments similar to Step 5 from the proof of Theorem \ref{thm:CP_red} we get that the coordinates given by the map $\pi_{\mathbf{O}}\circ \mathrm{pr}$ match the dynamics of spin Calogero--Painlev\'e $\mathrm{III\big(D_8^{(1)}\big)}$. 

We have $\mathrm{pr}((p,q)) = (P,Q)$, where $(P,Q)$ is given by \eqref{Proper_coordinates_case_3D6_to_3D8}. One can check that the pair $(P, Q)$ is conjugated to $\left(\hat{P}, \hat{Q}\right) : =\left(2\left(\mathfrak{p}_{12}\mathfrak{q}_{12}^{-1} + \mathfrak{q}_{21}^{-1}\mathfrak{p}_{21}\right), \frac{1}{4}\mathfrak{q}_{12}\mathfrak{q}_{21}\right)$, where $\mathfrak{p}_{ij}, \mathfrak{q}_{ij}$'s are given by \eqref{linear_case_solutions} properly specialized for case \ref{sssec:IIID6_IIID8}.

To obtain $\pi_{\mathbf{O}}((P, Q)) = \pi_{\mathbf{O}}\left(\left(\hat{P}, \hat{Q}\right)\right)$ following the identification \eqref{Spin_CM_coordinates} one should consider the pair $\left(\hat{P}, \hat{Q}\right)$ in the basis, where $\hat{Q} = \frac{1}{4}\mathfrak{q}_{12}\mathfrak{q}_{21}$ is diagonal. Since $A^{-1}qA = \tilde{q}$ it is easy to see that
\begin{equation*}
\alpha_1 \mathfrak{q}_{12} \alpha_2^{-1} = \alpha_2  \mathfrak{q}_{21} \alpha_1^{-1} = -\mathrm{diag}\left(q_{1}, ..., q_{n}\right).
\end{equation*}
This proves that the formula $Q_{i} = \frac{1}{4}q_{i}^2$ gives part of proper coordinates. As well it implies that $\alpha_{1}$ diagonalizes $\frac{1}{4}\mathfrak{q}_{12}\mathfrak{q}_{21}$, so momenta conjugated to $Q_{i}$'s are diagonal entries of $\mathrm{Ad}_{\alpha_1}\left(\hat{P}\right)$. At first let us compute
\begin{equation*}
p_{i} = (A^{-1}pA)_{ii} = -\frac{1}{2}\left(\alpha_1 \mathfrak{p}_{12}\alpha_{2}^{-1} + \alpha_2 \mathfrak{p}_{21}\alpha_{1}^{-1}\right)_{ii} + \frac{1}{2} - \frac{t}{2q_{i}^2},\;\; 1\leq i \leq n.
\end{equation*}
Then we obtain
\begin{multline}
\left(\mathrm{Ad}_{\alpha_1}\left(\hat{P}\right)\right)_{ii} = \left(2\left(\alpha_1\mathfrak{p}_{12}\alpha_{2}^{-1}\alpha_{2}\mathfrak{q}_{12}^{-1}\alpha_1^{-1} + \alpha_{1}\mathfrak{q}_{21}^{-1}\alpha_{2}^{-1}\alpha_{2}\mathfrak{p}_{21}^{-1}\alpha_{1}^{-1}\right)\right)_{ii} =\\= -\frac{2}{q_{i}}\left(\alpha_1 \mathfrak{p}_{12}\alpha_{2}^{-1} + \alpha_2 \mathfrak{p}_{21}\alpha_{1}^{-1}\right)_{ii} = \frac{4}{q_{i}}\left(p_{i} - \frac{1}{2} + \frac{t}{2q_{i}^2}\right).
\end{multline}
So, we have proved that the formula \eqref{ex:spin_CP_proper_coord_PQ} gives part of coordinates corresponding to the dynamics of spin Calogero--Painlev\'e $\mathrm{III\big(D_8^{(1)}\big)}$. The rest of coordinates is given by $[\mathrm{Ad}_{\alpha_1}\left([\hat{P}, \hat{Q}]\right)]$. By straightforward computation we know that $[\hat{P}, \hat{Q}] = m_1 + \ri g_2 \mathbf{1}_{n\times n}$, so
\begin{equation}
\mathrm{Ad}_{\alpha_1}\left([\hat{P}, \hat{Q}]\right) = M.
\end{equation}
Additionally we can check that $[((P_{j},Q_{j}))_{j=1,...,n},[M]]$ are proper coordinates at the level of dynamics. Since coordinates $[M]$ on $\tilde{U}$ are tricky it is more convenient to use the analogue of $\varphi^{-1}$ which is defined as follows
\begin{multline}
\tilde{\varphi}^{-1}:\left(\tilde{\mathbf{O}}//_{0}\mathrm{GL}_{1}^{n}\left(\mathbb{C}\right)\times \mathrm{T}^*\left(\left(\mathbb{C}^n\backslash\mathrm{diags}\right)\right)\right)/ \mathrm{S}_{n} \rightarrow \tilde{U}\\
\begin{bmatrix}
(((P_{j},Q_{j}))_{j=1,...,n},\\
[M])
\end{bmatrix} \mapsto 
\begin{bmatrix}
\Big((\big(\frac{\sqrt{Q_j}P_j + 1}{2} + \frac{t}{8Q_{j}}, 2\sqrt{Q_j}\big))_{j=1,...,n}\cup (\big(\frac{{-}\sqrt{Q_j}P_j + 1}{2} + \frac{t}{8Q_{j}}, {-}2\sqrt{Q_j}\big))_{j=1,...,n},\\
[\begin{pmatrix}
\frac{1}{2}M & \frac{1}{2}M - \ri g_2 \mathbf{1}_{n\times n}\\
\frac{1}{2}M - \ri g_2 \mathbf{1}_{n\times n} & \frac{1}{2}M
\end{pmatrix}]\Big)
\end{bmatrix}.
\end{multline}
One can check that this map is symplectic. Let us compute the Hamiltonian with respect to variables $[(((P_{j},Q_{j}))_{j=1,...,n}, [M])]$. Taking $s = \frac{t^2}{16}$ as the new time variable we get
	\begin{multline}
	s H([((P_j, Q_j))_{j = 1,..., n}, [M]], s) = \sum_{i = 1}^{n}\left(P_{i}^2Q_{i}^2 + P_{i}Q_{i} - Q_{i} - sQ_{i}^{-1}\right)-\\-\sum_{1\leq i<j \leq n}M_{ij}M_{ji}\frac{2Q_{i}Q_{j}}{(Q_i - Q_j)^2} - \frac{1}{2}\sum_{1\leq i<j \leq n}M_{ij}M_{ji} + \frac{ng^2}{4}.
	\end{multline}
Note that on $\tilde{\mathbf{O}}$, the term $- \frac{1}{2}\sum_{1\leq i<j \leq n}M_{ij}M_{ji} + \frac{ng^2}{4}$ is just a constant since $\sum_{1\leq i<j \leq n}M_{ij}M_{ji}$ is a Casimir function. So we get the system with the Hamiltonian
\begin{equation}
s H([((P_j, Q_j))_{j = 1,..., n}, [M]], s) = \sum_{i = 1}^{n}\left(P_{i}^2Q_{i}^2 + P_{i}Q_{i} - Q_{i} - sQ_{i}^{-1}\right) - \sum_{1\leq i< j \leq n}M_{ij}M_{ji}\frac{2Q_{i}Q_{j}}{(Q_i - Q_j)^2},
\end{equation}
which is the Hamiltonian of spin Calogero--Painlev\'e $\mathrm{III\big(D_8^{(1)}\big)}$.
\end{example}

\bibliographystyle{alpha}
\bibliography{bibtex}

\begin{thebibliography}{KNY17}

\bibitem[AS21]{Adler:2021Matrix}
V.~E. Adler and V.~V. Sokolov.
\newblock Matrix {P}ainlev\'{e} {II} equations.
\newblock {\em Teoret. Mat. Fiz.}, 207(2):188--201, 2021.
\newblock [{\href{https://arxiv.org/abs/2012.05639}
  {\texttt{arXiv:2012.05639}}]}.

\bibitem[BCR18]{BCR17}
Marco Bertola, Mattia Cafasso, and Vladimir Roubtsov.
\newblock {Noncommutative Painlev\'e Equations and Systems of Calogero Type}.
\newblock {\em Commun. Math. Phys.}, 363(2):503--530, 2018.
\newblock [\href{http://arxiv.org/abs/1710.00736}{\texttt{arXiv:1710.00736}}].

\bibitem[HOMS]{NCAlgebra}
{J. William} Helton, {Mauricio C. de} {Oliveira}, Bob Miller, and Mark Stankus.
\newblock {NCAlgebra} package.
\newblock
  [\href{https://mathweb.ucsd.edu/~ncalg/}{\texttt{https://mathweb.ucsd.edu/~ncalg/}}].

\bibitem[IM85]{Inozemtsev:1985Extension}
V.~I. Inozemtsev and D.~V. Meshcheryakov.
\newblock Extension of the class of integrable dynamical systems connected with
  semisimple {L}ie algebras.
\newblock {\em Lett. Math. Phys.}, 9(1):13--18, 1985.

\bibitem[Ino89]{Inozemtsev:1989Lax}
V.~I. Inozemtsev.
\newblock Lax representation with spectral parameter on a torus for integrable
  particle systems.
\newblock {\em Lett. Math. Phys.}, 17(1):11--17, 1989.

\bibitem[Kaw15]{Kawakami:2015Matrix}
Hiroshi Kawakami.
\newblock Matrix {P}ainlev\'{e} systems.
\newblock {\em J. Math. Phys.}, 56(3):033503, 27, 2015.

\bibitem[KKS78]{KKS:1978}
D.~Kazhdan, B.~Kostant, and S.~Sternberg.
\newblock Hamiltonian group actions and dynamical systems of {C}alogero type.
\newblock {\em Comm. Pure Appl. Math.}, 31(4):481--507, 1978.

\bibitem[KNY17]{KNY15}
Kenji Kajiwara, Masatoshi Noumi, and Yasuhiko Yamada.
\newblock Geometric aspects of {P}ainlev\'{e} equations.
\newblock {\em Journal of Physics A: Mathematical and Theoretical},
  50(7):073001, Jan 2017.
\newblock [{\href{https://arxiv.org/abs/1509.08186}
  {\texttt{arXiv:1509.08186}}]}.

\bibitem[Res02]{Reshetikhin02}
Nicolai Reshetikhin.
\newblock Degenerate {I}ntegrability of {S}pin {C}alogero-{M}oser {S}ystems and
  the duality with the spin {R}uijsenaars systems.
\newblock {\em Letters in Mathematical Physics}, 63, 03 2002.

\bibitem[Rum15]{R13}
Igor Rumanov.
\newblock {Classical integrability for beta-ensembles and general Fokker-Planck
  equations}.
\newblock {\em J. Math. Phys.}, 56(1):013508, 2015.
\newblock [\href{http://arxiv.org/abs/1306.2117}{\texttt{arXiv:1306.2117}}].

\bibitem[Rum16]{R14}
Igor Rumanov.
\newblock {Painlev\'e Representation of Tracy\textendash{}Widom${_\beta}$
  Distribution for ${\beta}$ = 6}.
\newblock {\em Commun. Math. Phys.}, 342(3):843--868, 2016.
\newblock [\href{http://arxiv.org/abs/1408.3779}{\texttt{arXiv:1408.3779}}].

\bibitem[Tak01]{Takasaki:2001Painleve}
Kanehisa Takasaki.
\newblock Painlev\'{e}-{C}alogero correspondence revisited.
\newblock {\em J. Math. Phys.}, 42(3):1443--1473, 2001.
\newblock
  [\href{http://arxiv.org/abs/math/0004118}{\texttt{arXiv:math/0004118}}].

\bibitem[TOS05]{TOS05}
Teruhisa Tsuda, Kazuo Okamoto, and Hidetaka Sakai.
\newblock Folding transformations of the {P}ainlev\'{e} equations.
\newblock {\em Math. Ann.}, 331(4):713--738, 2005.

\end{thebibliography}

\end{document}